\documentclass[12pt, a4paper]{article}
\pdfoutput=1
\usepackage{jheppub}
\usepackage{float}
\usepackage{bigints}

\usepackage{subcaption}
\usepackage[T1]{fontenc}
\usepackage{ae}
\usepackage{aecompl}
\usepackage{float}
\usepackage{latexsym}
\usepackage{amsmath,amstext,amssymb,amsfonts,
amscd,bm,array,multirow,amsbsy,mathrsfs,amsthm,calc}
\usepackage{t1enc}
\usepackage{indentfirst}
\usepackage{pb-diagram}
\usepackage{graphicx}
\usepackage{color}
\usepackage{hyperref}
\usepackage[usenames,dvipsnames]{xcolor}
\usepackage{url}
\usepackage[shortlabels]{enumitem}

\theoremstyle{plain}
\newtheorem{thm}{Theorem}[section]

\newtheorem{prop}[thm]{Proposition}

\theoremstyle{definition}
\newtheorem{definition}[thm]{Definition}

\theoremstyle{remark}
\newtheorem{remark}[thm]{Remark}

\newcommand{\be}{\begin{equation*}}
\newcommand{\ee}{\end{equation*}}
\newcommand{\ben}{\begin{equation}}
\newcommand{\een}{\end{equation}}
\newcommand{\beqa}{\begin{eqnarray*}}
\newcommand{\eeqa}{\end{eqnarray*}}
\newcommand{\beqan}{\begin{eqnarray}}
\newcommand{\eeqan}{\end{eqnarray}}
\newcommand{\nn}{\nonumber}

\def\i{\mathbf{i}}
\def \const{\mathrm{const}}

\def\Z{\mathbb{Z}}
\def\C{\mathbb{C}}
\def\R{\mathbb{R}}

\def\cK{\mathcal{K}}

\def\Aut{\mathrm{Aut}}
\def\Stab{\mathrm{Stab}}
\def\Mat{\mathrm{Mat}}
\def\Ad{\mathrm{Ad}}
\def\bAd{\overline{\Ad}}
\def\Iso{\mathrm{Iso}}

\newcommand{\pd}{\partial}
\def\dd{\mathrm{d}}

\newcommand{\id}{\mathrm{id}}
\newcommand{\Tr}{\mathrm{Tr}}
\newcommand{\tr}{\mathrm{tr}}

\newcommand{\diag}{\mathrm{diag}}
\newcommand{\sign}{\mathrm{sign}}

\def\cC{\mathcal{C}}

\def\cG{\mathcal{G}}
\def\cH{\mathcal{H}}
\def\cL{\mathcal{L}}

\def\cN{\mathcal{N}}
\def\cP{\mathcal{P}}
\def\cQ{\mathcal{Q}}
\def\cS{\mathcal{S}}
\def\cV{\mathcal{V}}

\def\cX{\mathcal{X}}
\def\cR{\mathcal{R}}
\def\cT{\mathcal{T}}

\def\bLambda{\boldsymbol{\Lambda}}

\def\mD{\mathbb{D}}
\def\mA{\mathbb{A}}

\def\rD{\mathrm{D}}

\def\rS{\mathrm{S}}
\def\rf{\mathrm{f}}

\def\rJ{\mathrm{J}}

\def\SO{\mathrm{SO}}

\def\U{\mathrm{U}}

\def\PSU{\mathrm{PSU}}
\def\SU{\mathrm{SU}}
\def\su{\mathrm{su}}
\def\SL{\mathrm{SL}}

\def\br{\mathbf{r}}

\def\Diff{\mathrm{Diff}}
\def\bpsi{\bar{\psi}}
\def\su{\mathrm{su}}
\def\sl{\mathrm{sl}}
\def\can{\mathrm{can}}

\newcommand{\eqdef}{\stackrel{{\rm def.}}{=}}

\DeclareMathOperator{\arccosh}{arccosh}
\DeclareMathOperator{\arcsinh}{arcsinh}
\DeclareMathOperator{\arctanh}{arctanh}

\def\grad{\mathrm{grad}}
\def\Hess{\mathrm{Hess}}

\def\End{\mathrm{End}}

\def\Re{\mathrm{Re}}
\def\Im{\mathrm{Im}}

\def\ttau{{\tilde \tau}}
\def\tsigma{{\tilde \sigma}}
\def\trho{{\tilde \rho}}
\def\ttheta{{\tilde \theta}}

\newcommand{\twopartdef}[4]
{
	\left\{
	\begin{array}{ll}
		#1 & \mbox{if } #2 \\
		#3 & \mbox{if } #4
	\end{array}
	\right.
}

\title{Hidden symmetries of two-field cosmological models}

\author{Lilia Anguelova${}^1$, Elena Mirela Babalic${}^2$ and Calin Iuliu Lazaroiu${}^{3,2}$}

\affiliation{${}^1$ Institute for Nuclear Research and Nuclear
  Energy, BAS, Sofia, Bulgaria\\ ${}^2$ Horia Hulubei National
  Institute for Physics and Nuclear Engineering,\\
  Bucharest-Magurele 077125, Romania\\ ${}^3$ Center for Geometry and
  Physics, Institute for Basic Science, Pohang 37673, \\Republic of
  Korea\\}

\emailAdd{${}^1$anguelova@inrne.bas.bg, ${}^2$mbabalic@theory.nipne.ro, ${}^3$calin@ibs.re.kr}

\abstract{We determine the most general time-independent Noether
  symmetries of two-field cosmological models with
  rotationally-invariant scalar manifold metrics. In particular, we
  show that such models can have hidden symmetries, which arise if and
  only if the scalar manifold metric has Gaussian curvature $-3/8$,
  i.e. when the model is of elementary $\alpha$-attractor type with a
  fixed value of the parameter $\alpha$. In this case, we find
  explicitly all scalar potentials compatible with hidden Noether
  symmetries, thus classifying all models of this type.  We also
  discuss some implications of the corresponding conserved quantity.}

\begin{document}

\maketitle

\pagebreak

\section{Introduction}

Cosmological models with two real scalar fields produce natural
generalizations of single scalar field cosmology, which may be favored
in fundamental theories of gravity \cite{OOSV,GK,OPSV,AP}. Pairs of real
scalar fields arise naturally in supergravity and string
theory. Indeed, the generic supersymmetric compactification of such
theories produces complex moduli, each of which can be viewed as a
pair of real scalars.

Scalar multifield models were considered in cosmology from various
points of view (see, for example, \cite{PT1,PT2,Welling,m1,m2,m3,m4}),
including numerically \cite{Dias1,Dias2,Mulryne} and to some extent
phenomenologically \cite{c1}. However, current insight into their
dynamics (which, as outlined in \cite{flows,Roest1,Roest2}, is quite
interesting from the perspective of the geometric theory of dynamical
systems \cite{Palis}) is far more limited than for one-field
models. In particular, numerous questions regarding the behavior of
such models have not been studied systematically.

In this paper, we address one fundamental problem regarding such
models, namely the question of symmetries and conserved quantities.  A
general two-field cosmological model based on a simply connected and
spatially-flat FLRW space-time with scale factor $a(t)$ is
parameterized by a {\em scalar manifold} $(\Sigma,\cG)$ (a complete
and connected two-dimensional Riemannian manifold which plays the role
of target space for the non-linear sigma model describing the two
scalar fields $\varphi^1,\varphi^2$) and a {\em scalar potential} $V$
(a smooth real-valued function defined on $\Sigma$, which describes
the self-interacting potential of the scalars). Here $\Sigma$ is the
internal manifold of the two scalar fields, which we allow to have
non-trivial topology (such as non-contractible cycles) while $\cG$ is
the scalar manifold metric, whose components locally determine the
kinetic terms of the two scalar fields. The scalar kinetic terms are
specified {\em globally} by the pair $(\Sigma,\cG)$.  Such models
admit a description (known traditionally as the ``minisuperspace
formulation'') as a Lagrangian system for the three degrees of freedom
$a$, $\varphi^1$, $\varphi^2$ subject to a constraint
(namely the zero energy shell condition), which implements the
Friedmann equation. The minisuperspace formulation allows one to study
symmetries of such models using the Noether approach. We followed this
method in reference \cite{ABL} --- where we restricted to scalar
manifolds which are rotationally invariant and we considered only a
special class of time-independent symmetries satisfying a certain
separation of variables Ansatz. In the present paper, we still assume
that the scalar manifold metric $\cG$ is rotationally
invariant\footnote{It is in fact possible to remove this assumption as
  well, as we will show in a separate paper by developing a more
  general theory.}, but study the problem without making any other
assumptions. This allows us to give a complete description of all
time-independent Noether symmetries for models with
rotationally-invariant scalar manifold metric and to classify
Noether-symmetric models of this type. In particular, we find many new
Noether symmetries that were never considered before, the vast
majority of which do not satisfy the separation of variables Ansatz
used in reference \cite{ABL}. Most of these symmetries, which we
describe explicitly, are not invariant under the $\U(1)$ group of
rotations which preserves the scalar manifold metric.

Using the minisuperspace description, we first show that any
time-independent Noether symmetry of a general two-field cosmological
model is a direct sum of a {\em visible} symmetry with a {\em Hessian}
symmetry. The first of these are obvious symmetries, corresponding to
those continuous isometries of the scalar manifold which preserve the
scalar potential $V$; such symmetries do not involve the cosmological
scale factor $a$. On the other hand, Hessian symmetries are {\em
  hidden}, in the sense that they are not immediately obvious; they
involve $a$ as well as an auxiliary real-valued smooth function
$\Lambda$ (called the {\em generating function} of the corresponding
Hessian symmetry), which is defined on the scalar manifold $\Sigma$.
For Hessian symmetries, the Noether condition dictates that $\Lambda$
must satisfy a system of two linear partial differential equations,
namely the so-called {\em Hesse equation} (a second order PDE for
$\Lambda$ which depends only on the scalar manifold metric) and the
{\em $\Lambda$-$V$ equation} (a relation which involves the scalar
manifold metric and the first order partial derivatives of $\Lambda$
and $V$). The Hesse equation relates the Hessian tensor of $\Lambda$
to the symmetric tensor $\Lambda \cG$ and a solution of this equation
will be called a {\em Hesse function}. The two-field cosmological
model is called {\em weakly Hessian} if its scalar manifold admits
non-trivial Hesse functions.  The model is called {\em Hessian} if it
admits a Hessian symmetry, i.e. if it is weakly Hessian and its scalar
potential $V$ satisfies the $\Lambda$-$V$ equation for some
non-trivial Hesse function $\Lambda$ of the scalar manifold. A Hessian
model can also admit visible symmetries, provided that its scalar
potential is sufficiently special. We show that the conservation law
associated to a Hessian symmetry allows one to compute the e-fold
number along cosmological trajectories without performing an integral,
thereby providing a useful tool for extracting phenomenologically
relevant information in Hessian two-field models.

When $\cG$ is rotationally invariant\footnote{Notice that the Hesse
  symmetry generating function $\Lambda$ is {\em not} assumed to be
  rotationally symmetric, even though we assume that the scalar
  manifold metric
  is.}, we show that the model is weakly Hessian iff $\Sigma$ is 
diffeomorphic to a disk, a punctured disk or an annulus and $\cG$ is a
metric of constant Gaussian curvature $K=-\frac{3}{8}$. In particular,
weakly-Hessian two-field cosmological models coincide with those
elementary two-field $\alpha$-attractor models (in the sense of
reference \cite{elem}) for which $\alpha=\frac{8}{9}$ (here and in 
the following we use the convention $K = - \frac{1}{3 \alpha}$), being special
examples of the much wider class of two-field generalized
$\alpha$-attractors introduced and studied in
\cite{genalpha,elem,modular, cosmunif} \footnote{Generalized two-field
  $\alpha$-attractors extend ordinary two-field $\alpha$-attractor
  models \cite{KL1,KL2,KL3,KL4,KL5,KL6,KL7,AKLV} (whose target space is the
  hyperbolic disk) to models whose target space is allowed to be an
  arbitrary complete hyperbolic surface. As explained in
  \cite{genalpha, elem, modular, cosmunif}, such models can be
  approached through uniformization theory, which relates them to the
  framework of modular inflation developed and studied in
  \cite{S1,S2,S3,S4,S5}.}. We show that such weakly-Hessian models are
Hessian iff their scalar potential has a specific form which we
determine explicitly in all cases, thereby classifying all Hessian
two-field models with rotationally invariant scalar manifold
metric. When the scalar manifold is a disk endowed with its complete
metric of Gaussian curvature $-3/8$, we find that Hessian models fall
naturally into three classes (which we call ``timelike'',
``spacelike'' and ``lightlike''), distinguished by the orbit type of
their Hesse generating function under the natural action of the group
of orientation-preserving isometries of the hyperbolic disk. In this
case, the space of Hesse functions is three-dimensional and can be
identified with the Minkowski space $\R^{1,2}$. Thus Hesse functions
are naturally parameterized by a Minkowski 3-vector $B$, whose
timelike, spacelike or lightlike character determines the type of the
model. For each of the three types, we give the explicit general form
of the scalar potential which is compatible with a Hessian symmetry,
as well as the explicit form of the corresponding Hesse
generator. When the scalar manifold is a punctured disk or an annulus
endowed with a metric of Gaussian curvature $-3/8$, we find that the
space of Hesse functions is one-dimensional and determine it, also
giving the explicit form of the scalar potential for which the
two-field model is Hessian. In each case mentioned above, we also
describe the special situation when the Hessian model admits visible
symmetries. These results allow for a complete description of
time-independent Noether symmetries in two-field cosmological models
with rotationally invariant scalar manifold metric. As we will show in
a separate paper, a deeper mathematical approach allows one to prove that
these are in fact the {\em only} Hessian models, hence the results
derived herein provide a complete classification of Hessian
cosmological two-field models.

The paper is organized as follows. Section \ref{sec:minisup} briefly
recalls the minisuperspace description of two-field cosmological
models. Section \ref{sec:Noether} gives a geometric characterization
of time-independent Noether symmetries in such models, showing that
any such symmetry can be written as a direct sum between a visible and
a Hessian symmetry. In particular, we prove that the generating
function $\Lambda$ of a Hessian symmetry satisfies the Hesse and
$\Lambda$-$V$ equations, showing how the latter can be used to
determine $V$ in terms of $\Lambda$ through the method of
characteristics. We also discuss the integral of motion of a Hessian
symmetry, showing that it can be used to extract an algebraic (as
opposed to integral) formula in terms of the scalars $\varphi^i$ for
the number of e-folds along cosmological trajectories. Section
\ref{sec:rot} gives the classification of weakly-Hessian models with
rotationally-invariant scalar manifold metric, summarizing the results
proved in Appendix \ref{app:SolHesse}. This shows that the only scalar
manifolds which are allowed in such models are the disk, punctured
disk and annuli of constant Gaussian curvature $K=-\frac{3}{8}$. It
follows that weakly-Hessian models with rotationally-invariant scalar
manifold metric are elementary two-field $\alpha$-attractors in the
sense of reference \cite{elem}, for the specific value
$\alpha=\frac{8}{9}$. We next proceed to determine the general form of
the scalar potential of the corresponding Hessian models in each of
the three cases. The hyperbolic disk is discussed in Section
\ref{sec:Disk}. In this case, the results of Appendix
\ref{app:SolHesse} imply that the space of Hesse functions is three
dimensional and admits a basis formed by the classical Weierstrass
coordinates. This allows us to identify the space of Hesse functions
with the Minkowski 3-space $\R^{1,2}$, on which the group of
orientation-preserving isometries of the hyperbolic disk acts through
proper and orthochronous Lorentz transformations (as recalled in
Appendix \ref{app:iso}). This provides a natural classification of
Hesse functions into functions of timelike, spacelike and lightlike
type. For each of these three cases, we use the method of
characteristics for the $\Lambda$-$V$ equation and
representation-theoretic arguments to extract the explicit general
expression for those scalar potentials $V$ for which the model admits
the Hessian symmetry generated by a given Hesse function $\Lambda$. We
also describe the special cases when a Hessian model admits
visible symmetries. In Section \ref{sec:pDisk}, we determine
explicitly the scalar potential of Hessian two-field cosmological
models whose scalar manifold is a punctured disk of Gaussian curvature
$-\frac{3}{8}$ and explain when such models also admit visible
symmetries. In Section \ref{sec:Annuli}, we perform the same analysis
for models whose target is a hyperbolic annulus of Gaussian curvature
$-\frac{3}{8}$. Finally, Section \ref{sec:conclusions} gives our
conclusions and some directions for further research. In the appendices,
we review or derive various technical results used in the main text.

\paragraph{Notations and conventions.}
Throughout this paper, all manifols considered are assumed to be
connected, paracompact, Hausdorff and smooth. The scalar manifold of a
two-field cosmological model is assumed to be complete, as required by
conservation of energy; for simplicity, we also assume it to be
oriented. By a hyperbolic metric we always mean a metric which is
complete and of unit negative Gaussian curvature.  The notation $\rD$
indicates the Euclidean disk of unit radius, while $\dot{\rD}\eqdef
\rD\setminus\{0\}$ indicates the punctured Euclidean disk.  Finally,
$A(R)$ (where $R>1$) indicates the Euclidean annulus of inner radius
$1/R$ and outer radius $R$.  The notations $\mD$ and $\mD^\ast$
indicate the surfaces $\rD$ and $\dot{\rD}$ when endowed with their
unique hyperbolic metric, while $\mA(R)$ indicates the annulus $A(R)$,
when the latter is endowed with its unique hyperbolic metric of
modulus $\mu=2\log R$. The isometry group of an oriented Riemannian
2-manifold $(\Sigma,\cG)$ is denoted by $\Iso(\Sigma,\cG)$, while its
subgroup of orientation-preserving isometries is denoted by
$\Iso_+(\Sigma,\cG)$.  We refer the reader to the Appendices and to
references \cite{genalpha,elem} for a summary of some relevant
information about elementary hyperbolic surfaces.

\section{The minisuperspace Lagrangian of two-field cosmological models}
\label{sec:minisup}

\noindent In this section, we recall briefly the definition of
two-field cosmological models and their constrained Lagrangian
description in the minisuperspace approach.

\subsection{Two-field cosmological models}

\noindent We consider cosmological models with two real scalar fields
whose underlying space-time is a simply-connected and spatially flat
FLRW universe. Our models are parameterized by a connected, oriented
and complete Riemannian $2$-manifold $(\Sigma,\cG)$ (called the {\em
  scalar manifold}), together with a scalar potential given by a
smooth real-valued function $V$ defined on $\Sigma$. This data
combines into a {\em scalar triple} $(\Sigma,\cG,V)$.  The condition
that $\Sigma$ be connected and oriented is purely technical and it
could be relaxed. The condition that the metric $\cG$ is complete
insures conservation of energy in such models.

The models are obtained from the following action\footnote{We work in
  units where the reduced Planck mass $M$ equals one. The rescaling
  ${\tilde \cG}=M^{2}\cG$ and ${\tilde V}=M^{2} V$ gives
  $\tilde\cS=M^{2}\cS$, where $\tilde\cS[g,\varphi] = \int_{\R^4}
  \dd^4 x \,\sqrt{|\det g|} \left[ \frac{M^2}{2} R(g) -
    \frac{1}{2}\Tr_g\varphi^\ast({\tilde \cG}) - {\tilde V} (\varphi)
    \right]$ is the action more commonly found in the literature. },
which describes the coupling of Einstein gravity defined on $\R^4$
with the non-linear sigma model of target space $(\Sigma, \cG)$ and
scalar potential $V$:
\ben
\label{Action}
\cS[g,\varphi]:=\cS_{\Sigma,\cG,V}[g,\varphi]\! =\!\! \int_{\R^4}\!\! \dd^4 x 
\,\sqrt{|\det g|} \left[ \frac{R(g)}{2} - \frac{1}{2} g^{\mu \nu}\cG_{ij}
\pd_\mu\varphi^i\pd_\nu\varphi^j - V (\varphi) \right] .
\een
Here $R(g)$ is the scalar curvature of the space-time metric $g$ (which has
`mostly plus' signature), while $\varphi$ is a smooth map
from $\R^4$ to $\Sigma$, whose components in local coordinates on $\Sigma$
(denoted $\varphi^i$ with $i=1,2$) are the two real scalar fields of the
model. The cosmological model defined by the scalar triple $(\Sigma,\cG,V)$ is obtained from
\eqref{Action} by fixing $g$ to be the background metric of a
simply-connected FLRW universe with flat spatial section:
\ben
\label{FLRW}
\dd s_g^2 := - \dd t^2 + a^2(t) \dd \vec{x}^2 ~~(\mathrm{where}~~x^0=t~~,
~~\vec{x}=(x^1,x^2,x^3)~~\mathrm{and}~~a(t)>0~~\forall t)
\een
and taking $\varphi$ to depend only on the cosmological time $t$:
\ben
\label{PhiHom}
\varphi=\varphi(t)~~.
\een
Let $H\eqdef \frac{\dot{a}}{a}$ denote the Hubble parameter, where
$\dot{a}\eqdef \frac{\dd a}{\dd t}$. 

\subsection{The minisuperspace Lagrangian}

\noindent Substituting \eqref{FLRW} and \eqref{PhiHom} in
\eqref{Action} and ignoring the integration over the spatial variables
$\vec{x}$ gives the {\em minisuperspace action}:
\ben 
\label{S}
S[a,\varphi]= \int_{-\infty}^{+\infty} \!\dd t \, a^3
\!\left[ \frac{3 (\dot{a}^2 + a \ddot{a})}{a^2} + \frac{1}{2}
\cG_{ij}(\varphi)\dot{\varphi}^i\dot{\varphi}^j - V (\varphi) \right]~,
\een
where $\dot{\varphi}\eqdef \frac{\dd \varphi}{\dd t}$. Notice that $S$
depends explicitly on the scale factor $a$ of the FLRW metric. The
functional \eqref{S} can be viewed as the classical action of a
mechanical system with three degrees of freedom. Indeed,
integration by parts in the $\ddot{a}$ term of \eqref{S} allows us to
write:
\ben
\label{cS}
S[a,\varphi]=\int_{-\infty}^{\infty} \! \dd t\, L(a(t),\dot a(t),\varphi(t)
,\dot{\varphi}(t))~~,
\een
where $L:=L_{(\Sigma,\cG,V)}$ is the {\em minisuperspace Lagrangian}:
\ben 
\label{L}
L(a,\dot a, \varphi,\dot{\varphi})\eqdef - 3 a \dot{a}^2 + a^3 \left[\frac{1}{2}
\cG_{ij}(\varphi)\dot{\varphi}^i\dot{\varphi}^j - V (\varphi)\right]~.
\een
In this formulation, the triplet $(a,\varphi^i)$ provides local
coordinates on the configuration space:
\ben
\label{cM}
\cN\eqdef \R_{>0}\times \Sigma~~
\een
of this mechanical system. The Euler-Lagrange equations of \eqref{L} take the form:
\beqan
\label{EL}
3H^2+2\dot{H}+\frac{1}{2}\cG_{ij}(\varphi) \dot{\varphi}^i\dot{\varphi}^j  
- V(\varphi) &=& 0\\
(\nabla_t +3H)\dot{\varphi}^j + \cG^{ij}(\varphi) (\partial_i V)(\varphi) &=& 0~~.\nn
\eeqan
Here $\nabla_t \dot{\varphi}^i \eqdef \ddot{\varphi}^i + \Gamma^i_{jk}
\dot{\varphi}^j \dot{\varphi}^k$, where $ \Gamma^i_{jk}$ are the
Christoffel symbols of the scalar manifold metric $\cG$, $\pd_i
V\eqdef \frac{\pd V}{\pd \varphi^i}$ and we used the relation:
\be
\frac{2a\ddot{a}+\dot{a}^2}{a^2}=3H^2+2\dot{H}~~.
\ee
To recover the cosmological equations of motion, one must subject the
Lagrangian \eqref{L} to the zero energy constraint $E_L = \frac{\pd
  L}{\pd \dot{a}} \dot{a}+ \frac{\pd L}{\pd \dot{\varphi}^i}
\dot{\varphi}^i - L =0$, which takes the following explicit form:
\ben
\label{Friedmann}
\frac{1}{2}\cG_{ij}(\varphi)\dot{\varphi}^i\dot{\varphi}^j + V(\varphi)=3 H^2~~,
\een
thus coinciding with the time-time component of the equations of
motion derived from \eqref{Action}. This is often called the Friedmann
constraint. On solutions of the Euler-Lagrange equations 
\eqref{EL}, this constraint gives a relation between the integration
constants, thus reducing the number of independent constants by one
(see, for instance, reference \cite{CR}). However, one can also use
the Friedmann constraint from the outset to solve algebraically for
the Hubble parameter $H$:
\ben
\label{Heq}
H=\frac{1}{\sqrt{6}} \left[\cG_{ij}(\varphi)\dot{\varphi}^i\dot{\varphi}^j 
+ 2V(\varphi)\right]^{1/2}~.
\een
To recapitulate, the system formed by the Euler-Lagrange
equations \eqref{EL} together with the Friedmann constraint \eqref{Friedmann} is
equivalent with the cosmological equations of the two-field model obtained 
from \eqref{Action}.

\section{Noether symmetries for general two-field models}
\label{sec:Noether}

\noindent In this section, we consider time-independent infinitesimal
Noether symmetries of the minisuperspace Lagrangian \eqref{L}. By
analyzing the Noether symmetry condition, we show that the
corresponding Noether generator decomposes as a direct sum between a
visible and a hidden symmetry, the latter type of symmetry being
determined by a {\em Hesse function} $\Lambda$.  For Hessian
symmetries, we explain how the $\Lambda$-$V$ equation allows one to
extract the general form of the scalar potential using the method of
characteristics. We also discuss the natural action of the isometry
group of the scalar manifold on the linear space of all Hesse
functions and on the linear space of all scalar potentials which
satisfy the $\Lambda$-$V$ equation. Finally, we consider the
conservation law associated to a Hessian symmetry, showing that it
allows one to determine the number of e-folds along cosmological
trajectories algebraically in $\varphi^i$, instead of through an integral.

\subsection{Noether generators and integrals of motion}

\noindent Recall that the configuration space $\cN=\R_{>0}\times
\Sigma$ of the minisuperspace model is a product of the target space
$\Sigma$ (which is locally parameterized by $\varphi^1$ and
$\varphi^2$) with the range $\R_{>0}$ of the scale factor
$a$. Geometrically, the Lagrangian \eqref{L} is a smooth real-valued
function defined on the tangent bundle $T\cN$, which identifies
naturally with the first jet bundle of curves of $\cN$ (see
\cite{Olver}). This tangent bundle decomposes as:
\be
T\cN=T_{(a)}\cN \oplus T_{(\varphi)} \cN~~,
\ee
where $T_{(a)}\cN$ is the pullback of the tangent bundle of the
half-line $\R_{>0}$ through the first projection and $T_{(\varphi)}
\cN$ is the pullback of the tangent bundle of $\Sigma$ through the
second projection. Hence any vector field $X$ defined on $\cN$
decomposes uniquely as:
\be
X=X_{(a)}+X_{(\varphi)}~~,
\ee
where $X_{(a)}$ is a vector field defined on the half-line and
$X_{(\varphi)}$ is a vector field defined on $\Sigma$. In local
coordinates on the configuration space $\cN$, we have:
\ben 
\label{X_perp_X_par}
X_{(a)}=X^a(a,\varphi) \frac{\pd}{\partial a}~~,
~~X_{(\varphi)}=X^i(a,\varphi) \frac{\pd}{\pd \varphi^i}~~.
\een
The first jet prolongation $\cX$ of $X$ is a vector field defined on
$T\cN$ which is given in local coordinates by the formula (see
\cite{Olver}):
\be
\cX=X+\dot{X}^a(a,\varphi,\dot{a},\dot{\varphi})\frac{\pd}{\pd \dot{a}}
+\dot{X}^{i}(a,\varphi,\dot{a},\dot{\varphi})\frac{\pd}{\pd \dot{\varphi}^i}~~,
\ee
where:
\be
\dot{\lambda}(a,\varphi,\dot{a},\dot{\varphi})\eqdef
 (\pd_t \lambda)(a,\varphi)+(\pd_a \lambda)(a,\varphi) \dot{a}+
(\pd_i\lambda)(a,\varphi) \dot{\varphi}^i
\ee
for any smooth function $\lambda$, where we use the notations
$\pd_t\eqdef\frac{\pd}{\pd t}$, $\pd_a\eqdef\frac{\pd}{\pd a}$,
$\pd_i\eqdef\frac{\pd}{\pd \varphi^i}$\,. The vector field $X$ is a
variational symmetry of the Lagrangian \eqref{L} provided that it
satisfies the {\em Noether condition}:
\ben
\label{StrongVarSymCond}
\cL_{\cX} (L)=0~~,
\een
where $\cL_{\cX}$ denotes the Lie derivative with respect to the
prolongation $\cX$.  In local coordinates on $\Sigma$, this condition
takes the form:
\ben
\label{VarSymCondLocal}
P(a,\varphi,\dot{a},\dot{\varphi})=0~~,
\een
where the polynomial $P$ is given by:
\ben 
\label{Pol}
P(a, \varphi, \dot{a},\dot{\varphi})\eqdef X^a \frac{\pd L}{\pd
  a} + X^{i} \frac{\pd L}{\pd \varphi^i} + \dot{X}^a \frac{\pd L}{\pd
  \dot{a}} + \dot{X}^{i} \frac{\pd L}{\pd \dot{\varphi}^i}~~.
\een
Given a variational symmetry $X$ of $L$, the associated
integral of motion has the following local expression (see \cite{Olver}):
\ben
\label{rJloc}
\rJ_X=X^a \pd_{\dot{a}}L+X^i \pd_{\dot{\varphi}^i}L =
-6 a \dot{a} X^a + \cG_{ij}(\varphi) a^3 X^{i} \dot{\varphi}^j ~~.
\een

\subsection{The characteristic system}

\noindent In this subsection, we show that the Noether symmetry
condition \eqref{StrongVarSymCond} for the minisuperspace Lagrangian \eqref{L}
amounts to the requirement that $X$ has the form
$X(a,\varphi)=X_{(a)}(a,\varphi) +X_{(\varphi)}(a,\varphi) = X^a \pd_a + X^i \pd_i$, with:
\ben
\label{Xsol}
X_{(a)}(a,\varphi)=\frac{\Lambda(\varphi)}{\sqrt{a}}\pd_a~~,
~~X_{(\varphi)}(a,\varphi)= \left[ Y^i(\varphi) 
- \frac{4}{a^{3/2}} \cG^{ij}(\varphi) \pd_j \Lambda(\varphi) \right] \!\pd_i~~,
\een
where $\Lambda$ is a smooth real-valued function defined on $\Sigma$ and $Y
= Y^i \pd_i$ is a smooth vector field defined on $\Sigma$, which satisfy the {\em
  characteristic system}:
\beqan
\label{Sindex}
&& \left(\partial_i \partial_j -\Gamma_{ij}^k \partial_k\right) \Lambda=
\frac{3}{8}\cG_{ij}\Lambda \nn\\
&& \cG^{ij} \partial_i V \partial_j\Lambda =\frac{3}{4}V \Lambda \nn\\
&& \nabla_i Y_j+\nabla_j Y_i=0~~\\
&& Y^i \partial_i V=0~~~.\nn
\eeqan
In index-free notation, the general Noether generator reads:
\ben
\label{Xgen}
X=X_\Lambda+Y~~,
\een
where:
\ben
\label{XLambda}
X_\Lambda=\frac{\Lambda}{\sqrt{a}}\pd_a-\frac{4}{a^{3/2}}\grad_\cG \Lambda~~.
\een
Let us begin by computing the polynomial \eqref{Pol}:
\beqa
P&=&3 a^2 X^a\left[-H^2+\frac{1}{2}\cG_{ij}\dot{\varphi}^i\dot{\varphi}^j
-V\right]+a^3 X^{i} \left[\frac{1}{2} \partial_i \cG_{jk} \dot{\varphi}^j\dot{\varphi}^k -
  \partial_i V \right]\nn\\
&& -6 a^2 H \dot{X}^a+ a^3 \cG_{ij} \dot{\varphi}^j  \dot{X}^{i}~~.
\eeqa
Using $H=\frac{\dot a}{a}$ and expanding in powers of $\dot{a}$ and $\dot{\varphi}$ gives:
\be
P=P_o(a,\varphi)+ P_{00}(a,\varphi)\dot{a}^2 + P_{0i}(a,\varphi) 
\dot{a}\dot{\varphi}^i+P_{ij}(a,\varphi) \dot{\varphi}^i\dot{\varphi}^j~~,
\ee
where:
\beqan
\label{P_terms}
P_{00} &=& -3X^a -6 a \partial_aX^a\nn\\
P_{0i} &=& -6 a \partial_iX^a +a^3 \cG_{ij} \partial_{a} X^{j}\nn\\
P_{ij}=P_{ji}&=& \frac{3}{2} a^2 \cG_{ij} X^a + a^3 \frac{1}{2}
\left(\nabla_i X_j+\nabla_jX_i\right)~~\\
P_o &=& -3 a^2 V X^a -a^3  X^i \partial_i V ~~.\nn
\eeqan
Explicitly, the $a^3$-term of $P_{ij}$ reads:
\beqa
&&\frac{1}{2}\left(\nabla_i X_j+\nabla_jX_i\right)=\frac{1}{2}\left(\partial_i X_j
+\partial_jX_i-2\Gamma_{ij}^k X_k \right)\nn\\
&&=\frac{1}{2}\left[(\partial_k \cG_{ij}) X^k+\cG_{ki} \partial_{j}X^k+\cG_{kj} 
\partial_i X^k\right]~~.  
\eeqa
Using \eqref{P_terms}, the Noether symmetry condition \eqref{VarSymCondLocal}
 amounts to the system:
\beqan
\label{CharSys}
&& (\mathrm{coeff.~of}~\dot{a}^2~)~~~:~~X^a +2 a \partial_aX^a=0\nn\\
&& (\mathrm{coeff.~of}~\dot{a}\dot{\varphi}^i)~~:~ -6 \partial_iX^a +a^2 
\cG_{ij} \partial_{a} X^{j}=0\nn\\
&& (\mathrm{coeff.~of}~\dot{\varphi}^i\dot{\varphi}^j)~:~ 3 \cG_{ij} X^a 
+ a \left(\nabla_i X_j+\nabla_jX_i\right)=0~~\\
&& (\mathrm{potential~term})~:~ 3 V X^a +a  X^i \partial_i V =0~~.\nn
\eeqan
The first equation in \eqref{CharSys} implies that the first
relation in \eqref{Xsol} holds for some smooth function $\Lambda
(\varphi)$. Using this into the second equation of \eqref{CharSys} gives:
\beqan
\label{CharSysEq2}
\cG_{ij} \pd_a X^j = 6 \frac{\pd_i \Lambda}{a^{5/2}}~~.
\eeqan
Integrating \eqref{CharSysEq2} with respect to $a$ gives the second relation in \eqref{Xsol}:
\beqan
X^i (a,\varphi) = - \frac{4}{a^{3/2}} \cG^{ij} \pd_j \Lambda (\varphi) + Y^i (\varphi)~~,
\eeqan
where $Y$ is a vector field defined on $\Sigma$. Substituting \eqref{Xsol} 
into the third equation of
\eqref{CharSys} gives: 
\beqan
\label{CharSysEq3}
\frac{1}{a^{1/2}} \left[ 3 \cG_{ij} \Lambda - 4 (\nabla_i \pd_j \Lambda 
+ \nabla_j \pd_i \Lambda) \right] + a ( \nabla_i Y_j + \nabla_j Y_i ) = 0~~.
\eeqan
Since the terms multiplying powers of $a$ in this relation are themselves independent of $a$, 
taking the limits $a\rightarrow 0$ and $a\rightarrow \infty$ shows that
\eqref{CharSysEq3} is equivalent with the first and third equations of
the characteristic system \eqref{Sindex}.  Finally, substituting
\eqref{Xsol} into the fourth equation of \eqref{CharSys} gives:
\beqan
\label{CharSysEq4}
\frac{1}{a^{1/2}} \left[3V\Lambda - 4 \cG^{ij} \pd_j \Lambda \pd_i V \right] 
+ a Y^i \pd_i V = 0~~.
\eeqan
Again taking the limits $a\rightarrow 0$ and $a\rightarrow \infty$ shows that
\eqref{CharSysEq4} is equivalent with the second and fourth equations
of \eqref{Sindex}. This completes the proof that the 
characteristic system \eqref{Sindex} is equivalent with the Noether
symmetry condition \eqref{StrongVarSymCond}.

The integral of motion \eqref{rJloc} of the Noether symmetry \eqref{Xsol}
described by a solution $(\Lambda,Y)$ of the characteristic system
takes the form:
\be
\rJ_X=-6 a \dot{a} X^a + \cG_{ij}(\varphi) a^3 X^{i} \dot{\varphi}^j=
-6 \dot{a}\sqrt{a} \Lambda(\varphi)+ a^3  \cG_{ij}(\varphi)Y^i(\varphi)\dot{\varphi}^j-
4 a^{3/2} \dot{\Lambda}~~,
\ee
i.e.: 
\ben
\label{rJX}
\rJ_X=-4\frac{\dd}{\dd t} \left[a^{3/2}\Lambda(\varphi)\right]+ 
a^3 \cG_{ij}(\varphi)Y^i(\varphi)\dot{\varphi}^j ~~.
\een

\subsection{Natural subsystems of the characteristic system. Hessian and visible symmetries}

\noindent Notice that those equations of the system \eqref{Sindex}
which contain $\Lambda$ decouple from those equations which contain
$Y$. Hence the characteristic system naturally splits into two
subsystems of partial differential equations, namely:\\ \\ the {\em
  $\Lambda$-system}:
\beqan
\label{LambdaSys}
&& \nabla_i \pd_j \Lambda = \frac{3}{8}\cG_{ij} \Lambda\nn\\
&& \cG^{ij} \pd_i V \pd_j \Lambda = \frac{3}{4} V \Lambda~~
\eeqan
and the {\em $Y$-system}:
\beqan
\label{YSys}
&& \nabla_i Y_j + \nabla_j Y_i = 0\nn\\
&& Y^i \pd_i V=0~~.
\eeqan
A vector field of the form \eqref{Xsol} for which $Y=0$ and $\Lambda$
is a smooth solution of the $\Lambda$-system will be called an
infinitesimal {\em Hessian symmetry} of $(\Sigma,\cG,V)$. A scalar
triple which admits Hessian symmetries will be called a {\em Hessian
  triple}; in this case, the corresponding two-field cosmological
model will be called a {\em Hessian model}. On the other hand, a
vector field of the form \eqref{Xsol} for which $\Lambda=0$ and the
vector field $Y$ is a non-trivial smooth solution of the $Y$-system
will be called an infinitesimal {\em visible symmetry} of
$(\Sigma,\cG,V)$. A scalar triple which admits visible symmetries is
called a {\em visibly symmetric triple} and the corresponding
two-field cosmological model will be called a {\em visibly-symmetric
  model}.

The result proved in the previous subsection implies that any
time-independent infinitesimal Noether symmetry of the minisuperspace
system decomposes as a direct sum of a visible symmetry with a Hessian
symmetry. Notice that infinitesimal visible symmetries coincide with
those Killing vector fields of $(\Sigma,\cG)$ which generate
isometries preserving the scalar potential $V$; they are the `obvious'
symmetries of the two-field cosmological model defined by the scalar
triple $(\Sigma,\cG,V)$. Unlike visible symmetries, Hessian symmetries
are not geometrically obvious and can be viewed as `hidden symmetries'
of the model. Also notice that the first and third equations in
\eqref{Sindex} do not depend on $V$. For a given scalar
manifold $(\Sigma,\cG)$, these equations can be solved for $\Lambda$
and $Y$. Fixing solutions $(\Lambda,Y)$ of these two equations, the
second and fourth equations of the characteristic system can then be
used to determine the scalar potential in terms of $\Lambda$ and $Y$.

\subsection{Rescaling the scalar manifold metric. The Hesse and $\Lambda$-$V$ equations}
\label{subsec:LambdaV}

\noindent It is convenient to consider the rescaled scalar manifold metric:
\ben
\label{G}
G\eqdef \beta^2 \cG=\frac{3}{8} \cG\Longleftrightarrow \cG
=\frac{1}{\beta^2}G=\frac{8}{3}G~~,
\een
where:
\ben
\label{betadef}
\beta\eqdef \sqrt{\frac{3}{8}}~~.
\een
Since the Levi-Civita connection of $\cG$ is invariant under such a
rescaling, the $\Lambda$-system becomes:
\beqan
\label{UnitLambdaSys}
&& \nabla \dd \Lambda = G \Lambda\nn\\
&& \langle \dd V, \dd \Lambda \rangle_G = 2 V \Lambda~~,
\eeqan
while the $Y$-system preserves its form when expressed with respect to
the rescaled metric $G$.

The first equation in \eqref{UnitLambdaSys}:
\ben
\label{HessCond}
\nabla \dd \Lambda = G \Lambda~~,
\een
(whose left hand side equals the Hessian tensor of $\Lambda$ computed
with respect to the scalar manifold metric $\cG$) will be called the
{\em Hesse equation} of the rescaled scalar manifold $(\Sigma,G)$ and
its smooth solutions $\Lambda$ will be called {\em Hesse functions} of
$(\Sigma,G)$. Let $\cS(\Sigma,G)$ denote the linear space of such
functions. A Riemannian manifold $(\Sigma,G)$ is called {\em
  Hesse}\footnote{This should not be confused with the notion of {\em
    Hessian} manifold, which is a different concept!} if it admits
non-trivial Hesse functions, i.e. if $\cS(\Sigma,G)\neq 0$.  The
second equation in \eqref{UnitLambdaSys}:
\ben
\label{LambdaVEq}
\langle \dd V, \dd \Lambda \rangle_G = 2 V \Lambda
\een
will be called the {\em $\Lambda$-$V$ equation} of the rescaled scalar
triple $(\Sigma,G,V)$. Let $\cV(G,\Lambda)$ denote the linear space of
smooth functions $V$ satisfying this equation.

The Hesse equation \eqref{HessCond} is invariant under the natural action:
\be
\Lambda\rightarrow \Lambda\circ \psi^{-1}~~,~~\forall \psi\in \Iso(\Sigma,G)
\ee
of the isometry group $\Iso(\Sigma,G)=\Iso(\Sigma,\cG)$ of the scalar manifold. In
particular, such transformations preserve the space $\cS(\Sigma,G)$ of
Hesse functions of $(\Sigma,G)$.  Similarly, equation
\eqref{LambdaVEq} is invariant under the natural action of the
isometry group of the scalar manifold on the pair $(V,\Lambda)$:
\be
(\Lambda,V)\rightarrow (\Lambda\circ \psi^{-1},V\circ \psi^{-1})~~,
~\forall \psi\in \Iso(\Sigma,G)~~.
\ee
Hence an isometry of $(\Sigma,G)$ takes the general solution of
\eqref{LambdaVEq} into the general solution of the same equation, but
with $\Lambda$ replaced by $\Lambda\circ \psi^{-1}$:
\ben
\label{cVAction}
\cV(G,\Lambda\circ \psi^{-1})=\cV(G,\Lambda)\circ \psi^{-1}~~,
~~\forall \psi\in \Iso(\Sigma,G)~~,~~\forall \Lambda\in \cS(\Sigma,G)~~.
\een

\begin{remark}
\label{rem:Vscale}
Equation \eqref{LambdaVEq} is invariant under rescalings
$\Lambda\rightarrow C\Lambda$, where $C$ is any non-zero constant.
This implies that the general solution of this equation is
unchanged if one rescales $\Lambda$ by $C$:
\be
\cV(G,C\Lambda)=\cV(G,\Lambda)~~,~~\forall C\in \R\setminus\{0\}~~.
\ee
In particular, the general solution of the $\Lambda$-$V$ equation
\eqref{LambdaVEq} depends only on the ray of $\Lambda$ in the real
projective space $\mathbb{P} \cS(\Sigma,G)$.
\end{remark}

\subsection{The scalar potential of a Hessian symmetry}
\label{subsec:V}

\noindent Given a Hesse function $\Lambda\in \cS(\Sigma,G)$, consider
the $\Lambda$-$V$-equation \eqref{LambdaVEq} with $G= \beta^2 \cG$ 
as defined in \eqref{G}: 
\ben
\label{LambdaV}
\langle \dd V,\dd \Lambda\rangle_\cG= 2 \beta^2 V \Lambda~~,
\een
where $\langle \dd V,\dd \Lambda\rangle_\cG=\langle \grad_\cG
V,\grad_\cG\Lambda\rangle_\cG=\cG^{ij} \pd_i V \pd_j \Lambda$.

One can solve \eqref{LambdaV} through the method of characteristics
(see Appendix \ref{app:char}). For this, let $\gamma$ be a
$\cG$-gradient flow curve of $\Lambda$ with gradient flow parameter $q$:
\ben
\label{gradflow}
\frac{\dd \gamma(q)}{\dd q}=-(\grad_\cG \Lambda)(\gamma(q))~~.
\een
Then \eqref{gradflow} and \eqref{LambdaV} imply: 
\be
\frac{\dd}{\dd q} V(\gamma(q))=-\langle \grad_\cG V,\grad_\cG\Lambda
\rangle_\cG\Big{|}_{\gamma(q)}=-\langle \dd V,\dd \Lambda\rangle_\cG\Big{|}_{\gamma(q)}
=-2 \beta^2 \Lambda (\gamma(q)) V(\gamma(q))~~,
\ee
which gives:
\ben
\label{Vgammaq}
V(\gamma(q))=V(\gamma(q_0))\,e^{\stackrel{-2\beta^2 \!\!
\bigintssss_{q_0}^q{\Lambda \dd q}}{\!\!\!\!\!\!\!\gamma}}~~.
\een
It is convenient to use the restriction $\lambda=\Lambda\circ \gamma$ (i.e.
$\lambda(q)=\Lambda(\gamma(q))$) of $\Lambda$ to $\gamma$ as a
parameter on the gradient flow curve (notice that $\lambda$ decreases
with $q$). We have:
\be
\frac{\dd \lambda}{\dd q}=(\dd_{\gamma(q)}\Lambda)\left(\frac{\dd\gamma(q)}{\dd q}\right)=
-||(\grad_\cG \Lambda)(\gamma(q))||^2_\cG=-||(\dd\Lambda)(\gamma(q))||^2_\cG~~,
\ee
which gives:
\ben
\label{dqdlambda}
\dd q=-\frac{\dd \lambda}{||\dd \Lambda||^2_\cG}~~.
\een
Hence \eqref{Vgammaq} becomes:
\ben
\label{Vgamma}
V(\gamma(q))=V(\gamma(q_0))\,e^{\stackrel{2\beta^2 
\bigintssss_{\Lambda(\gamma(q_0))}^{\Lambda(\gamma(q))}{\frac{\lambda \dd \lambda}
{||\dd \Lambda||_\cG^2}}}{\!\!\!\!\!\!\!\!\!\!\!\!\!\!\!\!\!\!\!\!\!\!\!\!\!\!\!\!\!\!\gamma}}~~.
\een
This relation allows us to find the general solution of
\eqref{LambdaV}, provided that we can determine the gradient flow of
$\Lambda$. To deal with the initial conditions, one can choose a
section of the gradient flow, i.e. a (possibly disconnected)
submanifold $\cQ:=\cQ_\Lambda$ of $\Sigma$ with the property that each
gradient flow curve $\gamma$ of $\Lambda$ meets $\cQ$ in exactly one
point. For any $p\in \Sigma$, let $\gamma$ be the gradient flow curve
which passes through $p$ and meets $\cQ$ at the point $p_0$ (namely,
we have $\gamma(q)=p$ and $\gamma(q_0)=p_0$. Then relation
\eqref{Vgamma} gives:
\be
V(p)=V(p_0)\,e^{\stackrel{2\beta^2 
\bigintssss_{\Lambda(p_0)}^{\Lambda(p)}{\frac{\lambda \dd \lambda}
{||\dd \Lambda||_\cG^2}}}{\!\!\!\!\!\!\!\!\!\!\!\!\!\!\!\!\!\!\!\!\!\!\gamma}}~~.
\ee
The correspondence $p\rightarrow p_0$ defines a smooth function
$F=F_\cQ:\cQ\rightarrow \Sigma$ which allows us to write the previous
equation as:
\ben
\label{Vsol}
V(p)=\omega(p)\,e^{\stackrel{2\beta^2 \bigintssss_{\Lambda(F(p))}^{\Lambda(p)}
{\frac{\lambda \dd \lambda}{||\dd \Lambda||_\cG^2}}}{\!\!\!\!\!\!\!\!\!\!\!\!\!\!\!\!\!\!\!\!\!\!\!\!\!\!\!\!\gamma}}~~,
\een
where $\omega\eqdef V\circ F$ is a real-valued smooth function defined on
$\cQ$. 

\begin{remark} Given any non-zero constant $C$, the gradient flow of
$\Lambda$ coincides with that of $C\Lambda$, up to a constant
rescaling $q\rightarrow q/C$ of the gradient flow parameter.
\end{remark}

\subsection{The integral of motion of a Hessian symmetry}
\label{sec:iom}

\noindent 
For a Hessian symmetry ($Y=0$) with generator $\Lambda$, the integral
of motion \eqref{rJX} gives:
\ben
\label{cons}
a(t)^{3/2}\Lambda(\varphi(t))=C-\frac{\rJ_\Lambda}{4} (t-t_0)~~,
~~\mathrm{where}~~C=a_0^{3/2}\Lambda_0~~
\een
and we defined:
\be
a_0\eqdef a(t_0)~~,~~\Lambda_0\eqdef \Lambda(\varphi(t_0))~~.
\ee
The conserved quantity $\rJ_\Lambda$ is independent of $t$ along
every solution of the Euler-Lagrange equations (but depends on the
initial conditions of the solution). Differentiating \eqref{cons} with
respect to time at $t=t_0$ gives:
\ben
\label{rJ}
\!\!\!-\frac{\rJ_\Lambda}{4}\!=\!\frac{\dd}{\dd t}\left[a(t)^{3/2}\Lambda(\varphi(t))\right]
\!\Big|_{t=t_0}\!\!\!\!=
a_0^{3/2} \!\left[\frac{3}{2}H_0\Lambda_0\!+
\!(\dd_{\varphi(t_0)}\Lambda)(\dot{\varphi}_0) \right] ,
\een
where $H_0\eqdef H(t_0)$ is determined by the Friedmann constraint
(assuming that $H(t)>0$):
\be
H_0=\frac{1}{\sqrt{6}}\left[||\dot{\varphi}_0||_\cG^2+2V(\varphi_0)\right]^{1/2}~~.
\ee
Here:
\be
\varphi_0\eqdef \varphi(t_0)~~\mathrm{and}~~\dot{\varphi}_0\eqdef \dot{\varphi}(t_0)~~.
\ee
Relation \eqref{rJ} allows one to determine $\rJ_\Lambda$ from the
initial conditions of the cosmological trajectory, while \eqref{cons} determines the
$t$-dependence of the cosmological scale factor along the trajectory:
\ben
\label{aLambda}
a(t)=\left[\frac{C-\frac{\rJ_\Lambda}{4} (t-t_0)}{\Lambda(\varphi(t))}\right]^{2/3}
=a_0\left[\frac{\Lambda_0+\left(\frac{3}{2}H_0
\Lambda_0+(\dd_{\varphi_0}\Lambda)(\dot{\varphi}_0) \right)
 (t-t_0)}{\Lambda(\varphi(t))}\right]^{2/3} .
\een
In turn, this determines the e-fold function $N_{t_0}(t)=
\log\left[\frac{a(t)}{a(t_0)}\right]$ along any given scalar field
trajectory $\varphi(t)$, where $t_0$ is a reference cosmological time:
\ben
\label{cNLambda}
N_{t_0}(t)=\frac{2}{3} \log\left[\frac{\Lambda_0+
\left(\frac{3}{2}H_0\Lambda_0+(\dd_{\varphi_0}\Lambda)
(\dot{\varphi}_0) \right) (t-t_0)}{\Lambda(\varphi(t))}\right]~~.
\een
\noindent In particular, a scalar field trajectory which is
inflationary for the cosmological time interval $[t_0,t]$ will produce
a desired number $N$ of e-folds during that time interval provided
that its initial and final points $\varphi(t_0),\varphi(t)\in \Sigma$
satisfy the condition:
\ben
\label{Ncond}
e^{\frac{3N}{2}}\Lambda(\varphi(t))-\Lambda_0=\left(\frac{3}{2}H_0
\Lambda_0+(\dd_{\varphi_0}\Lambda)(\dot{\varphi}_0) \right) (t-t_0)~~.
\een

\begin{remark}
The e-fold function is also determined as follows by the Friedmann constraint:
\ben
\label{cNF}
N_{t_0}(t)=\int_{t_0}^t  H(\tau)\dd \tau =\frac{1}{\sqrt{6}}\int_{t_0}^t 
\Big[||\dot{\varphi}(\tau)||_\cG^2+2V(\varphi(\tau))\Big]^{1/2}\dd \tau~~.
\een
This non-local relation involves integration of a complicated quantity
depending on both $\varphi(\tau)$ and $\dot{\varphi}(\tau)$ for
$\tau\in [0,t]$, unlike the much simpler formula \eqref{cNLambda}
(which involves no integrations). 
\end{remark}

\noindent Differentiating \eqref{cNLambda} with respect to $t$ gives:
\ben
\label{HLambda}
H(t)=\frac{2}{3}
\frac{\frac{3}{2}H_0\Lambda_0\!+\!(\dd_{\varphi_0}\Lambda)
(\dot{\varphi}_0)}{\Lambda_0\!+\!\Big(\!\frac{3}{2}H_0
\Lambda_0\!+\!(\dd_{\varphi_0}\Lambda)(\dot{\varphi}_0)
  \!\Big) (t\!-\!t_0)} \!-\!\frac{(\dd_{\varphi(t)}
  \Lambda)(\dot{\varphi}(t))}{\Lambda(\varphi(t))}~~.
\een

\paragraph{The case $\dot{\varphi}_0=0$.}

When $\dot{\varphi}_0=0$, relation \eqref{Ncond} reduces to: 
\ben
\label{Ncond0}
\left[1+\frac{3 H_0}{2}(t-t_0)\right]\frac{\Lambda_0}{\Lambda_\varphi(t)} =e^{\frac{3}{2}N_{t_0}(t)}~~,
\een
where $H_0=\frac{1}{\sqrt{3}}V_0^{1/2}>0$ with $V_0\eqdef
V(\varphi_0)$ and we defined $\Lambda_\varphi(t)\eqdef
\Lambda(\varphi(t))$. On the other hand, relation \eqref{HLambda}
reduces to the following condition when $\dot{\varphi}_0=0$:
\ben 
\label{HubP}
H(t)=\frac{2H_0}{2+3H_0 (t-t_0)}
-\frac{\dot{\Lambda}_\varphi(t)}{\Lambda_\varphi(t)}~~.
\een
Hence positivity of $H$ requires: 
\ben
\label{Ineq0}
\frac{\dot{\Lambda}_\varphi}{\Lambda_\varphi}<\frac{2H_0}{2+3 H_0(t-t_0)}~~\mathrm{i.e.}~~\frac{\Lambda_\varphi(t)}{\Lambda_0}<\left[2+3 H_0(t-t_0)\right]^{2/3}~~.
\een
Let:
\ben
s(t)\eqdef 3H_0-[2+3 (t-t_0) H_0] \frac{\dot{\Lambda}_\varphi(t)}{\Lambda_\varphi}~~.
\een
Then condition \eqref{Ineq0} amounts to:
\ben
\label{Ineq1}
s(t)>H_0~~.
\een
Recall that the inflationary time periods of the cosmological
trajectory are defined by the condition that $a$ is a convex and
strictly increasing function of $t$, i.e. $\dot{a} (t)>0$
and $\ddot{a}(t)>0$. Using \eqref{Ncond0}, this
amounts to requiring that the function:
\ben
\label{ft}
a(t)=a(t_0) \left(\left[1+\frac{3}{2}H_0(t-t_0)\right]\frac{\Lambda_0}{\Lambda_\varphi(t)}\right)^{2/3}
\een
be convex and strictly increasing. Let us assume that $\Lambda(t)>0$ along the
trajectory (hence $\Lambda_0>0$). Since \eqref{Ineq1} implies
$\dot{a}(t)>0$, the requirement for inflation amounts to $\ddot{a}(t)>0$, i.e.:
\ben
\label{Ineq2}
\frac{\ddot{\Lambda}_\varphi(t)}{\Lambda_\varphi(t)}<\frac{s(t) [5 s(t)-18 H_0]}{3 [2+3 H_0 (t-t_0)]^2}~~.
\een
Conditions \eqref{Ineq1} and \eqref{Ineq2} can be used to determine
the upper limit $t_f$ of an inflationary time interval $[t_0,t_f]$ for
which $\dot{\varphi}_0=0$.

\section{Weakly-Hessian models with rotationally-invariant scalar manifolds}
\label{sec:rot}

\noindent In this section, we give the characteristic system for
models with rotationally-invariant scalar manifold metrics and the
classification of weakly-Hessian two-field models. The proof of this
classification is given in Appendix \ref{app:SolHesse}. As shown in
that appendix, the scalar manifold of any weakly-Hessian model is a
disk, a punctured disk or an annulus, endowed with its complete metric
of Gaussian curvature $K=-3/8$. We also list the general solutions of
the Hesse equation in each of the three cases, solutions which are
derived in the same appendix. In the next sections, we will consider
each case in turn, extracting the explicit form of the scalar
potential for which the corresponding weakly-Hessian models
admit a Hessian symmetry.

\subsection{The characteristic system}

\noindent Consider the case when $\Sigma$ is diffeomorphic with the
unit disk $\rD$ or with the punctured unit disk $\dot{\rD}$, endowed
with a metric $\cG$ which is rotationally-invariant:
\ben
\label{ssg}
\dd s_{\cG}^2=\dd r^2+f(r)\dd\theta^2~~.
\een
Here $r$ and $\theta$ are normal polar coordinates for $\cG$ and $f$
is a smooth and everywhere-positive real-valued function (which
extends to the origin in the case $\Sigma\simeq \rD$). For application
to cosmological models, we must require that the scalar metric $\cG$
is complete (otherwise the dynamics of the cosmological model
  would violate conservation of energy). The non-vanishing
Christoffel symbols are:
\be
\Gamma_{\theta\theta}^r=-\frac{f'}{2}~~,~~\Gamma^\theta_{r\theta}=
\Gamma^\theta_{\theta r}=\frac{f'}{2f}~~,
\ee
while the Gaussian curvature of $\cG$ takes the form: 
\ben
\label{K_G}
K_\cG=-\frac{(\sqrt{f})^{''}}{\sqrt{f}}=\frac{(f')^2- 2f f''}{4f^2}=
-\frac{1}{2}\left(\frac{f'}{f}\right)'-\frac{1}{4}\left(\frac{f'}{f}\right)^2~~,
\een
where we use the notation $'=\frac{\dd}{\dd r}$. For such models, the
Noether generator \eqref{Xsol} takes the form:
\ben
\label{Xsolrot}
X=\frac{\Lambda}{\sqrt{a}}\pd_a+\left[Y^r - \frac{4}{a^{3/2}} \pd_r 
\Lambda\right]\pd_r+\left[ Y^\theta - \frac{4}{a^{3/2}f(r)}\pd_\theta 
\Lambda\right]\pd_\theta
\een
Note that one can have both purely visible symmetries, i.e. 
with $\Lambda = 0$ and $Y\neq0$, and purely hidden symmetries, 
i.e. with $\Lambda \neq 0$ and $Y = 0$. 
Let us write down the $Y$- and $\Lambda$-systems for the case 
of a rotationally-invariant metric $\cG$.

\paragraph{The $Y$-system.}

\noindent The $Y$-system \eqref{YSys} takes the form:
\beqan
\label{YSysElem}
&& \pd_r Y^r=0\nn\\
&& \pd_\theta Y^r+f\pd_r Y^\theta =0\nn\\
&& \pd_\theta Y^\theta+\frac{f'}{2f} Y^r=0\\
&& Y^r\pd_r V+Y^\theta\pd_\theta V=0~~,\nn
\eeqan
where the first three equations form the condition that $Y$ be a
Killing vector.

Let $\cK(\Sigma,\cG)$ denote the $\R$-vector space consisting of all
Killing vector fields of $(\Sigma,\cG)$. Since the metric $\cG$ is
rotationally-invariant, the Killing equation has the obvious solution:
\be
Y=\pd_\theta
\ee
and hence $\cK(\Sigma,\cG)$ contains the one-dimensional subspace
$\R\pd_\theta$. For a generic rotationally invariant metric, we have
$\cK(\Sigma,\cG)=\R\pd_\theta$, though in some cases\footnote{For
  example, we have $\dim \cK(\Sigma,\cG)=3$ when $(\Sigma,\cG)$ is
  the Poincar\'e disk. In that case, the first three equations of \eqref{YSysElem} 
  give: $Y^r = \tilde{c}_1 \sin \theta + \tilde{c}_2
\cos \theta$, $Y^{\theta} = \beta \coth (\beta r) \,(\tilde{c}_1
\cos \theta - \tilde{c}_2 \sin \theta) + \tilde{c}_3$, 
where $\tilde{c}_{1,2,3}=\const$.}  the space of Killing vectors may be
higher-dimensional.  In the generic case, the last equation in
\eqref{YSysElem} amounts to the condition that $V$ is
$\SO(2)$-invariant, i.e.:
\be
V=V(r)~~(\mathrm{indep.~of~}\theta)~~.
\ee

\paragraph{The $\Lambda$-system.}

\noindent The $\Lambda$-system \eqref{LambdaSys} takes the form:
\beqan
\label{LambdaSysElem}
&& \pd_r^2 \Lambda=\frac{3}{8}\Lambda\nn\\
&& \pd_r\pd_\theta\Lambda-\frac{f'}{2f}\pd_\theta \Lambda=0\nn\\
&& \pd_\theta^2 \Lambda+\frac{f'}{2}\pd_r\Lambda=\frac{3}{8}f \Lambda~~\\
&&\pd_r V\pd_r\Lambda+\frac{1}{f}\pd_\theta V\pd_\theta\Lambda=\frac{3}{4} V\Lambda~~,\nn
\eeqan
where the first three equations are equivalent with the Hesse condition
\eqref{HessCond}.

\begin{remark}
\label{FootnoteSepVar}
In reference \cite{ABL}, we studied two-field rotationally invariant
models with the separation of variables Ansatze: $X^a \!= \!A_1 (a) R_1
(r) \Theta_1 (\theta)$ and $X^i \!= \!A^i (a) R^i (r) \Theta^i (\theta)$,
assuming that each of the functions $A_1$, $R_1$, $\Theta_1$, $A^i$,
$R^i$ and $\Theta^i$ is non-constant. Comparing with \eqref{Xsol}, one
finds that these assumptions imply $Y = 0$ as well as:
\ben
\label{SepFunctions}
A_1=\frac{1}{\sqrt{a}}~~,~~ A^i=-\frac{1}{a^{3/2}}~~,~~R_1 (r) \Theta_1 (\theta) 
= \Lambda~~,~~R^i (r)\Theta^i (\theta)=4 \cG^{ij}\pd_j\Lambda~~,
\een
in agreement\footnote{Note that the overall constant in $A^i R^i
  \Theta^i$ is distributed differently between the factors $A^i$ and
  $R^i \Theta^i$ in \eqref{SepFunctions} when compared to reference \cite{ABL}.}
with \cite{ABL}. Substituting $\Lambda = R_1 (r) \Theta_1 (\theta)$ 
in the first equation of \eqref{LambdaSysElem} gives:
\ben
R_1 (r) = C_1 \cosh(\beta r) + C_2 \sinh (\beta r) \quad , \quad \beta^2 = \frac{3}{8}~~,
\een
which agrees with \cite[eq. (3.19)]{ABL}. Similarly, the second
equation of \eqref{LambdaSysElem} gives $\pd_r R_1 = \frac{f'}{2 f}
R_1$, where we used the assumption that $\Theta_1$ is not
constant. Substituting $\pd_r R_1 = \frac{f'}{2 f} R_1$ into the third
equation of \eqref{LambdaSysElem} gives:
\ben
8 f \,\Theta_1'' (\theta) + (2 f'^2 - 3 f^2) \,\Theta_1 (\theta) = 0 \,\,\, ,
\een
which coincides with \cite[eq. (3.37)]{ABL}.
\end{remark}

\subsection{Classification of weakly-Hessian models with 
rotationally-invariant scalar manifold metric}

\noindent The system formed by the first three equations in
\eqref{LambdaSysElem} is studied in detail in Appendix
\ref{app:SolHesse}; here we summarize the results of that analysis. As
before, let $\beta\eqdef \sqrt{\frac{3}{8}}$. For a
rotationally-invariant Riemannian 2-manifold $(\Sigma,\cG)$, it is
shown in Appendix \ref{app:SolHesse} that the first three equations of
the system \eqref{LambdaSysElem} admit solutions iff the metric $\cG$
has Gaussian curvature equal to $-\beta^2$, i.e. iff the rescaled
metric $G=\beta^2\cG$ has Gaussian curvature $-1$.  In particular, the
rescaled scalar manifold $(\Sigma,G)$ is a Hesse manifold in the sense
of Subsection \ref{subsec:LambdaV} iff it is a hyperbolic
surface. Since $(\Sigma,G)$ is rotationally-invariant, a well-known
result (which is summarized in Appendix \ref{app:elem}) implies that
$(\Sigma,G)$ must be {\em elementary hyperbolic}, i.e. that it is
isometric with the Poincar\'e disk $\mD$, with the hyperbolic
punctured disk $\mD^\ast$ or with a hyperbolic annulus $\mA(R)$ of
modulus $\mu=2\log R$ (where $R>1$). We refer the reader to Appendix
\ref{app:elem} and to reference \cite{elem} for the description of
elementary hyperbolic surfaces. We will use the notations $\mD_\beta$,
$\mD^\ast_\beta$ and $\mA_\beta(R)$ for the disk, punctured disk and
annulus endowed with the metric $\cG=\frac{1}{\beta^2}G$ of Gaussian
curvature equal to $-\beta^2$. Then the following statements hold (see
Appendix \ref{app:SolHesse}):
\begin{enumerate}
\itemsep 0.0em
\item If $(\Sigma,\cG)=\mD_\beta$, then we have:
\ben
\label{fmD}
f(r)=\frac{1}{\beta^2} \sinh^2(\beta r)~~
\een
and:
\ben
\label{metricmD}
\dd s^2_\cG=\dd r^2+\frac{1}{\beta^2}\sinh^2(\beta r) \,\dd \theta^2~,
\een
with $r\geq 0$. In this case, the general solution of the Hesse equation of
$(\Sigma,G)$ is given by:
\ben
\label{LambdamD}
\Lambda(r,\theta)=B_0\cosh(\beta r)+\sigma \sinh(\beta r)\cos (\theta-\theta_0)~~,
\een
where $B_0$ (denoted ${\hat B}_1$ in Appendix
\ref{app:SolHesse}) and $\theta_0$ are arbitrary real constants, while
$\sigma\geq 0$ is a non-negative constant (denoted $\frac{\zeta}{\beta}$ in Appendix 
\ref{app:SolHesse}).\footnote{The result of \cite{ABL}, namely 
$\Lambda (r,\theta)  =  (C_1 \sin \theta + C_2 \cos \theta)\sinh(\beta r)$, is obtained from \eqref{LambdamD} 
for $B_0 = 0$ and $C_1 = \sigma \sin \theta_0$\,, $C_2 = \sigma \cos \theta_0$.}

In particular, the space of Hesse functions on the hyperbolic disk is
three-dimensional. Let $(\rho,\theta)$ be Euclidean polar coordinates
on the disk, related to the normal polar coordinates
$(r,\theta)$ of the metric $\cG$ through (cf. \eqref{rrhoD}):
\ben
\label{rrhoDisk}
\rho\!=\!\tanh(\beta r/2)\in [0,1)\Longleftrightarrow r\!=\!\frac{2}{\beta}\arctanh(\rho)\!
=\!\frac{1}{\beta}\log \frac{1\!+\!\rho}{1\!-\!\rho}\in [0,+\infty)~.
\een
Then \eqref{metricmD} becomes:
\ben
\label{metricmDCart}
\dd s^2_\cG=\frac{4}{\beta^2(1-\rho^2)^2}(\dd \rho^2+\rho^2\dd\theta^2)
=\frac{4}{\beta^2(1-\rho^2)^2}(\dd x^2+\dd y^2)~~,
\een
where $x=\rho \cos \theta$ and $y=\rho \sin \theta$, while \eqref{LambdamD} takes the form:
\ben
\label{LambdamDEuc}
\Lambda(\rho,\theta)=\frac{B_0(1+\rho^2)+2\sigma\rho\cos(\theta-\theta_0)}{1-\rho^2}~~.
\een
Notice the relations:
\ben
\label{nrels}
\beta r=2\arctanh(\rho)=\log\frac{1+\rho}{1-\rho}~~\mathrm{i.e.}~~
\rho=\tanh\left(\frac{\beta r}{2}\right)=\frac{e^{\beta r}-1}{e^{\beta r}+1}~~,
\een
the second of which implies:
\ben
\label{nrels2}
\cosh(\beta r)=\frac{1+\rho^2}{1-\rho^2}~~,~~\sinh(\beta r)=\frac{2\rho}{1-\rho^2}~~.
\een

\item If $(\Sigma,\cG)=\mD^\ast_\beta$, then we have:
\ben
\label{fmDast}
f(r)=\frac{1}{(2\pi\beta)^2}e^{-2\beta r}~~
\een
and:
\ben
\label{metricmDast}
\dd s^2_\cG=\dd r^2+\frac{1}{(2\pi\beta)^2} e^{-2\beta r} \dd \theta^2~~,
\een
with $r\in \R$. In this case, the general solution of the Hesse
equation \eqref{HessCond} is given by:
\ben
\label{LambdamDast}
\Lambda(r)=C e^{-\beta r}~,
\een
where $C$ (denoted ${\hat B}$ in Appendix \ref{app:SolHesse}) is an
arbitrary constant.\footnote{The Noether condition is solved locally by
  $\Lambda(r,\theta) = (\tilde{C}_1 \theta + \tilde{C}_2) e^{- \beta r}$
  \cite{ABL}. Requiring $\Lambda$ to be globally defined on the scalar
  manifold implies that one must set $\tilde{C}_1 = 0$, in agreement
  with \eqref{LambdamDast}.} In particular, the space of Hesse
functions on the hyperbolic punctured disk is one-dimensional. Let
$(\rho,\theta)$ be Euclidean polar coordinates on the punctured disk,
related to the normal polar coordinates $(r,\theta)$ of the metric
$\cG$ through (cf. eqs. \eqref{rrhoDast}):
\ben
\label{rrhopDisk}
\rho=e^{-2\pi e^{\beta r}}\in (0,1)\Longleftrightarrow r=\frac{1}{\beta} 
\log\left(\frac{|\log \rho|}{2\pi}\right)\in (-\infty,\infty)~~.
\een
Then \eqref{metricmDast} becomes:
\ben
\label{metricmDastCart}
\dd s^2_\cG=\frac{1}{\beta^2(\rho\log\rho)^2}(\dd \rho^2+\rho^2\dd\theta^2)
=\frac{1}{\beta^2(\rho\log\rho)^2}(\dd x^2+\dd y^2)~~,
\een
while \eqref{LambdamDast} takes the form:
\ben
\label{LambdamDastCart}
\Lambda(\rho)=\frac{2\pi C}{|\log \rho|}~~.
\een

\item If $(\Sigma,\cG)=\mA_\beta(R)$, then we have:
\ben
\label{fmmA}
f(r)=\frac{\ell^2}{(2\pi\beta)^2}\cosh^2(\beta r)~~
\een
and:
\ben
\label{metricmA}  
\dd s^2_\cG=\dd r^2+\frac{\ell^2}{(2\pi\beta)^2} \cosh^2(\beta r) \,\dd \theta^2~~,
\een
where $r\in \R$ and $\ell>0$ is given in \eqref{ell}. In this case,
the general solution of the Hesse equation is:
\ben
\label{LambdamA}
\Lambda(r)=C\sinh(\beta r)~~,
\een
where $C$ (denoted ${\hat B}_2$ in Appendix \ref{app:SolHesse}) is an
arbitrary constant.\footnote{For hyperbolic annuli, the Hesse equation
  is solved locally by $\Lambda(r,\theta) \!=\! \left[ \hat{C}_1 \cosh (C_R \theta)
    + \hat{C}_2 \sinh (C_R \theta) \right] \cosh (\beta r) + \hat{C}_3
  \sinh (\beta r)$ with $C_R = \frac{\pi}{2 \log R}$.  When
  $\pd_{\theta} \Lambda \neq 0$, one is left with the term $\cosh
  (\beta r)$ (see \cite{ABL}). Requiring that the solution is globally
  defined on the scalar manifold forces the choice
  $\hat{C}_1=\hat{C}_2 = 0$ in the local solutions of loc. cit.}  In
particular, the space of Hesse functions on any hyperbolic annulus is
one-dimensional. Let $(\rho,\theta)$ be Euclidean polar coordinates on
the annulus, related to the normal polar coordinates $(r,\theta)$ of
the metric $\cG$ through (cf. eqs. \eqref{rrhoA}):
\ben
\label{rrhopAnnulus}
\rho= e^{-\frac{\mu}{\pi}\arccos\left[\frac{1}{\cosh(\beta r)}\right]}
\Longleftrightarrow |r|=\frac{1}{\beta}\arccosh\left[\frac{1}{\cos\left(\frac{\pi}{\mu}|\log\rho|\right)}\right]~~.
\een
Then \eqref{metricmA} becomes:
\ben
\label{metricmACart}
\dd s^2_\cG\!=\!\left(\frac{\pi}{2\beta\log R}\right)^2 \!\frac{\dd \rho^2\!+
\!\rho^2\dd\theta^2}{\left[\rho\cos\left(\frac{\pi\log\rho}{2\log R}\right)\right]^2}\!=\!
\left(\frac{\pi}{2\beta\log R}\right)^2\! \frac{\dd x^2+\dd y^2}
{\left[\rho\cos\left(\frac{\pi\log\rho}{2\log R}\right)\right]^2}~,
\een
while \eqref{LambdamA} takes the following form:
\ben
\label{LambdamACart}
\Lambda(\rho)=C \tan\left(\frac{\pi}{\mu}\log\rho\right)~~.
\een
\end{enumerate}

\section{Hessian models for the hyperbolic disk}
\label{sec:Disk}

\noindent In this section, we show that the space $\cS(\mD)$ of Hesse
functions on the hyperbolic disk identifies naturally with the
3-dimensional Minkowski space $\R^{1,2}$ such that the natural action
of the orientation-preserving isometry group of $\mD$ on such
functions identifies with the fundamental action of the group of
proper and orthochronous Lorentz transformations in 3 dimensions. The
identification follows from the fact that the general Hesse function
on the hyperbolic disk is a linear combination of the components of
the classical Weierstrass map and hence the classical Weierstrass
coordinates of $\mD$ form a basis for the space of Hesse
functions. This leads to a description of Hesse functions on the
hyperbolic disk in terms of three-dimensional Minkowski geometry and
allows for a natural classification of such functions into functions
of timelike, spacelike and lightlike type.  We also discuss the level
sets and critical points of such functions, showing that they behave
quite differently in each of the three cases. For each type, we show
that the gradient flow of a Hesse function can be described explicitly
in certain classical coordinate systems on the hyperbolic
disk. Finally, we combine these results and the method of
characteristics to extract the explicit form of the most general
scalar potential which solves the $\Lambda$-$V$ equation, thus
classifying all Hessian two-field cosmological models whose rescaled
scalar manifold is a hyperbolic disk. We find that such scalar
potentials admit a natural description in terms of three-dimensional
Minkowski geometry. The results of this section are summarized in
Subsection \ref{subsec:DiskSummary}, which the reader may consult
first. Throughout this section, $G$ denotes the Poincar\'e metric
(which has Gaussian curvature equal to $-1$), while $\cG$ denotes the
physically-relevant metric (which has Gaussian curvature equal to
$-\beta^2=-3/8$).

\subsection{The space of Hesse functions}

\noindent We start by studying the space of Hesse functions on the hyperbolic disk
$\mD=(\rD,G)$. 

\paragraph{The Weierstrass basis.}
The general Hesse function \eqref{LambdamD} of $\mD$ can be written as:
\ben
\label{LambdamDCart}
\Lambda\!=\!B_0\cosh(\beta r)\!+\!B_1\sinh(\beta r)\cos\theta\!+
\!B_2\sinh(\beta r)\sin \theta
=\!\frac{B_0(1\!+\!\rho^2)\!+\!2B_1 x \!+\!2 B_2 y}{1-\rho^2}~,
\een
where (see equation \eqref{B12def}):
\ben
\label{Bdef}
B_1\eqdef \sigma\cos\theta_0~~,~~ B_2\eqdef \sigma\sin\theta_0~~.
\een
Here $x=\rho\cos\theta$ and $y=\rho\sin\theta$ are Euclidean Cartesian
coordinates on the disk (with $\rho=\sqrt{x^2+y^2}$) while
$(r,\theta)$ are normal polar coordinates for the physically-relevant
metric $\cG=\frac{1}{\beta^2}G$; for the relation between $\rho$ and $r$, 
see \eqref{rrhoD}. Relation \eqref{LambdamDCart} shows
that the functions:
\ben
\label{Ximu}
\Xi^0\!\eqdef \!\frac{1+\rho^2}{1-\rho^2}\!=\!\cosh\beta r~,
~~\Xi^1\!\eqdef\! \frac{2x}{1-\rho^2}\!=\!\sinh(\beta r)\cos\theta~,
~~\Xi^2\!\eqdef\! \frac{2 y}{1-\rho^2}\!=\! \sinh(\beta r)\sin\theta
\een
form a basis of the linear space $\cS(\mD)$ of smooth solutions to the
Hesse equation. The fundamental solutions \eqref{Ximu} coincide with
the classical ``Weierstrass coordinates'' of $\mD$, i.e. with the
components of the {\em Weierstrass map} $\Xi:\rD\rightarrow \R^3$ (see
Appendix \ref{app:iso}) which realizes the hyperbolic disk as the
future sheet:
\be
S^+\eqdef \big\{X=(X^0,X^1,X^2)\,|\,X^0=\sqrt{1+(X^1)^2+(X^2)^2}\,\big\}
\ee
of the unit hyperboloid in the 3-dimensional Minkowski space
$\R^{1,2}=(\R,(~,~))$ (see Figure \ref{fig:Weierstrass}).
Here:
\ben
\label{Mink}
(X,X')\eqdef X^0{X'}^0-X^1 {X'}^1-X^2{X'}^2~~~~\forall X,X'\in \R^3
\een
is the Minkowski pairing of signature $(1,2)$, whose coefficients we
denote by $\eta_{\mu\nu}$:
\be
(\eta_{\mu\nu})_{\mu,\nu=0,\ldots 2}=\diag(1,-1,-1)~~
\ee
and which we use to raise and lower indices.

\begin{figure}[H]
\centering \centering
\vskip 1em
\includegraphics[width=0.4\linewidth]{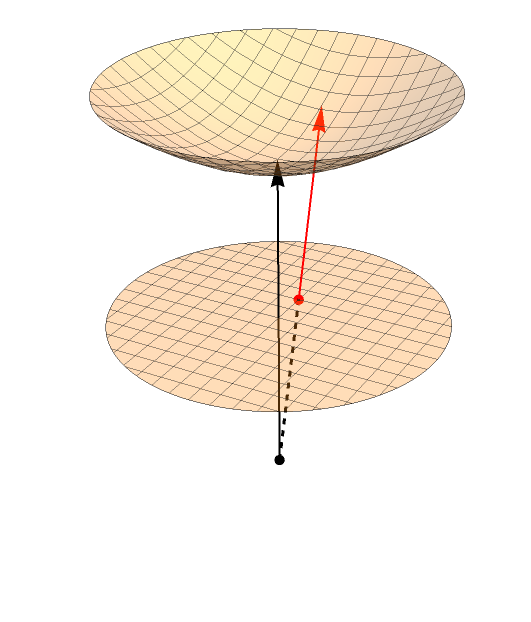}
\vskip -2.5em
\caption{The Weierstrass map $\Xi:\rD\rightarrow S^+$ coincides with
  the projection of $\rD$ from the point $(-1,0,0)$ of
  three-dimensional Minkowski space, when the Poincar\'e disk is
  placed in the plane $X^0=0$. Notice that the conformal boundary of
  $\mD$ is mapped to the circle at infinity of $S^+$. When $u\in \mD$
  approaches the conformal boundary, the 3-vector $\Xi(u)$ (shown in
  red) becomes lightlike.}
\label{fig:Weierstrass}
\end{figure}

\noindent The Weierstrass coordinates of any point $u\in D$ satisfy:
\be
\left( \Xi^0(u) \right)^{\!2}-\left( \Xi^1(u) \right)^{\!2}-\left( \Xi^2(u) \right)^{\!2}=1~~\mathrm{and}~~\Xi^0(u)>0
\ee
and we have:
\be
\Xi(u)=\Xi^\mu(u) E_\mu~~,~~B=B^\mu E_\mu~~,
\ee
where $E_0\eqdef (1,0,0)$, $E_1\eqdef (0,1,0)$ and $E_2\eqdef
(0,0,1)$. 

\paragraph{The 3-vector parameterization.}

The general Hesse function \eqref{LambdamDCart} reads:
\ben
\label{LambdaXi}
\Lambda_B(u)=B_\mu \Xi^\mu(u)\!=\!\eta_{\mu\nu}B^\mu \Xi^\nu(u)
\!=\!(B, \Xi(u))\!=\!B^0\Xi^0(u)-B^1 \Xi^1(u)-B^2\Xi^2(u)~~,
\een
where we defined $B^\mu\eqdef \eta^{\mu\nu}B_\nu$ and we combined the
constants $B^0=B_0$, $B^1=-B_1$ and $B^2=-B_2$ into the 3-vector:
\be
B\eqdef (B^0,B^1,B^2)=B^\mu E_\mu \in \R^3~~.
\ee
Notice the relation $\Xi_\mu(u)=(\Xi(u),E_\mu)$. Since the Weierstrass
coordinates form a basis of the space of Hesse functions, the linear
map $\bLambda:\R^3\stackrel{\sim}{\rightarrow} \cS(\mD)$ defined
through:
\ben
\label{bLambda}
\bLambda(B)\eqdef \Lambda_B~~,~~\forall B\in \R^3~~
\een
is an isomorphism of vector spaces from $\R^3$ to the space $\cS(\mD)$.

\paragraph{Action of orientation-preserving isometries.}

Since the Hesse equation is invariant under isometries of the scalar
manifold, the group $\PSU(1,1)\simeq \Iso_+(\mD)$ of
orientation-preserving isometries of $\mD$ acts linearly on the space
$\cS(\mD)$ of Hesse functions through the representation $\cH$ defined
through:
\be
\cH(U)(\Lambda)\eqdef \Lambda\circ \psi_U^{-1}~,~~\forall U\in \PSU(1,1)
\ee
i.e.:
\be
\cH(U)(\Lambda)(u)=\Lambda (\psi_U^{-1}(u))~~,~~\forall U\in \PSU(1,1)~,~\forall u\in \rD~~.
\ee
Here $\psi_U\in \Iso_+(\mD)$ is the orientation-preserving isometry of
$\mD$ corresponding to an element $U$ of $\PSU(1,1)$ (see Appendix
\ref{app:iso}). The equivariance property \eqref{HypAction} of
the Weierstrass map gives:
\be
\Xi_\mu(\psi_U(u))=(\Ad_0(U)(\Xi(u)),E_\mu)=(\Xi(u),\Ad_0(U^{-1})(E_\mu))~~.
\ee
while equation \eqref{LambdaXi} implies:
\be
\Lambda_B(\psi_{U^{-1}}(u))=(B, \bAd_0(U^{-1})(\Xi(u)))=(\bAd_0(U)(B),\Xi(u))
=\Lambda_{\bAd_0(U)(B)}(u)~~.
\ee
This gives: 
\ben
\label{LambdaIsoTf}
\cH(U)(\Lambda_B)=\Lambda_{\bAd_0(U)(B)}~~,~~\forall U\in \PSU(1,1)~~,
~~\forall B\in \R^3~~,
\een
i.e.:
\be
\cH(U)\circ \bLambda=\bLambda\circ \bAd_0(U)~~,~~\forall U\in \PSU(1,1)~~,
\ee
showing that the linear isomorphism \eqref{bLambda} is an equivalence
of representations between $\cH$ and ${\overline \Ad}_0$. As recalled
in Appendix \ref{app:iso}, the representation $\bAd_0$ (which is
equivalent with the adjoint representation of $\PSU(1,1)$) preserves
the Minkowski pairing \eqref{Mink}. In fact, this representation
defines an isomorphism of groups
$\bAd_0:\PSU(1,1)\stackrel{\sim}{\longrightarrow} \SO_o(1,2)$, where
$\SO_o(1,2)$ denotes the connected component of the identity of the
Lorentz group, i.e. the group of proper and orthochronous Lorentz
transformations in three dimensions.

\begin{definition}
The Hesse function $\Lambda_B$ on the hyperbolic disk is called {\em
  spacelike}, {\em timelike} or {\em lightlike} if its parameter
3-vector $B\in\R^3$ is spacelike, timelike or lightlike,
respectively. Similarly, $\Lambda_B$ is called future (resp. past)
timelike or lightlike if it is timelike or lightlike and $B^0>0$
(respectively $B^0<0$).
\end{definition}

\subsection{Degenerate and non-degenerate Hesse functions}

\begin{definition}
A non-trivial Hesse function $\Lambda_B$ is called {\em
  non-degenerate} if $B^0\neq 0$ and {\em degenerate} if $B^0=0$. 
\end{definition}

\noindent Notice that a degenerate Hesse function is necessarily spacelike.

\paragraph{Rescaling non-degenerate Hesse functions.}
Recall from \eqref{Bdef} that $\sigma\!=\!\sqrt{B_1^2\!+\!B_2^2}$. 
When $B^0\neq 0$, we define:
\ben
\label{DeltaDef}
\Delta=\frac{\sigma}{B^0}=\frac{\sqrt{B_1^2+B_2^2}}{B^0}~~,
~~b^1\eqdef\frac{B^1}{B^0}=-\Delta\cos\theta_0~,
~~b^2\eqdef\frac{B^2}{B^0}=-\Delta\sin\theta_0~~
\een
and $\vec{b}\eqdef (b^1,b^2)$, so that
$B=B^0(1,b^1,b^2)=B^0(1,\vec{b})$ and $b_1^2+b_2^2=\Delta^2$. This
allows us to write non-degenerate Hesse functions as:
\be
\Lambda_B=B^0 \lambda_{\vec{b}}~~(\mathrm{when}~~B^0\neq 0)~,~
\ee
with
\be
\lambda_{\vec{b}}\eqdef\frac{1+\rho^2-2b^1 x-2b^2 y}{1-\rho^2}=
\frac{1+\rho^2+2\Delta \rho \cos(\theta-\theta_0)}{1-\rho^2}\eqdef
\lambda_{\Delta,\theta_0}~~.
\ee
In normal polar coordinates $(r,\theta)$ for the metric $\cG$ we have:
\ben
\label{lambda}
\lambda_{\Delta,\theta_0}= \cosh(\beta r)+\Delta \sinh(\beta r) \cos (\theta-\theta_0)~~.
\een
Notice that a non-degenerate Hesse function is:
\begin{itemize}
\item timelike, iff $|\Delta|<1$.
\item spacelike, iff $|\Delta|>1$.
\item lightlike, iff $|\Delta|=1$, i.e. if $\Delta=+1$ (future
  lightlike) or $\Delta=-1$ (past lightlike).
\end{itemize}
When $\sigma\neq 0$, we have $\sign(B_0)=\sign(\Delta)$. The shape of
non-degenerate Hesse functions on $\mD$ is illustrated in
Figures \ref{fig:LambdaTimelike}, \ref{fig:LambdaSpacelike} and
\ref{fig:LambdaLightlike} for the case $B^0=1$
(i.e. $\Delta=\sigma\geq 0$) with $\theta_0=-\pi/2$, which gives
$B=(1,0,\Delta)$ and:
\be
\Lambda_B=\Lambda_{1, 0, \Delta}=\frac{1+\rho^2-2\Delta y}{1-\rho^2}~~.
\ee

\begin{figure}[H]
\centering\centering\centering
\begin{minipage}{.31\textwidth}
\centering \includegraphics[width=1\linewidth]{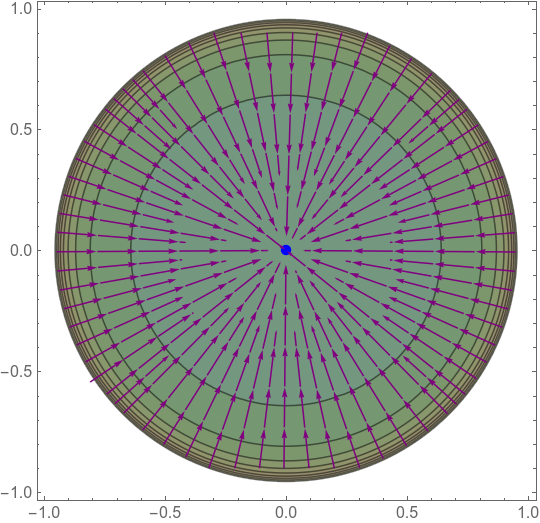}
\subcaption{$B=(1,0,0)$}
\label{fig:LambdaTimelike0}
\end{minipage}\hfill
\begin{minipage}{.31\textwidth}
\centering \includegraphics[width=1\linewidth]{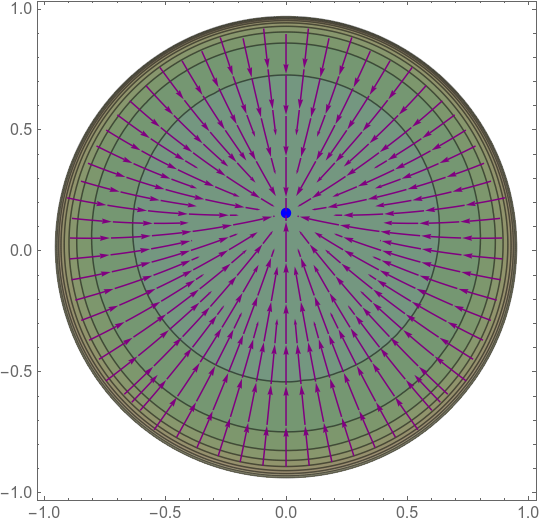}
\vskip 0.45em
\subcaption{$B=(1,0,0.3)$}
\label{fig:LambdaTimelike0.3}
\end{minipage}\hfill
\begin{minipage}{.31\textwidth}
\centering \includegraphics[width=1\linewidth]{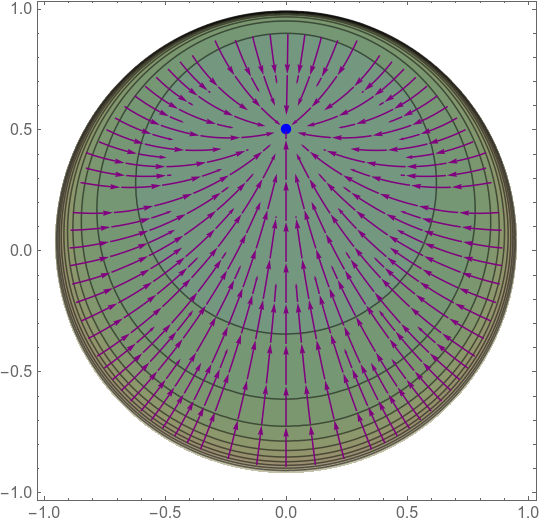}
\vskip 0.45em
\subcaption{$B=(1,0,0.8)$}
\label{fig:LambdaTimelike0.8}
\end{minipage}\hfill
\caption{Contour plot of $\Lambda:=\Lambda_{1, 0, \Delta}$ for the
  non-degenerate timelike case ($\Delta<1$), where the gradient flow
  of $\Lambda$ is indicated by purple arrows. The values of $\Lambda$
  decrease from lightest brown to darkest green. The critical point of
  $\Lambda$ is shown in blue. In this case, the level sets are
  hyperbolic circles centered at the critical point.  From left to
  right, the figure shows the cases $\Delta=0,0.3,0.8$.}
\label{fig:LambdaTimelike}
\end{figure}

\begin{figure}[H]
\centering
\begin{minipage}{.31\textwidth}\hfill
\centering \includegraphics[width=1\linewidth]{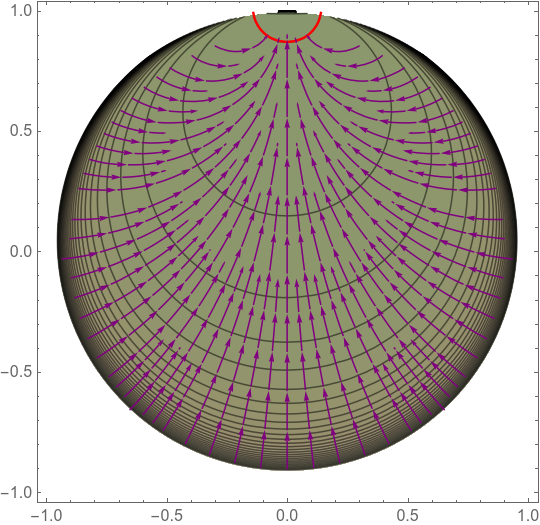}
\subcaption{$B=(1,0,1.01)$}
\label{fig:LambdaSpacelike1.01}
\end{minipage}\hfill
\begin{minipage}{.31\textwidth}
\centering \includegraphics[width=1\linewidth]{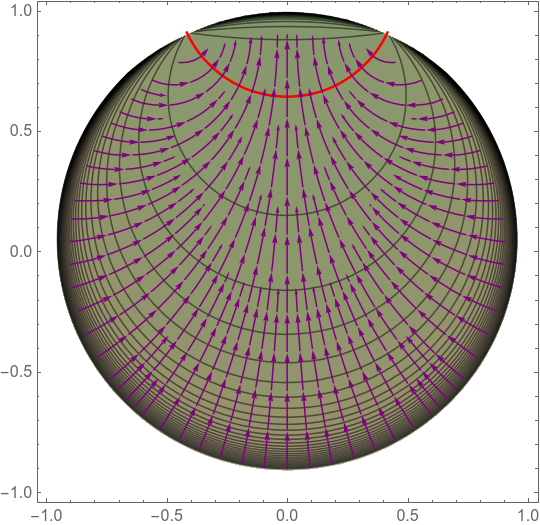}
\subcaption{$B=(1,0,1.1)$}
\label{fig:LambdaSpacelike1.1}
\end{minipage}\hfill
\begin{minipage}{.31\textwidth}
\centering \includegraphics[width=1\linewidth]{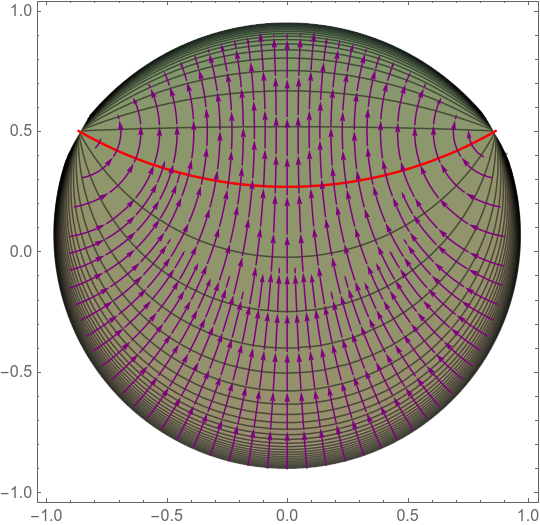}
\subcaption{$B=(1,0,2)$}
\label{fig:LambdaSpacelike2}
\end{minipage}\hfill
\caption{Contour plot of $\Lambda:=\Lambda_{1, 0, \Delta}$ for the
  non-degenerate spacelike case, where the gradient flow of $\Lambda$
  is indicated by purple arrows. The values of $\Lambda$ decrease from
  lightest brown to darkest green. In this case, $\Lambda$ has no
  critical point but vanishes along the curve shown in red. In this
  case, the level sets are hypercycles with axis given by the
  vanishing locus of $\Lambda$. From left to right, the figure shows
  the cases $\Delta=1.01, 1.1,2$.}
\label{fig:LambdaSpacelike}
\end{figure}

\begin{figure}[H]
\centering \centering
\includegraphics[width=0.45\linewidth]{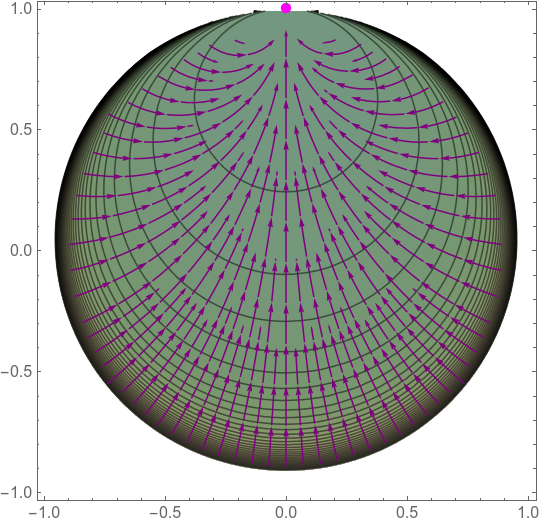}
\caption{Contour plot of $\Lambda_{1, 0, 1}$ (non-degenerate lightlike
  case with $\Delta=1$), where the gradient flow of $\Lambda$ is
  indicated by purple arrows. The values of $\Lambda$ decrease from
  lightest brown to darkest green. The point where $\Lambda$ tends to
  zero on the conformal boundary of the hyperbolic disk is shown in
  magenta. In this case, the level sets are horocycles centered at
  this ideal point.}
\label{fig:LambdaLightlike}
\end{figure}

\paragraph{Rescaling degenerate Hesse functions.}
Non-trivial but degenerate (i.e. with $B^0 = 0$) Hesse functions have the form:
\be
\Lambda_{0,B^1,B^2}=\sigma \mu_{\theta_0}~~,
\ee
where $B^1=-\sigma\cos\theta_0$, $B^2=-\sigma \sin\theta_0$ with $\sigma
=\sqrt{B_1^2+B_2^2}>0$ and:
\ben
\label{mu}
\mu_{\theta_0}\eqdef\frac{2 x\cos \theta_0+2y \sin \theta_0}{1-\rho^2}
=\frac{2\rho }{1-\rho^2}\cos(\theta-\theta_0)=\sinh(\beta r)\cos (\theta-\theta_0)~~.
\een
See Figure \ref{fig:LambdaDeg} for a contour plot of the function:
\ben
\label{mucan}
\Lambda_{0,0,1}=\mu_{-\pi/2}=-\frac{2y}{1-\rho^2}~~.
\een

\begin{remark}
When $\sigma\neq 0$, we have $\Delta\neq 0$ and
$B^0=\frac{\sigma}{\Delta}$. In this case, we can write
$\Lambda_B=\frac{\sigma}{\Delta}\lambda_{\Delta,\theta_0}$ and we
have:
\be
\lim_{\Delta\rightarrow \pm \infty}\frac{\lambda_{\Delta,\theta_0}(u)}{\Delta}=
\mu_{\theta_0}(u)\Longrightarrow \lim_{\Delta\rightarrow \pm \infty}\Lambda_B(u)
=\sigma \mu_{\theta_0}(u)~~.
\ee
Hence a non-degenerate Hesse function with $\sigma\neq 0$ point-wisely
approximates the degenerate Hesse function with the same $\theta_0$ in
the limits $\Delta\rightarrow \pm \infty$.
\end{remark}

\vskip 0.5em

\begin{remark}
It is easy to see that a Hesse function is separated in the
coordinates $(r,\theta)$ or $(\rho,\theta)$ iff it is either
degenerate or non-degenerate with $\Delta=0$. Hence $\Lambda_B$
separates in these coordinates only for $\Delta=0$ or in the limits
$\Delta\rightarrow \pm \infty$.
\end{remark}

\begin{figure}[H]
\centering
\begin{minipage}{.5\textwidth}
\includegraphics[width=\linewidth]{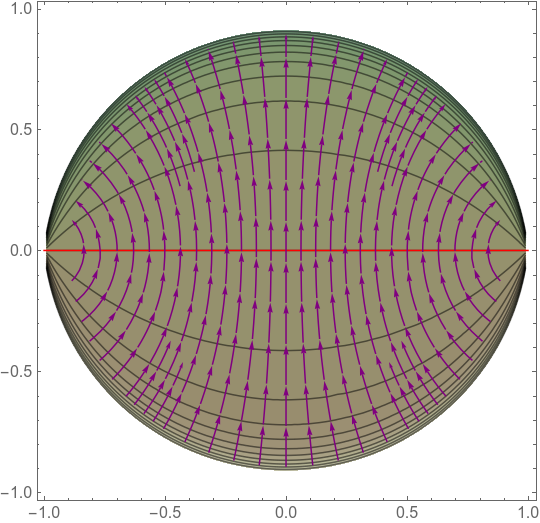}
\vskip 0.45em \subcaption{$B=(0,0,1)$.}
\end{minipage}\hfill
\caption{Contour plot of the degenerate spacelike Hesse function
  $\Lambda_{0,0,1}$, where the gradient flow of $\Lambda$
  is indicated by purple arrows. The values of $\Lambda$ decrease from lightest
  brown to darkest green. The vanishing locus of $\Lambda$ is the
  horizontal segment shown in red. The level sets are hypercycles with
  axis given by the vanishing locus.}
\label{fig:LambdaDeg}
\end{figure}

\subsection{Critical points of Hesse functions}

\begin{definition} 
A non-trivial Hesse function $\Lambda$ on the hyperbolic disk is called
{\em critical} if it has at least one critical point, and {\em
  non-critical} if it has no critical points.
\end{definition}

\begin{prop}
A non-trivial Hesse function $\Lambda_B$ on the hyperbolic
disk is critical iff it is timelike (and hence non-degenerate). In
this case, $\Lambda_B$ has exactly one critical point on $\rD$, namely:
\ben
\label{uc}
u_c=\twopartdef{0}{\Delta=0}{\frac{\sqrt{1-\Delta^2}-1}{\Delta} 
e^{\i \theta_0}}{0<|\Delta|<1}~~
\een
and the critical value of $\Lambda_B$ is given by: 
\ben
\label{LambdaCritVal}
\Lambda_c\eqdef \Lambda(u_c)=B_0\frac{1+\Delta^2}{\sqrt{1-\Delta^2}}~~.
\een
Moreover, $u_c$ is an absolute minimum when $B_0>0$ (i.e. when
$\Lambda_B$ is future-timelike) and an absolute maximum when $B_0<0$
(i.e. when $\Lambda_B$ is past-timelike).
\end{prop}

\begin{proof}
It is easy to see that a non-trivial degenerate Hesse function has no
critical points in $\rD$. If
$\Lambda=\Lambda_B=B^0\lambda_{\Delta,\theta_0}$ is a non-degenerate
Hesse function, then a counterclockwise rotation of the coordinates by
an angle $\theta_0$ allows us to assume, without loss of generality,
that $\theta_0=0$. Hence it suffices to study the critical points of
the function:
\be
\lambda:=\lambda_{\Delta, 0}=\frac{1+x^2+y^2+2 \Delta x}{1-x^2-y^2}~~.
\ee
The condition $(\dd\lambda)(x,y)=0$ amounts to
the system:
\beqan
\label{LambdaPCrit}
&& \Delta(1+x^2-y^2)+2x=0\nn\\
&& y(\Delta x+1)=0~~.
\eeqan
Multiplying the first equation by $y$ and the second equation by $x$
and subtracting the two gives:
\be
\Delta y(1-x^2-y^2)=0~~,
\ee
which implies $\Delta y=0$ since $1-x^2-y^2>0$ for all points $u=x+\i
y\in \rD$. Using this, the second equation of \eqref{LambdaPCrit} reduces to
$y=0$, while the first equation becomes:
\ben
\label{Ex}
\Delta x^2+2x+\Delta=0~.
\een
Distinguish the cases:

\begin{enumerate}[1.]
\itemsep 0.0em
\item $\Delta=0$. Then $\Lambda$ is timelike, equation \eqref{Ex} gives
  $x=0$ and the only critical point of $\Lambda$ is $u_c=0$.
\item $\Delta\neq 0$ (hence $\Lambda$ is timelike or lightlike). Then \eqref{Ex}
  has real solutions iff $|\Delta|\leq 1$, in which case the two solutions
  are $x_\pm =\frac{-1\pm \sqrt{1-\Delta^2}}{\Delta}$. The case
  $\Delta=1$ leads to $x_+=x_-=1$, which is forbidden since the points
  $(x,y)=(1,0),(0,1)$ do not lie in the interior of the unit disk. Hence
  critical points inside $\rD$ can exist only if $|\Delta|<1$,
  i.e. when $\Lambda$ is timelike. In this case, we have
  $|x_+|<1<|x_-|$, so the point $(x,y)=(x_-,0)$ lies outside $\rD$,
  while $(x,y)=(x_+,0)$ lies inside $\rD$. We conclude that $\Lambda$
  is critical iff it is timelike, in which case the only critical
  point is at $u_c=\frac{\sqrt{1-\Delta^2}-1}{\Delta}$.
\end{enumerate}

\noindent Since $\Lambda$ satisfies the Hesse equation
$\Hess_{\cG}(\Lambda)=\frac{3}{8}\Lambda \cG$, it follows that its
Hessian at $u_c$ is positive definite when $B_0>0$ and negative
definite when $B_0<0$. Thus $u_c$ is a local minimum or maximum
according to the sign of $B_0$. Substituting $u_c$ in the expression
for $\Lambda$ gives \eqref{LambdaCritVal}.

The conclusion of the theorem now follows by performing a clockwise
rotation of the coordinates by $\theta_0$, in order to restore the
$\theta_0$-dependence in the position of the critical point.
\end{proof}

\subsection{Level sets of Hesse functions}

\noindent Relation \eqref{LambdamDCart} shows that the $\lambda$-level
set $\{u\in \rD\,|\,\Lambda_B(u)=\lambda\}$ of a non-trivial Hesse
function $\Lambda_B$ has the equation:
\ben
\label{LevelSet}
(B_0+\lambda)(x^2+y^2)+2B_1 x+2B_2 y+B_0-\lambda=0~~.
\een
We distinguish the cases:
\begin{enumerate}[1.]
  \itemsep 0.0em
\item $\lambda\neq -B_0$. Then \eqref{LevelSet} takes the form:
\ben
\label{circ}
(x-x_0)^2+(y-y_0)^2=\frac{\lambda^2-(B,B)}{(B_0+\lambda)^2}~~,
\een
where $x_0\eqdef -\frac{B_1}{B_0+\lambda}$ and
$y_0=-\frac{B_2}{B_0+\lambda}$. This equation has solutions only for
$\lambda^2\geq (B,B)$, in which case it describes a Euclidean circle
of radius $R\eqdef \frac{\sqrt{\lambda^2-(B,B)}}{|B_0+\lambda|}$
centered at the point:
\be
u_0=x_0+\i y_0=\frac{B^1+\i B^2}{B^0+\lambda}~~,
\ee
which is reduced to this point for $\lambda^2=(B,B)$. The radius $R$
tends to infinity for $\lambda=-B_0$, in which case the circle
degenerates to a line.

\item $\lambda=-B_0$. Then \eqref{LevelSet} takes the form:
\be
B_1 x+B_2 y=-B_0~~.
\ee
Since $\Lambda_B$ is non-trivial, existence of solutions to this
equation requires $B_1^2+B_2^2> 0$ i.e. $\sigma>0$, in which case the
equation describes a line in the $u$ plane which passes through the
points $u_1=-\frac{B_0}{B_1}$ and $u_2=-\i\frac{B_0}{B_2}$ of the one
point compactification of this plane. The relations $B_1=\sigma
\cos\theta_0$ and $B_2=\sigma \sin \theta_0$ give:
\be
u_1=-\frac{1}{\Delta\cos\theta_0}~~,~~u_2=-\frac{\i}{\Delta \sin\theta_0}~~.
\ee
and bring the equation to the form:
\be
\rho\,\cos(\theta-\theta_0)=-\frac{B_0}{\sigma}=-\frac{1}{\Delta}~~,
\ee
where the case $B_0=0$ is included for $\Delta\rightarrow \pm \infty$. 
\end{enumerate}

\noindent One can show that the Euclidean circles defined by equation
\eqref{circ} are contained inside $\rD$ iff $B$ is lightlike and that
they meet the Euclidean circle of radius one at one point when $B$ is
timelike and at two points when $B$ is spacelike. It follows that the
level sets of a timelike Hesse function are hyperbolic circles, while
they are horocycles for a lightlike Hesse function and hypercycles for
a spacelike Hesse function. These facts also follow more directly from
the similar statements satisfied by the level sets of the three
canonical Hesse functions discussed in Subsection \ref{subsec:PotD},
upon acting on those canonical forms with an orientation-preserving
isometry of $\mD$.

\paragraph{The vanishing locus of a Hesse function.}

Let:
\be
Z(\Lambda)\eqdef \{u\in \rD\,|\,\Lambda(u)=0\}
\ee
denote the set of zeroes of the Hesse function $\Lambda$. The proof of
the following statement follows by inspection of equation \eqref{lambda}.

\begin{prop}
A non-trivial Hesse function $\Lambda=\Lambda_B$ on the hyperbolic
disk has zeroes iff it is spacelike, lightlike or degenerate. In this
case, the vanishing locus of $\Lambda$ is a curve given by the
following quadratic equation in Euclidean Cartesian coordinates on
$\rD$:
\ben
\label{Zeq}
B_0(1+x^2+y^2)+2B_1 x+2 B_2 y=0~~.
\een
Moreover:
\begin{itemize}
\itemsep 0.0em
\item When $\Lambda$ is non-degenerate ($B_0\neq 0$), 
equation \eqref{Zeq} is equivalent with:
\ben
(x+b_1)^2+(y+b_2)^2=b_1^2+b_2^2-1 (\geq 0)~~,
\een
where $b_1=\frac{B_1}{B^0}$ and $b_2=\frac{B_2}{B^0}$.  When $\Lambda$
is non-degenerate spacelike ($b_1^2+b_2^2>1$), the vanishing locus is
a hypercycle which coincides with the intersection of $\rD$ with a
Euclidean circle of radius $\sqrt{b_1^2+b_2^2}$ centered at the point
$u_0=-b_1-\i b_2$, which lies outside of $\rD$. When $\Lambda$ is
non-degenerate lightlike ($b_1^2+b_2^2=1$), the vanishing locus
degenerates to the single point $u_0$, which lies on the conformal
boundary of $\mD$ (the unit Euclidean circle). In this case, the
function $\Lambda$ tends to zero at this point of the conformal
boundary.
\item When $\Lambda$ is degenerate ($B_0=0$) and hence spacelike, the
  vanishing locus coincides with the intersection of $\rD$ with the
  line obtained by rotating the $y$ axis counterclockwise by an angle
  equal to $\theta_0$. 
\end{itemize}
\end{prop}

\subsection{The scalar potential determined by a Hesse function}
\label{subsec:PotD}

\noindent In this subsection, we solve the $\Lambda$-$V$ equation
\eqref{LambdaV} for a general Hesse function $\Lambda\in \cS(\mD)$. We
shall do so by combining representation-theoretic arguments with the
method of characteristics. First, we notice that acting with an
appropriate element $U$ of the group $\PSU(1,1)$ (and possibly
rescaling by a constant) allows us to bring any non-trivial Hesse
function $\Lambda$ to one of three specific canonical forms, depending
on whether $\Lambda$ is timelike, spacelike or lightlike. We next
determine the scalar potential $V$ by solving the $\Lambda$-$V$
equation for each of these three canonical choices of
$\Lambda$. Finally, we act with the inverse of $U$ in order to recover
the form of $V$ for a general Hesse function of timelike, spacelike or
lightlike type. Equivalently, we write the scalar potentials for the
three canonical cases in manifestly Lorentz-invariant form, which allows us
to extend them to general Hesse functions of lightlike, spacelike and
timelike type.

\paragraph{Reduction to canonical cases.}

Let $V_B$ be the general solution of equation \eqref{LambdaV}, where
$\Lambda=\Lambda_B$ is a non-trivial Hesse function of $\mD$.
Relations \eqref{cVAction} and \eqref{LambdaIsoTf} imply:
\ben
\label{Vpsi}
V_B(u)=V_{\bAd_0(U)(B)}(\psi_U(u))~~,~~\forall B\in \R^3~
~\forall U\in \PSU(1,1)~~\forall u\in \rD~~.
\een
Moreover, the discussion of Subsection \ref{subsec:V} shows that the
general solution of the $\Lambda$-$V$ equation \eqref{LambdaV} is
unchanged when one rescales $\Lambda$ by a non-zero constant. These
observations imply the following:
\begin{itemize} \itemsep 0.0em
\item If $B$ is timelike or lightlike, there exists a proper
  orthochronous Lorentz transformation which brings $B$ to either
  of the following two forms:
\begin{itemize}
  \itemsep 0.0em
\item $B'=(C,0,0)$ with $C=\sign(B^0)\sqrt{(B,B)}$ (when $B$ is timelike)
\item $B'=(C,0,C)$ with $C=B^0$ (when $B$ is lightlike).
\end{itemize}
\item If $B$ is spacelike, there exists a proper orthochronous Lorentz
transformation (namely a spatial rotation) which brings $B$ to the
form $B'=(0,0,C)$, where $C=\sqrt{|(B,B)|}$.
\end{itemize}
Moreover, Remark \ref{rem:Vscale} of Subsection \ref{subsec:LambdaV}
shows that we can rescale $\Lambda$ by $1/C$ without changing
$V$. This allows us to further reduce to one of the thee {\em
  canonical cases} $B=B_\can\in \{(1,0,0),(1,0,1),(0,0,1)\}$. In each
of the three cases, we have $B'=C B_\can$ and:
\be
V_{B'}=V_{B_\can}=V_{B'/C}
\ee
as well as:
\be
V_B(u)=V_{B'}(\psi_U(u))~~,
\ee
where $B'=\bAd_0(U)(B)$ and $\bAd_0(U)\in \SO_0(1,2)$ (with $U\in
\PSU(1,1)$) is the corresponding Lorentz transformation.

In conclusion, we can reduce the problem of determining $V$ to the
three canonical cases $B=B_\can\in \{(1,0,0),(0,0,1),(1,0,1)\}$,
depending on whether $B$ is timelike, spacelike or lightlike. We next
study each case in turn.

\subsubsection{The case of timelike $\Lambda$}
\label{subsec:timelike}

\noindent In this case, there exists a proper and orthochronous
Lorentz transformation $\bAd_0(U)$ which brings $B=(B^0,B^1,B^2)$ to
the form $B'\eqdef\bAd_0(U)(B)=(C,0,0)=C B_\can$ with $C=\epsilon
\sqrt{(B,B)}$ (where $\epsilon\eqdef \sign(B_0)$) and
$B_\can=E_0=(1,0,0)$. We can take $U=U(t,a,0)\in \PSU(1,1)$, with
$t>0$ and $a$ determined by the relations:\footnote{The parameter $a$
 (see Appendix \ref{app:iso}) should not be confused 
with the scale factor $a(t)$.}
\ben
\label{paramtimelike}
t=\arccosh\left[\frac{|B^0|}{\sqrt{(B,B)}}\right]~,
~\cos(a)=-\epsilon\frac{B^1}{\sqrt{B_1^2+B_2^2}}~,
~\sin(a)=\epsilon\frac{B^2}{\sqrt{B_1^2+B_2^2}}~~.
\een

\paragraph{The canonical timelike Hesse function.}

For $B=B_\can=(1,0,0)$, we have $\Delta=0$ and the corresponding Hesse
function:
\ben
\label{LambdaCanTimelike}
\Lambda_{B_\can}(u)=\Lambda_{1,0,0}(u)=\Xi^0(u)
=\frac{1+\rho^2}{1-\rho^2}=\frac{2}{1-\rho^2}-1
\een
has a single critical point located at $u_c=0$, which is an absolute
minimum with $\Lambda_{1,0,0}(u_c)=1$; moreover, $\Lambda_{1,0,0}$
tends to $\infty$ at the conformal boundary of $\mD$.  For each
$\lambda\in [1,+\infty)$, the level set $\Lambda_{1,0,0}=\lambda$ is
the Euclidean circle centered at the origin of radius
$R_\lambda=\sqrt{\frac{\lambda-1}{\lambda+1}}$, which varies from
$R_1=0$ to $R_\infty=1$ as $\lambda$ increases from $1$ to
$+\infty$. The level sets are hyperbolic circles, since they are all
contained inside $\rD$. 

\begin{figure}[H]
\centering \centering
\includegraphics[width=0.39\linewidth]{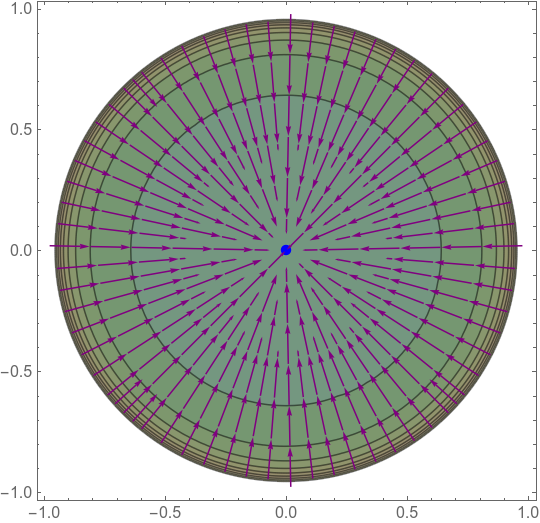}
\vskip 0.45em
\caption{Contour plot of the canonical timelike Hesse function
  $\Lambda_{1,0,0}$. The values of the function decrease from
  lightest brown to darkest green. The critical point is shown as a
  blue dot. The gradient vector field is shown in purple. This coincides with
  Figure \ref{fig:LambdaTimelike0}, which we recall here for convenience of the reader.}
\label{fig:LambdaCanTimelike}
\end{figure}

\paragraph{The scalar potential in the canonical timelike case.}

The gradient flow equations of $\Lambda_{1,0,0}$ with respect to the
metric $\cG$ have the following form in polar Euclidean coordinates
$(\rho,\theta)$:
\beqan
\label{Timelikegf}
\frac{\dd \rho}{\dd q}=-\beta^2 \rho~~,~~\frac{\dd\theta}{\dd q}=0~~,
\eeqan with the solution $\theta=\const$ and $\rho=e^{-\beta^2 q}$,
where $q\in (0,+\infty)$ and we chose the integration constant such
that $\rho\rightarrow 1$ for $q\rightarrow 0$.  Hence the gradient
flow curves of $\Lambda_{1,0,0}$ are straight line segments flowing
from the conformal boundary to the origin of $\mD$ as $q$ varies from
$0$ to $+\infty$ (see Figure \ref{fig:LambdaCanTimelike}). Let
$\gamma_\theta$ denote the gradient flow line of polar angle $\theta$.

For $\Lambda=\Lambda_{1,0,0}$, relation \eqref{dqdlambda} becomes:
\be
\frac{\dd q}{\dd\lambda}=-\frac{1}{||\dd\Lambda_{1,0,0}||_\cG^2}=-\frac{1}{\beta ^2 \left(\lambda ^2-1\right)}~~.
\ee
Along the gradient flow curve $\gamma_\theta$, we have:
\be
\lambda=\frac{1+\rho^2}{1-\rho^2}\Longleftrightarrow \rho=\frac{\sqrt{\lambda -1}}{\sqrt{\lambda +1}}
\ee
and:
\be
\int_{(\gamma)}^\lambda \frac{\lambda'\dd\lambda'}{||\dd\Lambda_{1,0,0}||_\cG^2}=\frac{1}{2\beta^2}\log
  \left(\lambda ^2-1\right)+\cC(\theta)~~,
\ee
where $\gamma(\lambda)=\rho\cos\theta+\i \rho\sin\theta$ and where
$\cC(\theta)$ is a constant of integration which can depend on $\theta$
in a $2\pi$-periodic manner. Relation \eqref{Vgamma} gives:
\ben
\label{VCanTimelike}
V_{B_\can}(\rho,\theta)=\omega(\theta)(\Lambda_{B_\can}(\rho)^2-1)=\omega(\theta) \frac{4\rho^2}{(1-\rho^2)^2}~~,
\een
where $\omega(\theta)\eqdef e^{2\beta^2 \cC(\theta)}$ is a positive
$2\pi$-periodic smooth function of $\theta$.  

\paragraph{Accidental visible symmetries in the canonical timelike case.}

It is clear that $V_{1,0,0}$ is invariant under a continuous subgroup of
$\Iso_+(\mD)\simeq\PSU(1,1)$ iff $\omega$ is independent of $\theta$,
in which case $V_{1,0,0}$ is stabilized by the $\U(1)$ subgroup $\cR$
corresponding to rotations of the disk around its origin (see Appendix
\ref{app:iso}). The image of this subgroup in the adjoint
representation $\bAd_0$ coincides with the $\SO(2)$ group of spatial
rotations which stabilizes the timelike 3-vector $B_\can=(1,0,0)$ in
the Lorentz group. Hence the Hessian two-field model defined by
$V_{1,0,0}$ also admits visible symmetries iff $\omega$ is independent
of $\theta$, in which case the space of infinitesimal visible
symmetries is one-dimensional and generated by the vector field
$\pd_\theta$.

\paragraph{The scalar potential when $B=B'=C B_\can$.}
Recalling the relations $V_{B'}=V_{B_\can}$ for $B'=C B_\can$ as well
as $\Lambda_{B_\can}=\Lambda_{B'}/C$ (where
$C=\epsilon\sqrt{(B',B')}$, relation \eqref{VCanTimelike} gives:
\ben
\label{Vtimelike}
V_{B'}(\rho,\theta)=\omega(\theta)\left[\frac{\Lambda_{B'}
(\rho)^2}{(B',B')}-1\right]~~(\mathrm{when}~~B'_1=B'_2=0)~~.
\een

\paragraph{Lorentz-invariant form of the scalar potential.}

When $B=B'=CE_0$, the polar angle $\theta$ on the hyperbolic disk
parameterizes the unit spacelike vector $n_B(u)\eqdef \cos\theta
E_1+\sin\theta E_2$ obtained by normalizing the projection
$\Xi_B(u)$ of $\Xi(u)$ onto the spacelike plane orthogonal to $B$
(see Figure \ref{fig:ProjTimelike}), which in this case is spanned by
the three-vectors $E_1$ and $E_2$.

We have:
\be
\Xi_B(u)=\Xi(u)-\frac{(B,\Xi(u))B}{(B,B)}~~,
\ee
and $n_B(u)=\frac{\Xi_B(u)}{||\Xi_B(u)||_E}$, where the
Euclidean norm of $\Xi_B(u)$ is given by:
\be
||\Xi_B(u)||_E=\sqrt{-(\Xi_B(u),\Xi_B(u))}=\sqrt{\frac{(B,\Xi(u))^2}{(B,B)}-1}~~,
\ee
where we used the relation $(\Xi(u),\Xi(u))=1$. Combining these formulas gives:
\ben
\label{ntimelike}
n_B(u)=\frac{(B,B)\Xi(u)-(B,\Xi(u)) B}{\sqrt{(B,B)(B,\Xi(u))^2-(B,B)^2}}
=\frac{(B,B)\Xi(u)-B\Lambda_B(u)}{\sqrt{(B,B)\Lambda_B(u)^2-(B,B)^2}}~~.
\een
Since this relation is manifestly Lorentz invariant, it is valid not
only for $B=B'$, but also for any lightlike vector $B$. In particular,
$\omega$ can be viewed as a function of the unit spacelike vector
$n_B$ and relation \eqref{Vtimelike} can be written in the manifestly
Lorentz-invariant form:
\ben
\label{VInvTimelike}
V_B(u)=\omega(n_B(u))\left[\frac{\Lambda_B(u)^2}{(B,B)}-1\right]~~.
\een
Direct computation using \eqref{LambdaXi} gives:
{\scriptsize
\ben
\label{Sq1}  
\frac{\Lambda_B(u)^2}{(B,B)}-1=\frac{B_1^2(1+2x^2-2y^2+\rho^4)+B_2^2(1-2x^2+2y^2+\rho^4)
+4B_0^2\rho^2+4B_0 (1+\rho^2)(B_1x+B_2y)+8 B_1 B_2 xy}{(B,B)(1-\rho^2)^2}~~.
\een}

\begin{figure}[H]
\centering \centering
\includegraphics[width=0.4\linewidth]{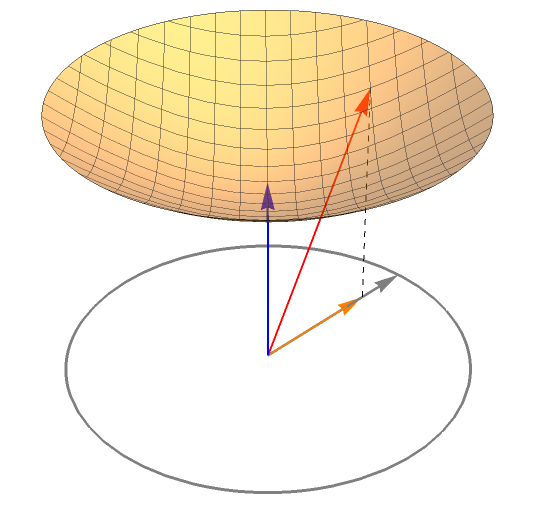}
\vskip 0.42em
\caption{The unit spacelike vector $n_B(u)$ (shown in gray) determined
  by the unit timelike 3-vector $\Xi(u)$ (shown in red) and a timelike
  3-vector $B$ (shown in blue); the projection $\Xi_B(u)$ is shown in
  orange. We also show the future sheet of the unit hyperboloid. As
  $u$ varies in $\rD$, the vector $n_B(u)$ describes a circle of unit
  radius (shown in gray) contained in the spacelike plane orthogonal
  to $B$; one can think of the function $\omega$ as being defined on
  this circle. The figure shows the case when $B$ is future-pointing
  with $(B,B)=1$.}
\label{fig:ProjTimelike}
\end{figure}

\begin{remark}
Equation \eqref{Vpsi} shows that the general solution of
\eqref{LambdaV} for a Hesse function of timelike parameter $B$ is
obtained by acting on $\mD$ with the $\PSU(1,1)$ transformation
$U=U(t,a,0)$, with $a$ and $t$ given in \eqref{paramtimelike}:
\be
V_B(u)=V_{B'}(\psi_U(u))~~,
\ee
where $B'=\bAd_0(U)(B)$. This amounts to replacing $(\rho,\theta)$ in
expression \eqref{Vtimelike} by polar semi-geodesic coordinates
$(\trho,\ttheta)$ centered at the critical point $u_c$ of $\Lambda_B$.
In these new coordinates, the curves $\trho=\const$ (which coincide
with the level sets of $\Lambda_B$) are hyperbolic circles with center
$u_c$, while the curves $\ttheta=\const$ (which coincide with the
gradient flow curves of $\Lambda_B$) are hyperbolic geodesics
orthogonal to these hyperbolic circles and passing through $u_c$ (see
Figure \ref{fig:LambdaTimelike}). We have:
\ben
\label{VTimelike}
V_B(x,y)=\omega(\ttheta(x,y))\left[\frac{\Lambda_B(u)^2}{(B,B)}-1\right]~~,
\een
where:
\be
\ttheta(x,y)=\arg\left[\frac{ \sign(B_0) \,(B_1 - \i B_2) \,(x+\i y) 
+ \left( |B_0| - \sqrt{(B,B)} \right)}{ \left(|B_0| - \sqrt{(B,B)}\right) (x+\i y)
+ \sign(B_0) \,(B_1+\i B_2)}\right]~~.
\ee
\end{remark}

\paragraph{Accidental visible symmetries in the general timelike case.}

The potential \eqref{VTimelike} is stabilized by a non-trivial
continuous subgroup of $\PSU(1,1)\simeq \Iso_+(\mD)$ (and hence the
corresponding cosmological model also admits visible symmetries) iff
the function $\omega$ is constant on the unit circle. In this case,
the group of visible symmetries of the model coincides with the
stabilizer of $V_B$ in $\Iso_+(\mD)$. This is an elliptic $\U(1)$
subgroup of $\PSU(1,1)$ which identifies with the stabilizer of the
3-vector $B$ under the adjoint representation:
\be
\Stab_{\PSU(1,1)}(V_B)=\Stab_{\PSU(1,1)}(B)\simeq \U(1)
\ee
and is conjugate to the canonical rotation subgroup $\cR$ 
through the adjoint action of the group element
$U=U(t,a,0)$:
\be
\Stab_{\PSU(1,1)}(V_B)=U^{-1}\cR U~~.
\ee

\subsubsection{The case of spacelike $\Lambda$}
\label{subsec:spacelike}

\noindent In this case, there exists $U=U(t,a,\frac{\pi}{2}-a)\in
\PSU(1,1)$ such that $\bAd_0(U)(B)=B'=C B_\can$, where
$B_\can=(0,0,1)=E_2$, $C=\sqrt{|(B,B)|}$ and the parameters $t,a$ are
determined by the relations:
\ben
\label{paramspacelike}
t=\!-\arcsinh\left(\!\frac{B^0}{\sqrt{|(B,B)|}}\!\right),
~\sin(2a)\!=\!\frac{B^1}{\sqrt{B_1^2+B_2^2}} ~,~\cos(2a)\!=\!\frac{B^2}{\sqrt{B_1^2+B_2^2}}~.
\een

\paragraph{The canonical spacelike Hesse function.}

We have:
\ben
\label{LambdaCanSpacelike}
\Lambda_{B_\can}(u)=\Lambda_{0,0,1}(u)=-\Xi^2(u)=-\frac{2y}{1-\rho^2}~~.
\een
This Hesse function has no critical point on $\rD$ (see Figure
\ref{fig:LambdaCanSpacelike}). It
vanishes along the horizontal segment $(-1,1)$ defined by $y=0$, being
positive in the lower half plane (where it tends to $+\infty$ for
$\rho\rightarrow 1$) and negative in the upper half plane (where it
tends to $-\infty$ for $\rho\rightarrow 1$). For each $\lambda\in
\R\setminus \{0\}$, the level set $\Lambda_{0,0,1}=\lambda$ is the
intersection with $\rD$ of the circle with center
$u=\frac{\i}{\lambda}$ and radius
$R_\lambda=\sqrt{1+\frac{1}{\lambda^2}}$, which is the hypercycle
consisting of all points of $\mD$ located at signed hyperbolic
distance $d_\lambda=-\arcsinh(\lambda)$ from the axis $(-1,1)$.

\paragraph{Fermi coordinates with axis $(-1,1)$.}
  
To describe the gradient flow lines of \eqref{LambdaCanSpacelike}, it is
convenient to pass to {\em hypercyclic (a.k.a. Fermi) coordinates}
$(\tau,\sigma)$ with axis $(-1,1)$. These are semi-geodesic
coordinates defined through:\footnote{The Fermi coordinate $\sigma$
 should not be mistaken with $\sigma=\sqrt{B_1^2+B_2^2}$ defined in \eqref{Bdef}.}
\beqa
&&\sign(y)\sigma\!\eqdef \!\arccosh\sqrt{1\!+\!\frac{4y^2}{(1\!-\!\rho^2)^2}}
\!=\!\arccosh\sqrt{\Xi^0(u)^2\!-\!\Xi^1(u)^2}\!=\!\arccosh\sqrt{1\!+\!\Xi^2(u)^2}~~\nn\\
&&\!\!\!\tau\!\eqdef\! \arcsinh\!\left(\!\frac{2x}{\sqrt{4y^2\!+\!(1\!-\!\rho^2)^2}}
\!\right)\!\!=\!\arcsinh\!\left(\!\frac{\Xi^1(u)}{\sqrt{\Xi^0(u)^2\!-\!\Xi^1(u)^2}}\!\right)\!\!
=\!\arcsinh\left(\!\frac{\Xi^1(u)}{\sqrt{1\!+\!\Xi^2(u)^2}}\!\right)~~
\eeqa
i.e.:
\be
\Xi^0(u)=\cosh\sigma\cosh \tau~~,~~\Xi^1(u)=\cosh\sigma\sinh\tau~~,
~~\Xi^2(u)=\sinh\sigma~~.
\ee

\begin{figure}[H]
\centering
\begin{minipage}{.45\textwidth}
\centering
\includegraphics[width=\linewidth]{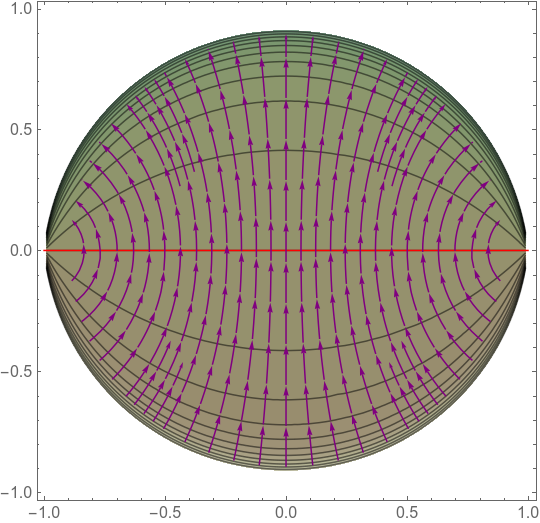}
\vskip 0.45em
\subcaption{$\Lambda_{0,0,1}$ on the hyperbolic disk.}
\label{fig:LambdaCanSpacelikeDisk}
\end{minipage}\hfill
\begin{minipage}{.45\textwidth}
 \centering \includegraphics[width=\linewidth]{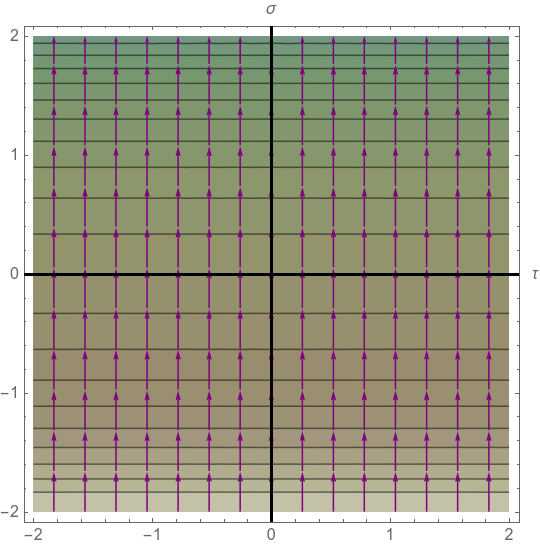}
\vskip 0.45em
\subcaption{$\Lambda_{0,0,1}$ in Fermi coordinates.}
\label{fig:LambdaCanSpacelikeFermi}
\end{minipage}\hfill 
\caption{Contour plot of the canonical spacelike Hesse function
  $\Lambda_{0,0,1}$ and of its gradient flow. The values of the
  function decrease from lightest brown to darkest green. The figure 
  to the left is Figure \ref{fig:LambdaDeg} recalled here for convenience of 
the reader. The vanishing locus of the function is the horizontal segment shown in
  red. The gradient vector field is shown in purple. The figure to the
  right shows only a portion of the $(\tau,\sigma)$-plane.}
\label{fig:LambdaCanSpacelike}
\end{figure}

\noindent In Fermi coordinates, the metric $\cG=G/ \beta^2$ takes the form:
\be
\dd s^2_\cG=\frac{1}{\beta^2}\left[\dd\sigma^2+\cosh^2(\sigma)\dd\tau^2\right]~~.
\ee
We have $\sign(\sigma(u))=\sign(y)=\sign(\Xi^2(u))$ and
$\tau,\sigma\in \R$. The curves $\sigma=\const$ are hypercycles with
axis given by the horizontal geodesic $(-1,1)$, while the curves
$\tau=\const$ are hyperbolic geodesics orthogonal to these hypercycles
(and hence also orthogonal to the $x$-axis).  In these coordinates,
the point $u=0$ corresponds to $(\tau,\sigma)=(0,0)$ while the
conformal boundary of $\mD$ corresponds to $\sigma^2+\tau^2\rightarrow
\infty$, being mapped to a circle at infinity of the
$(\tau,\sigma)$-plane. The $y$-axis $x=0$ corresponds to the line
$\tau=0$ while the $x$-axis $y=0$ corresponds to the line
$\sigma=0$. The squeeze transformation $T(t)\in \PSU(1,1)$ acts by:
\be
\sigma\rightarrow \sigma~~,~~\tau\rightarrow \tau+t~~.
\ee
In particular, the hypercycles with axis $(-1,1)$ are the orbits of
the squeeze subgroup $\cT$ of $\PSU(1,1)$ under the action of the
latter on $\mD$ by fractional transformations. Since
$\Lambda_{0,0,1}=-\sinh(\sigma)$, the level sets of the canonical
spacelike Hesse function coincide with the curves $\sigma=\const$,
which are hypercycles located at signed hyperbolic distance
$\sigma$ from the $x$-axis.

\paragraph{The scalar potential in the canonical spacelike case.}

The gradient flow equations of $\Lambda_{0,0,1}$ take the following
form in hypercyclic coordinates:
\be
\frac{\dd \sigma}{\dd q}=\beta^2 \cosh(\sigma)~~,~~\frac{\dd\tau}{\dd q}=0~~,
\ee
with the solution $\tau=\const$ and
$\sigma=2\arctanh\left[\tan\left(\frac{\beta^2
    q}{2}\right)\right]=\arcsinh(\tan(q))$, where $q$ runs in the
interval $(-\pi/2,\pi/2)$ and we chose $q=0$ to correspond to
$\sigma=0$, i.e. to the unique point $u_\tau=\tanh(\tau/2)\in (-1,1)$
where the gradient flow curve corresponding to $\tau$ intersects the
$x$-axis.  We denote by $\gamma_\tau$ this gradient flow curve.

We have $\lambda=-\sinh(\sigma)$ and equation
\eqref{dqdlambda} becomes:
\be
\frac{\dd q}{\dd\lambda}=-\frac{1}{||\dd\Lambda_{0,0,1}||_\cG^2}= 
-\frac{1}{\beta ^2 \left(\lambda ^2+1\right)} ~~.
\ee
Along the gradient flow curve $\gamma_\tau$ which passes through the point $u=x+\i
y=\rho(\cos\theta+\i\sin\theta)\in \rD$, we have:
\be
\int_{(\gamma)}^{\lambda} \frac{\lambda'\dd\lambda'}{||\dd\Lambda_{0,0,1}||_\cG^2}
=\frac{1}{2\beta^2}\log
\left(\lambda ^2+1\right) +\cC(\tau)~~,
\ee
where $\gamma(\lambda)=u$ and where $\cC(\tau)$ is a constant of
integration which can depend on
$\tau=\arcsinh\left(\frac{2x}{\sqrt{4y^2+(1-\rho^2)^2}}\right)$.
Equation \eqref{Vgamma} gives:
\ben
\label{VFermi}
V_{B_\can}=\omega(\tau)\cosh^2(\sigma)=\omega(\tau)(1+\Lambda_{B_\can}(\sigma)^2)~~,
\een
i.e.:
\ben
\label{VCanSpacelike}
\!\!\!\!\!\!V_{B_\can}\!(u)\!=\!
\omega(\tau)\!\left[1\!+\!\Lambda_{B_\can}(\sigma)^2\right]\!\!
=\omega(\tau)\frac{1\!-\!2x^2\!+\!2y^2\!+\!\rho^4}{(1\!-\!\rho^2)^2}~~,
\een
where $\omega\in \cC^\infty(\R)$ is a positive smooth real-valued
function defined through $\omega(\tau)=e^{2\beta^2 \cC(\tau)}$.

\paragraph{Accidental visible symmetries in the canonical spacelike case.}

It is clear that $V_{B_\can}$ is invariant under a continuous subgroup
of isometries of $\mD$ iff $\omega$ is independent of $\tau$, in which
case the stabilizer of $V_{B_\can}$ coincides with the squeeze subgroup
$\cT$ of $\PSU(1,1)$.  This corresponds to the group of boosts
$\Ad_0(T(t))$ in the two-plane $(X^0,X^1)$ of the Minkowski space
$\R^{1,2}$ (see Appendix \ref{app:iso}) which stabilize the 3-vector
$B_\can\!=\!(0,0,1)$. This subgroup is isomorphic with $(\R,+)$. Hence the
Hessian two-field model defined by $V_{B_\can}$ also admits visible
symmetries iff $\omega$ is independent of $\tau$, in which case the
group of visible symmetries coincides with $\cT$.

\paragraph{The scalar potential when $B=B'=C B_\can$.}
Recalling the relations $V_{B'}=V_{B_\can}$ for $B'=C B_\can$ as well
as $\Lambda_{B_\can}=\Lambda_{B'}/C$ (where $C=\sqrt{|(B',B')|}$),
equation \eqref{VCanSpacelike} gives:
\ben
\label{Vspacelike}
V_{B'}(\tau,\sigma)=\omega(\tau)\left[1+\frac{\Lambda_{B'}(\sigma)^2}{|(B',B')|}\right]~
~(\mathrm{when}~~B'_0=B'_1=0)~~.
\een

\paragraph{Lorentz-invariant form of the scalar potential.}

When $B=B_\can=(0,0,1)$, the hyperbolic angle $\tau$ parameterizes the
unit timelike vector $n_B(u)=(\cosh \tau) E_0+(\sinh\tau) E_1$, which lies
in the direction of the projection $\Xi_B(u)=\Xi^0(u) E_0+\Xi^1(u)
E_1$ of $\Xi(u)$ onto the Minkowski plane orthogonal to
$B$ (see Figure \ref{fig:ProjSpacelike}).

\begin{figure}[H]
\centering \centering
\includegraphics[width=0.45\linewidth]{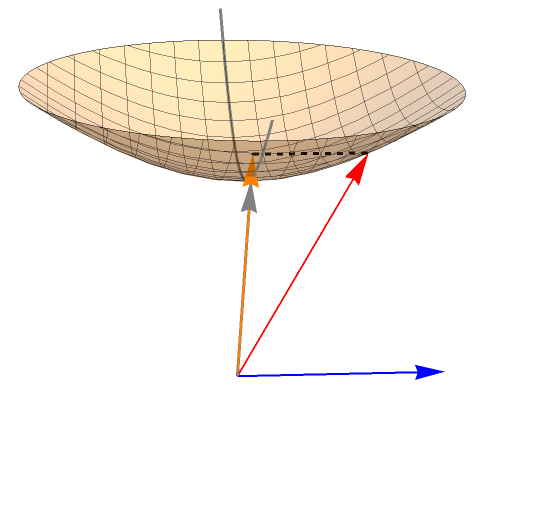}
\vskip -2.1em
\caption{The unit timelike vector $n_B(u)$ (shown in gray) determined
  by the unit timelike vector $\Xi(u)$ and by the spacelike 3-vector
  $B$ (shown in blue). The timelike 3-vector $\Xi(u)$ is shown in red,
  while $\Xi_B(u)$ is shown in orange. As $u$ varies in $\rD$, the
  vector $n_B(u)$ describes the hyperbola (shown in gray) obtained by
  intersecting the future sheet of the unit hyperboloid with the
  Minkowski plane orthogonal to $B$; one can think of the function
  $\omega$ as being defined on this hyperbola. }
\label{fig:ProjSpacelike}
\end{figure}

\noindent We have:
\be
\Xi_B(u)=\Xi(u)+\frac{(B,\Xi(u)) B}{|(B,B)|}~~,
\ee
which gives:
\be
(\Xi_B(u),\Xi_B(u))=1+\frac{(B,\Xi(u))^2}{|(B,B)|}~~,
\ee
where we used the relation $(\Xi(u),\Xi(u))=1$. Thus:
\ben
\label{nspacelike}
n_B(u)\!=\!\frac{\Xi_B(u)}{\sqrt{\!(\Xi_B(u),\!\Xi_B(u))}}\!=\!\frac{|(B,\!B)|\,\Xi(u)
\!+\!(B,\Xi(u)) B}{\sqrt{\!(B,\!B)^2\!+\!|(B,\!B)|(B,\Xi(u))^2}}\!=\!
\frac{|(B,\!B)|\,\Xi(u)\!+\!\Lambda_B(u) B}{\sqrt{\!(B,\!B)^2\!+\!|(B,\!B)|\Lambda_B(u)^2}}.
\een
Since this relation is manifestly Lorentz invariant, it is valid not
only for $B=B_\can$ but also for any spacelike vector $B$. In
particular, $\omega$ can be viewed as a function of the unit timelike
vector $n_B$ and relation \eqref{Vspacelike} can be written in the
manifestly Lorentz-invariant form:
\ben
\label{VInvSpacelike}
V_B(u)=\omega(n_B(u))\left[1+\frac{\Lambda_B(u)^2}{|(B,B)|}\right]
=-\omega(n_B(u))\left[\frac{\Lambda_B(u)^2}{(B,B)}-1\right]~~,
\een
where the quantity $\frac{\Lambda_B^2}{(B,B)}-1$ has the form given in \eqref{Sq1}.

\begin{remark}
Equation \eqref{Vpsi} shows that the general solution of the
$\Lambda$-$V$ equation \eqref{LambdaV} for a Hesse function of
spacelike parameter $B$ is obtained by acting on $\mD$ with the
$\PSU(1,1)$ transformation $U=U(t,a,\pi/2-a)$, with $t$ and $a$ given
in \eqref{paramspacelike}.  This amounts to replacing $(\tau,\sigma)$
in expression \eqref{VFermi} by hypercyclic coordinates
$(\ttau,\tsigma)$ with axis given by the vanishing locus $Z_B$ of the
spacelike Hesse function $\Lambda_B$ (which is a hyperbolic
geodesic). In the new coordinates, the curves $\tsigma=\const$ (which
coincide with the level sets of $\Lambda_B$) are hypercycles with axis
$Z_B$ while the curves $\ttau=\const$ (which coincide with the
gradient flow curves of $\Lambda_B$) are hyperbolic geodesics
orthogonal to these hypercycles (see Figure
\ref{fig:LambdaSpacelike}). We have:
\ben
\label{VSpacelike}
V_B(x,y)=\omega(\ttau(x,y))\left[1+\frac{\Lambda_B(x,y)^2}{|(B,B)|}\right]~~.
\een
with:
{\scriptsize
  \ben
\label{ttauSpacelike}
\sinh\ttau(x,y)\!=\!\frac{2 \sqrt{|(B,B)|} (B_1 y\!-\!B_2x)}{\sqrt{(B_1^2\!+\!B_2^2)
    \left[B_1^2(1\!+\!2x^2\!-\!2y^2\!+\!\rho^4)\!+\!B_2^2(1\!-\!2x^2\!+\!2y^2\!+\!\rho^4)\!+\!4B_0^2\rho^2
\!+\!4B_0(1\!+\!\rho^2)(B_1x\!+\!B_2y)\!+\!8B_1B_2xy\right]}}
\een
}
\end{remark}

\paragraph{Accidental visible symmetries in the general spacelike case.}

It is clear that the potential \eqref{VSpacelike} is stabilized by a
non-trivial continuous subgroup of $\Iso_+(\mD)\simeq \PSU(1,1)$ (and
hence the corresponding cosmological model also admits visible
symmetries) iff the function $\omega$ is constant.  In this case, the
group of visible symmetries of the model coincides with the stabilizer
of $V_B$. This is a hyperbolic subgroup isomorphic with $(\R,+)$ which
identifies with the stabilizer of the spacelike 3-vector $B$ under the
adjoint representation:
\be
\Stab_{\PSU(1,1)}(V_B)=\Stab_{\PSU(1,1)}(B)\simeq (\R,+)
\ee
and is conjugate to the squeeze subgroup $\cT$ of $\PSU(1,1)$:
\be
\Stab_{\PSU(1,1)}(V_B)=U^{-1} \cT U~~,
\ee
where $U=U(t,a,\pi/2-a)$.

\subsubsection{The case of lightlike $\Lambda$}
\label{subsec:lightlike}

\noindent In this case, there exists $U=U(0,a,0)\in \PSU(1,1)$ such
that $\bAd_0(U)(B)=(C,0,C)$, where $C=B^0$ and $a$ is determined by
the relations:
\ben
\label{paramlightlike}
\sin(2a)=\frac{B^1}{B^0}~~,~~\cos(2a)=\frac{B^2}{B^0}~~.
\een
\noindent Using \eqref{Vpsi}, we can thus always reduce to the case
$B=B'=(C,0,C)$, while a rescaling of $\Lambda$ allows us to further
reduce to the case $B=B_\can=(1,0,1)=E_0+E_2$.

\paragraph{The canonical lightlike Hesse function.}

We have:
\ben
\label{LambdaCanLightlike}
\Lambda_{B_\can}(u)=\Lambda_{1,0,1}(u)=\Xi^0(u)-\Xi^2(u)=\frac{\rho^2-2y+1}{1-\rho^2}~~.
\een
This Hesse function has no critical points on $\rD$ and is positive
everywhere inside $\rD$ (see Figures \ref{fig:LambdaCanLightlike} and
\ref{fig:LambdaLightlike}). It tends to $+\infty$ at all
points of the conformal boundary of $\mD$ except for the point
$u_0=\i$, where it tends to zero. For any $\lambda\in (0,+\infty)$,
the level set $\Lambda_{1,0,1}=\lambda$ is a horocycle with center
$u_0=\i$.

\begin{figure}[H]
\centering
\begin{minipage}{.45\textwidth}
\centering
\includegraphics[width=\linewidth]{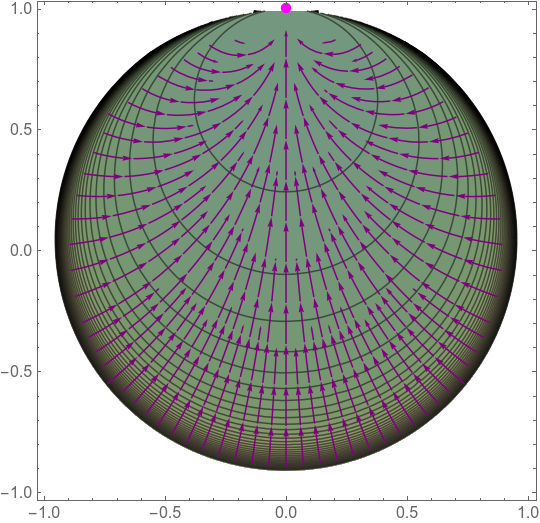}
\vskip 0.45em
\subcaption{$\Lambda_{1,0,1}$ on the hyperbolic disk.}
\label{fig:LambdaLightlikeDisk}
\end{minipage}\hfill
\begin{minipage}{.45\textwidth}
 \centering \includegraphics[width=\linewidth]{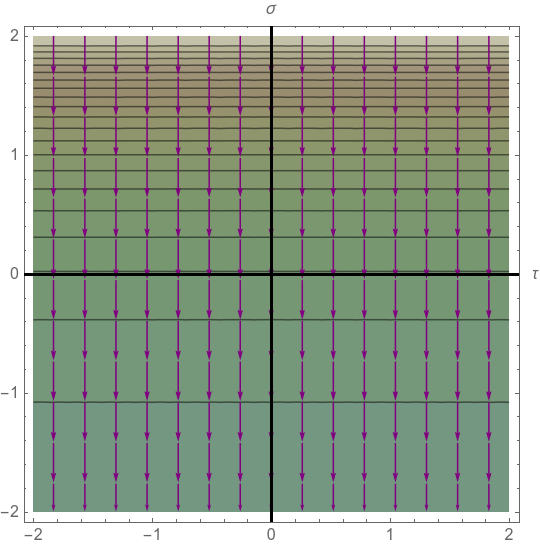}
\vskip 0.45em
\subcaption{$\Lambda_{1,0,1}$ in horocyclic coordinates.}
\label{fig:LambdaLightlikeHor}
\end{minipage}\hfill 
\caption{Contour plot of the canonical lightlike Hesse function
  $\Lambda_{1,0,1}$ and of its gradient flow. The values of the
  function decrease from lightest brown to darkest green. In the
  figure on the left (which coincides with Figure \ref{fig:LambdaLightlike}), 
the single point on the conformal boundary of $\mD$
  where the function tends to zero is shown as a red dot. The gradient
  vector field is shown in purple. The figure to the right shows only
  a portion of the $(\tau,\sigma)$-plane.}
\label{fig:LambdaCanLightlike}
\end{figure}

\paragraph{Horocyclic coordinates centered at $\i$.}

To describe the gradient flow of $\Lambda_{1,0,1}$, it is convenient
to pass to {\em horocyclic coordinates} (which we again denote by
$(\tau,\sigma)$) centered at $u=\i$. These are the hyperbolic polar
geodesic coordinates defined through:
\beqan
\sigma &\eqdef& \log \left(\frac{\rho^2-2y+1}{1-\rho^2}\right)=\log\left[\Xi^0(u)
-\Xi^2(u)\right]\nn\\
\tau &\eqdef& \frac{2x}{\rho^2-2y+1}=\frac{\Xi^1(u)}{\Xi^0(u)-\Xi^2(u)}~~,
\eeqan
i.e.:
\be
\Xi^0(u)=\frac{1}{2}\tau^2e^\sigma+\cosh\sigma~~,~\Xi^1(u)=\tau e^\sigma~~,
~~\Xi^2(u)=\frac{1}{2}\tau^2e^\sigma-\sinh\sigma~~.
\ee
In particular, we have
$\Lambda_{B_\can}(u)=\Xi^0(u)-\Xi^2(u)=e^\sigma$. In horocyclic
coordinates, the metric $\cG = G/ \beta^2$ takes the
form:
\be
\dd s^2_\cG=\frac{1}{\beta^2}(\dd \sigma^2+e^{2\sigma}\dd \tau^2)~~.
\ee
In these coordinates, the curves $\sigma=\const$ are the horocycles
centered at $u=\i$ while the curves $\tau=\const$ are the geodesics
normal to these horocycles, which have $\i$ as a limit point. In coordinates
$(\tau, \sigma)$, the disk $\rD$ is mapped to the entire plane $\R^2$,
the conformal boundary of $\mD$ corresponding to a circle at
infinity. The origin of the disk corresponds to the origin of the
$(\tau,\sigma)$-plane. The $y$-axis $x=0$ is mapped to the
$\sigma$-axis $\tau=0$, while the Euclidean circle of radius $1/2$
centered at $u=\frac{\i}{2}$ (which is a horocycle of $\mD$) is mapped
to the $\tau$-axis $\sigma=0$. The interior of this horocycle is
mapped to the half-plane $\sigma<0$, while its exterior is mapped to
the half-plane $\sigma>0$; moreover, the limit $u\rightarrow \i$
corresponds to $\sigma\rightarrow -\infty$ and $\tau\rightarrow \pm
\infty$, where $\tau\rightarrow +\infty$ if $u$ approaches $\i$ from
the half-disk defined by $\Re u>0$ and $\tau\rightarrow -\infty$ if
$u$ approaches $\i$ from the half-disk defined by $\Re u<0$. The
horocycles with center $\i$ correspond to the curves $\sigma=\const$,
while the hyperbolic geodesics which asymptote to $\i$ correspond to
the curves $\tau=\const$. The shear transformation $P(\kappa)\in
\SU(1,1)$ acts as:
\be
\sigma\rightarrow \sigma~~,~~\tau\rightarrow \tau+2\kappa~~.
\ee
In particular, the horocycles centered at $\i$ coincide with the
orbits of the shear subgroup $\cP$ under the action by fractional
transformations.  Since $\Lambda_{1,0,1}=e^\sigma$, the level sets of
$\Lambda$ correspond to the curves $\sigma=\const$, which are
horocycles passing through the point $\i$.

\paragraph{The scalar potential in the canonical lightlike case.}

In horocyclic coordinates, the gradient flow equations of $\Lambda$
take the form:
\be
\frac{\dd\sigma}{\dd q}=-\beta^2 e^\sigma~~,~~\frac{\dd\tau}{\dd q}=0~~,
\ee
with the solution $\tau=\const$ and $\sigma=-\log\left[\beta^2
  q+1\right]$, where we chose the integration constant such that $q$
runs between $-1/\beta^2$ and $+\infty$, with $q\rightarrow
-1/\beta^2$ corresponding to $\sigma\rightarrow +\infty$ and
$q\rightarrow +\infty$ corresponding to $\sigma\rightarrow -\infty$.
In the first limit, the gradient flow approaches a point on the
conformal boundary of $\mD$ where $\Lambda_{1,0,1}$ tends to plus
infinity, while in the second limit the gradient flow approaches the
ideal point $u=\i$ (where $\Lambda_{1,0,1}$ tends to zero). The value
$q=0$ corresponds to the horocycle defined by the equation
$\sigma=0$, which is the level set where $\Lambda_{1,0,1}=1$.

We have $\lambda=e^\sigma=\frac{\rho^2-2y+1}{1-\rho^2}$ and equation
\eqref{dqdlambda} becomes:
\be
\frac{\dd q}{\dd\lambda}=-\frac{1}{||\dd\Lambda_{1,0,1}||_\cG^2}=-\frac{1}{\beta ^2 \lambda^2} ~~.
\ee
Along a gradient flow curve $\gamma_\tau$ of $\Lambda_{1,0,1}$ which
passes through the point $u=x+\i y=\rho(\cos\theta+\i \sin\theta)\in
\rD$, we have:
\be
\int_{(\gamma)}^{\lambda} \frac{\lambda'\dd\lambda'}{||\dd\Lambda_{1,0,1}||_\cG^2}
=\frac{1}{\beta^2}\log (\lambda) + c(\tau)~~,
\ee
where $\gamma(\lambda)\!=\!u$ and $c(\tau)$ is a constant of integration
which can depend on $\tau\!=\!\frac{2x}{\rho^2-2y+1}$. Relation
\eqref{Vgamma} gives:
\ben
\label{VCanLightlike}
V_{B_\can}(\tau,\sigma)=\omega(\tau)e^{2\sigma}=\omega(\tau)\Lambda_{B_\can}(\sigma)^2~~,
\een
i.e.:
\ben
V_{B_\can}(u)\!=\omega(\tau)e^{2\sigma}\!= \omega(\tau)
\frac{\left(\rho^2\!-\!2 \rho \sin\theta\!+\!1\right)^2}{(1-\rho^2)^2}
\!=\omega\!\left(\!\frac{2x}{\rho^2\!-\!2y\!+\!1}\!\right)\!\!
\frac{\left(\rho^2\!-\!2 y\!+\!1\right)^2}{(1-\rho ^2)^2}~,
\een
where the function $\omega\in \cC^\infty(\R)$ is defined through
$\omega(\tau)=e^{2\beta^2 c(\tau)}$.

\paragraph{Accidental visible symmetries in the canonical lightlike case.}

It is clear that $V_{1,0,1}$ is invariant under a continuous subgroup
of $\PSU(1,1)$ iff $\omega$ is independent of $\tau$, in which case
the stabilizer of $V_{1,0,1}$ is the shear subgroup $\cP$. Notice that
$\cP$ identifies with the stabilizer of $B_\can=(1,0,1)$ under the
adjoint representation $\bAd_0$ of $\PSU(1,1)$. Hence the Hessian
two-field model with potential $V_{1,0,1}$ also admits visible
symmetries iff $\omega$ is independent of $\tau$, in which case the
group of visible symmetries coincides with the shear subgroup $\cP$ of
$\PSU(1,1)\simeq \Iso_+(\mD)$.

\paragraph{The scalar potential when $B=B'=C B_\can$.}

Recalling the relations $V_{B'}=V_{B_\can}$ for $B'=C B_\can$ as well
as $\Lambda_{B_\can}=\Lambda_{B'}/C$ (where
$C=B'_0$), relation \eqref{VCanLightlike} gives:
\ben
\label{Vlightlike}
V_{B'}(\tau,\sigma)=\omega(\tau)\frac{\Lambda_{B'}(\sigma)^2}{(B'_0)^2}~
~(\mathrm{when}~~B'_1=0)~~.
\een

\paragraph{Lorentz invariant form of the scalar potential.}

Let $\Xi_B(u)$ denote the projection of the timelike vector $\Xi$ onto
the light cone of $\R^{1,2}$, taken parallel to the lightlike vector
$B$ (hence the 3-vector $\Xi_B(u)$ lies in the intersection of the
light cone with the Minkowski plane spanned by $\Xi(u)$ and $B$ (see
Figure \ref{fig:ProjLightlike}). We have $\Xi_B(u)=\Xi(u)-\alpha B$,
where $\alpha\in \R$ is determined by the condition
$(\Xi_B(u),\Xi_B(u))=0$, which gives
$\alpha=\frac{1}{2(B,\Xi(u))}$. Thus:
\be
\Xi_B(u)=\Xi(u)-\frac{B}{2(B,\Xi(u))}~~.
\ee
Consider the lightlike 3-vector:
\be
n_B(u)\eqdef \frac{\Xi_B(u)}{(B,\Xi(u))}=\frac{2(B,\Xi(u))\Xi(u)-B}{2(B,\Xi(u))^2}
=\frac{2\Lambda_B(u)\Xi(u)-B}{2\Lambda_B(u)^2}~~,
\ee
which satisfies $(n_B,B)=1$ and hence lies inside the affine plane
$\Pi_{B}\subset \R^{1,2}$ defined by the equation $(X,B)=1$. We have
$\Lambda_{B_\can}(u)=(B_\can,\Xi(u))=e^\sigma$ and:
\be
\Xi_{B_\can}(u)=e^\sigma \left(\frac{\tau^2+1}{2},\tau, \frac{\tau^2-1}{2}\right)~~,
\ee
which gives:
\be
n_{B_\can}(u)=\left(\frac{\tau^2+1}{2},\tau, \frac{\tau^2-1}{2}\right)~~.
\ee
This shows that the horocyclic coordinate $\tau$ parameterizes the
light-like vector $n_{B_\can}$. We have $\Pi_{B_\can}=n_0+\Pi_0$,
where the 3-vector $n_0=(1/2,0,-1/2)$ lies inside the light-cone and
$\Pi_0$ is the linear plane defined by the equation $(X,B_\can)=0$,
i.e. $X^0=X^2$. The vectors $\epsilon_1=(0,1,0)$ and
$\epsilon_2=(1,0,1)$ form a basis of this linear plane. Thus
$n_{B_\can}(u)=n_0+\nu(u)$, where:
\be
\nu(u)=\left(\frac{\tau^2}{2},\tau, \frac{\tau^2}{2}\right)
=\tau\epsilon_1+\frac{\tau^2}{2}\epsilon_2\in \Pi_0~~,
\ee
which shows that $n_{B_\can}(u)$ describes the parabola obtained by
intersecting the light cone with the plane $\Pi_{B_\can}$. The apex of
this parabola is the 3-vector $\nu_0$, which corresponds to $\tau=0$.

In particular, $\omega$ can be viewed as a function of the unit timelike
vector $n_B$ and relation \eqref{Vlightlike} can be written in the
manifestly Lorentz-invariant form:
\ben
\label{VInvLigtlike}
V_B(u)=\omega(B_0 n_B(u))\frac{\Lambda_B(u)^2}{B_0^2}~~.
\een
where:
\ben
\label{VSq2}
\frac{\Lambda_B(u)^2}{B_0^2}=\frac{\left(B_0(1+\rho^2)+2B_1 x
+2B_2 y\right)^2}{B_0^2(1-\rho^2)^2}~~,
\een

\begin{figure}[H]
\centering \centering
\includegraphics[width=0.6\linewidth]{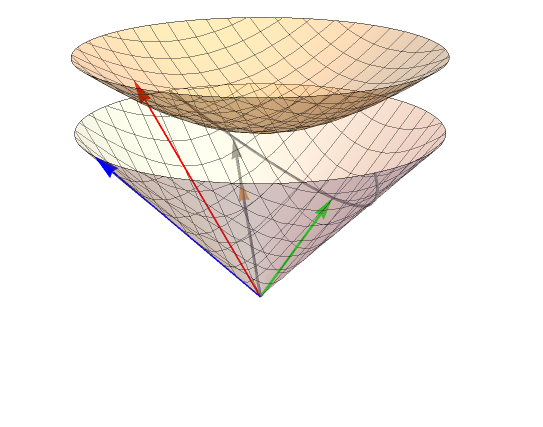}
\vskip -3em
\caption{The lightlike vector $n_B(u)$ (shown in gray) determined by
  $\Xi(u)$ (shown in red) and by a lightlike 3-vector $B$ (shown in
  blue), where the 3-vector $\Xi_B(u)$ is shown in orange. As $u$
  varies in $\rD$, the vector $n_B(u)$ describes the parabola (shown
  in gray) obtained by intersecting the light cone with the plane
  defined in Minkowski 3-space by the equation $(X,B)=1$; one can
  think of the function $\omega$ as being defined on this
  parabola. The 3-vector $n_0$ corresponding to the apex of this
  parabola is shown in green. The figure shows the case when $B$ is
  future-lightlike.}
\label{fig:ProjLightlike}
\end{figure}

\begin{remark}
Equation \eqref{Vpsi} shows that the general solution of
\eqref{LambdaV} for a Hesse function of lightlike parameter $B$ is
obtained by acting on $\mD$ with the $\PSU(1,1)$ transformation
$U:=U(0,a,0)$, where $a$ is given in \eqref{paramlightlike}.  This
amounts to replacing $(\tau,\sigma)$ in expression
\eqref{VCanLightlike} by horocyclic coordinates $(\ttau,\tsigma)$
centered at the point $u_0\eqdef \psi_{U^{-1}}(\i)$ of the conformal
boundary of $\mD$ where the lightlike Hesse function $\Lambda_B$ tends
to zero. In the new coordinates, the curves $\tsigma=\const$ (which
coincide with the level sets of $\Lambda_B$) are horocycles centered
at $u_0$, while the curves $\ttau=\const$ (which coincide with the
gradient flow lines of $\Lambda_B$) are hyperbolic geodesics having
$u_B$ as a limit point. This gives:
\ben
\label{VLightlike}
V_B(x,y)=\omega(\ttau(x,y))\frac{\Lambda_B(x,y)^2}{B_0^2}~~,
\een
where $\ttau(x,y)\eqdef \tau(\psi_U(x+\i y))$ is given by:
\ben
\label{ttauLightlike}
\ttau(x,y)=\frac{2 (B_1 y- B_2 x)}{B_0 (1+\rho ^2)+2 B_1 x+2 B_2 y}~~.
\een
\end{remark}

\paragraph{Accidental visible symmetries in the general lightlike case.}

The potential \eqref{VLightlike} is stabilized by a non-trivial
continuous subgroup of $\PSU(1,1)\simeq \Iso_+(\mD)$ (and hence the
corresponding cosmological model also admits visible symmetries) iff
the function $\omega$ is constant on $\R$. In this case, the
stabilizer of $V_B$ is a parabolic subgroup isomorphic with $(\R,+)$
which coincides with the stabilizer of the 3-vector $B$ under the
adjoint representation:
\be
\Stab_{\PSU(1,1)}(V_B)=\Stab_{\PSU(1,1)}(B)\simeq (\R,+)~~.
\ee
We have:
\be
\Stab_{\PSU(1,1)}(B)=U^{-1}\cP U~~,
\ee
where $U=U(0,a,0)$. 

\subsection{Summary}
\label{subsec:DiskSummary}

\noindent The results of the previous subsections are summarized by
the following theorem, which gives a complete classification of
Hessian two-field models with scalar manifold $\mD_\beta$:

\begin{thm}
The space of Hesse functions of the hyperbolic disk is
3-dimensional. A basis of this space is given by the classical
Weierstrass coordinates $\Xi^0,\Xi^1,\Xi^2$ and the general Hesse
function has the form:
\ben
\label{LambdaB}
\Lambda_B(u)=(B,\Xi(u))~~\forall u\in \rD~~,
\een
where $\Xi=(\Xi^0,\Xi^1,\Xi^2):\rD\rightarrow \R^3$ is the Weierstrass
map, $B=(B^0,B^1,B^2)\in \R^3$ is an arbitrary non-vanishing 3-vector parameter
and $(~,~)$ is the Minkowski pairing of signature $(1,2)$ on $\R^3$. Moreover, 
the following statements hold for the weakly-Hessian two-field
cosmological model whose scalar manifold is the disk
$\mD_\beta=(\rD,\cG)$, where $\cG$ is the complete metric of constant
negative curvature $K=-\frac{3}{8}$:
\begin{enumerate}[1.]
\itemsep 0.0em
\item When $B$ is timelike, the two-field model with scalar manifold
  $\mD_\beta$ admits the Hessian symmetry generated by \eqref{LambdaB}
  iff the scalar potential $V$ has the form:
\ben
\label{VTimelike2}
V_B(u)=\omega(n_B(u))\left[\frac{\Lambda_B(u)^2}{(B,B)}-1\right]~~,
\een
where $\omega\in \cC^\infty(\rS^1)$ is an arbitrary smooth function
defined on the unit circle and $n_B(u)$ is the 3-vector given by:
\ben
n_B(u)\!=\!\frac{(B,B)\Xi(u)\!-\!(B,\Xi(u)) B}{\sqrt{(B,B)(B,\Xi(u))^2\!
-\!(B,B)^2}}\!=\!\frac{(B,B)\Xi(u)-B\Lambda_B(u)}{\sqrt{(B,B)\Lambda_B(u)^2\!
-\!(B,B)^2}}~,
\een
which lies on the circle of unit radius located in the spacelike plane
orthogonal to $B$ in $\R^{1,2}$ (see Figure \ref{fig:ProjTimelike}).
Here, the function $\omega$ is thought of as being defined on this
circle. The model also admits visible symmetries iff $\omega$ is
constant, in which case the group of visible symmetries is an elliptic
subgroup of $\PSU(1,1)$ conjugate to the canonical rotation subgroup
$\cR\simeq \U(1)$; moreover, the group of visible symmetries coincides
with the stabilizer of $B$ under the adjoint representation $\bAd_0$
of $\PSU(1,1)$.
\item When $B$ is spacelike, the two-field model with scalar manifold
  $\mD_\beta$ admits the Hessian symmetry generated by \eqref{LambdaB}
  iff its scalar potential $V$ has the form:
\ben
\label{VSpacelike2}
V_B(u)=\omega(n_B(u))\left[\frac{\Lambda_B(u)^2}{|(B,B)|}+1\right]~~,
\een
where $\omega\in \cC^\infty(\R)$ is an arbitrary smooth function
defined on the real line and $n_B(u)$ is the unit timelike 3-vector
given by:
\ben
n_B(u)\!=\!\frac{|(B,B)|\,\Xi(u)+(B,\Xi(u)) B}{\sqrt{(B,B)^2\!
+\!|(B,B)|(B,\Xi(u))^2}}\!=\!
\frac{|(B,B)|\,\Xi(u)+\Lambda_B(u) B}{\sqrt{(B,B)^2\!+\!|(B,B)|\Lambda_B(u)^2}}~,
\een
which lies on the hyperbola obtained by intersecting the unit
hyperboloid with the Minkowski plane orthogonal to $B$ (see Figure
\ref{fig:ProjSpacelike}).  Here, the function $\omega$ is thought of
as being defined on this hyperbola. The explicit form of $V$ in Euclidean
Cartesian coordinates is given in equations \eqref{VSpacelike} and
\eqref{ttauSpacelike}. The model also admits visible symmetries iff
$\omega$ is constant, in which case the group of visible symmetries is
a hyperbolic subgroup of $\PSU(1,1)$ conjugate to the canonical
squeeze subgroup $\cT\simeq (\R,+)$; moreover, the group of visible
symmetries coincides with the stabilizer of $B$ under the adjoint
representation $\bAd_0$ of $\PSU(1,1)$.
\item When $B$ is lightlike, the two-field model with scalar
  manifold $\mD_\beta$ admits the Hessian symmetry generated by
  \eqref{LambdaB} iff its scalar potential $V$ has the form:
\ben
\label{VLightlike2}
V_B(u)=\omega(B_0 n_B(u))\frac{\Lambda_B(u)^2}{B_0^2}~~,
\een
where $\omega\in\cC^\infty(\R)$ is an arbitrary smooth function defined
on the real line and $n_B(u)$ is the lightlike 3-vector given by:
\ben
n_B(u)=\frac{2(B,\Xi(u))\Xi(u)-B}{2(B,\Xi(u))^2}
=\frac{2\Lambda_B(u)\Xi(u)-B}{2\Lambda_B(u)^2}~~,
\een
which lies on the parabola obtained by intersecting the light cone of
$\R^{1,2}$ with the affine plane defined by the equation $(X,B)=1$
(see Figure \ref{fig:ProjLightlike}). Here, the function $\omega$ is
thought of as being defined on this parabola. The explicit form of $V$ in
Euclidean Cartesian coordinates is given in equations
\eqref{VLightlike} and \eqref{ttauLightlike}. The model also admits
visible symmetries iff $\omega$ is constant, in which case the group
of visible symmetries is a parabolic subgroup of $\PSU(1,1)$ conjugate
to the canonical shear subgroup $\cP\simeq (\R,+)$; moreover the group of
visible symmetries coincides with the stabilizer of $B$ under the
adjoint representation $\bAd_0$ of $\PSU(1,1)$.
\end{enumerate}
In each of the three cases, there exists an orientation-preserving
isometry of the scalar manifold which brings the Hesse generator and
the scalar potential to the corresponding canonical forms (see
equations \eqref{LambdaCanTimelike} and \eqref{VCanLightlike} for the
timelike case, \eqref{LambdaCanSpacelike} and \eqref{VCanSpacelike}
for the spacelike case, \eqref{LambdaCanLightlike} and
\eqref{VCanLightlike} for the lightlike case).
\end{thm}

\noindent Notice that $V_B$ depends only on the ray of the 3-vector $B$ in the
projective Minkowski space $\mathbb{P}\R^{1,2}$.  The explicit forms of the
  scalar potential in the three cases are as follows:
\vskip 0.2cm
\noindent $\bullet$ For timelike $B$ (i.e., for $(B,B)\eqdef B_0^2-B_1^2-B_2^2 > 0$):
\ben
\label{PotTimelike}
V_B(x,y)\!=\!\omega(\ttheta(x,y))\,\frac{P}{(B,B)(1-\rho^2)^2}~~,
\een
where:
\ben \label{P_expr}
P \!= \!(B_1^2+B_2^2)(1\!+\!\rho^4)\!+\!2(B_1^2 - B_2^2)(x^2\!-\!y^2)
\!+\!4B_0^2\rho^2\!+\!4B_0 (1\!+\!\rho^2)(B_1x\!+\!B_2y)\!+\!8 B_1 \!B_2 xy
\een
and:
\ben
\label{ArgOmegaTimelike}
\ttheta(x,y)=\arg\!\left[\frac{ \sign(B_0) \,(B_1 - \i B_2) \,(x+\i y) + \left( |B_0| - \sqrt{(B,B)} \right)}{ \left(|B_0| - \sqrt{(B,B)}\right) (x+\i y)+ \sign(B_0) \,(B_1+\i B_2)}\right]~~.
\een

\noindent $\bullet$ For spacelike $B$ (i.e., for $B_0^2-B_1^2-B_2^2 < 0$):
\ben
\label{PotSpacelike}
\!\!\!\!\!V_B(x,y)\!=\!\omega(\ttau(x,y))\,\frac{P}{|(B,B)|(1-\rho^2)^2}~,
\een
where $P$ is given by \eqref{P_expr} and:
\ben
\label{ArgOmegaSpacelike}
\ttau(x,y)=\arcsinh \!\left[ \frac{2 \,\sqrt{|(B,B)|} \,(B_1 y - B_2x)}{\sqrt{(B_1^2+B_2^2) P}} \right]~.
\een

\noindent $\bullet$ For lightlike $B$ (i.e., for $B_0^2-B_1^2-B_2^2 = 0$):
\ben
\label{PotLightlike}
V_B(x,y)=\omega(\ttau(x,y))\,\frac{\left(B_0(1+\rho^2)+2B_1 x+2B_2 y\right)^2}{B_0^2(1-\rho^2)^2}~~,
\een
where:
\ben
\label{ArgOmegaLightlike}
\ttau(x,y)=\frac{2 (B_1 y- B_2 x)}{B_0 (1+\rho ^2)+2 B_1 x+2 B_2 y}~~.
\een

\section{Hessian models for the hyperbolic punctured disk}
\label{sec:pDisk}

\noindent In this case, all Hesse functions are
rotationally-invariant. Taking $C=1$ in \eqref{LambdamDast}, we find
that the gradient vector field of the Hesse function
$\Lambda=e^{-\beta r}$ has the following components in the rescaled
normal polar coordinates $(r,\theta)$:
\beqan
\label{gradLambdamDast}
&& (\grad_\cG \Lambda)^r=-\beta e^{-\beta r} \nn\\
&& (\grad_\cG \Lambda)^\theta=0~~.
\eeqan
Since $\rho=e^{-2\pi e^{\beta r}}$, the level sets of $\Lambda$ are
Euclidean circles centered at the origin of $\mD^\ast$, while the
gradient flow curves are half lines passing through the origin (which
corresponds to $r\rightarrow +\infty$); the gradient curves flow from
the outer component of the conformal boundary of $\mD^\ast$, which is
the Euclidean circle of radius $1$ corresponding to $r\rightarrow
-\infty$. The gradient flow equations of $\Lambda$:
\be
\frac{\dd r}{\dd q}=\beta e^{-\beta r}~~,~~\frac{\dd \theta}{\dd q}=0
\ee
give $\theta=\const$ and:
\ben
\label{gradflowmDast}
e^{\beta r}=1+\beta^2 q~\Longleftrightarrow ~q
=\frac{1}{\beta^2}\left(e^{\beta r}-1\right)~~,
\een
where we chose the constant of integration such that $r|_{q=0}=0$,
i.e. such that $\rho|_{q=0}=e^{-2\pi}$; this amounts to using the
Euclidean circle $C_0$ of radius $\rho_0\eqdef e^{-2\pi}$ (which has
unit hyperbolic circumference) as a section $\cQ$ for the gradient
flow. We have:
\be
\int_0^q \Lambda(\gamma(q'))\dd q'=\int_0^q  e^{-\beta r(q')} \dd q'
=\int_0^q \frac{1}{1+ \beta^2  q'} \dd q'
=\frac{1}{\beta^2} \log\left(1+\beta^2 q\right)=\frac{1}{\beta} r~~.
\ee
Relation \eqref{Vgammaq} gives:
\ben
\label{VSolmDast}
V=\omega(\theta)e^{-2 \beta r}=\omega(\theta)e^{-\sqrt{\frac{3}{2}} r}
=\omega(\theta)\frac{4\pi^2}{(\log\rho)^2}~~,
\een
where we used equation \eqref{rrhopDisk} and we defined
$\omega(\theta)\eqdef
V(r,\theta)|_{r=0}=V(\rho,\theta)|_{\rho=\rho_0}$. Notice that
$\omega$ (which can be viewed as a smooth function defined on the unit
circle) can be identified with the restriction of $V$ to the circle
$C_0$, which plays the role of section for the gradient flow of
$\Lambda$.

The space of Killing vector fields of $\mD^\ast$ is generated by
$\pd_\theta$, which is a visible symmetry iff $\pd_\theta V=0$, which
amounts to the condition that $\omega$ be independent of
$\theta$. The following statement summarizes these results:

\begin{thm}
The space of Hesse functions of the hyperbolic punctured disk is
one-dimensional, being generated by:
\ben
\label{LambdaSolpd}
\Lambda=e^{-\beta r}=\frac{2\pi}{|\log \rho|}~~,
\een
where $(r,\theta)$ are polar semi-geodesic coordinates for the complete
metric $\cG$ of Gaussian curvature $-\beta^2=-3/8$. For $C\neq 0$,
this Hesse function generates a Hessian symmetry of the two-field
cosmological model with scalar manifold $\mD^\ast_\beta$ iff the
scalar potential has the form:
\ben
\label{VSolpd}
V(r,\theta)=\omega(\theta)e^{-2\beta r}=\omega(\theta)\frac{4\pi^2}{(\log\rho)^2}~,
\een
where $\omega\in \cC^\infty(\rS^1)$ is an arbitrary smooth function
defined on the unit circle (viewed as a $2\pi$-periodic
smooth function of the polar angle $\theta$). In this case, the
corresponding Hessian symmetry is generated by the vector field:
\be
X_\Lambda=\frac{\Lambda}{\sqrt{a}}\pd_a-\frac{4}{a^{3/2}}\grad \Lambda
=e^{-\beta r}\Big(\frac{1}{\sqrt{a}} \pd_a +\frac{4\beta}{a^{3/2}}\pd_r\Big)
=2\pi\left[\frac{1}{\sqrt{a}|\log\rho|}\pd_a-4\beta^2\frac{\rho}{a^{3/2}}\pd_\rho\right]~~.
\ee
When $\omega$ is not constant, the space of Noether symmetries of such
a model is one-dimensional and coincides with the space of Hessian
symmetries, being spanned by the vector field $X_\Lambda$.  When
$\omega$ is constant, the model also admits visible symmetries, whose
generators form a one-dimensional vector space spanned by
$\pd_\theta$. In this special case, the space of Noether symmetries is
two-dimensional and admits a basis given by the vector fields
$X_\Lambda$ and $\pd_\theta$.
\end{thm}

\

\noindent The radial profiles of $\Lambda$ and $V$ are plotted in
Figure \ref{fig:pDiskLambdaV}.

\begin{figure}[H]
\centering
\begin{minipage}{.44\textwidth}
\centering \includegraphics[width=\linewidth]{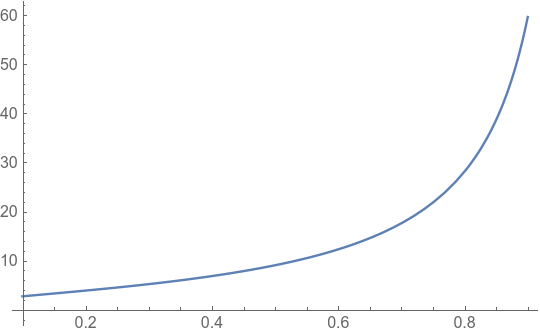}
\vskip 0.45em \subcaption{Plot of $\Lambda(\rho)/C$.}
\label{fig:pDiskLambda}
\end{minipage}\hfill
\begin{minipage}{.44\textwidth}
\centering \includegraphics[width=1\linewidth]{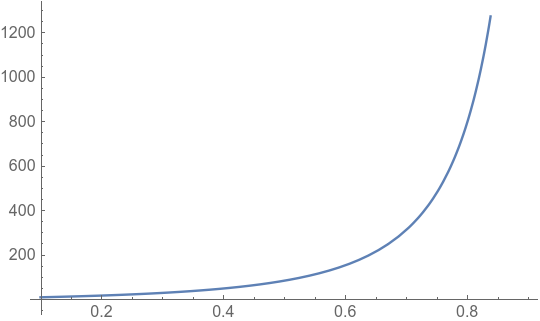}
\vskip 0.45em \subcaption{Plot of $V(\rho,\theta)/\omega(\theta)$.}
\label{fig:pDiskV}
\end{minipage}
\caption{Plot of the radial profiles of $\Lambda/C$ and $V$ for the hyperbolic punctured disk.}
\label{fig:pDiskLambdaV}
\end{figure}

\section{Hessian models for the hyperbolic annuli}
\label{sec:Annuli}

\noindent In this case, all Hesse functions are 
rotationally-invariant. Using \eqref{LambdamA} with $C\!=\!1$, we find
that the gradient vector field of $\Lambda\!=\!\sinh(\beta r)$ has the
following components in normal polar coordinates $(r,\theta)$ for the
metric $\cG$ of Gaussian curvature $-\beta^2$:
\beqan
\label{gradLambdamA}
&& (\grad_\cG \Lambda)^r=\beta \cosh(\beta r) \nn\\
&& (\grad_\cG \Lambda)^\theta=0~~.
\eeqan
The gradient flow equations of $\Lambda$:
\ben \label{gradFlEq}
\frac{\dd r}{\dd q}=-\beta \cosh(\beta r)~~,~~\frac{\dd \theta}{\dd q}=0
\een
give $\theta=\theta_0=\const$ and:
\ben
\label{gfA}
\tanh(\beta r/2)=-\tan(\beta^2 q/2)~~,
\een
where we chose the gradient flow parameter such that $r|_{q=0}=0$
(i.e. $\rho|_{q=0}=1$) and we used the formula:
\be
\int \frac{\dd r}{\cosh(\beta r)}=\frac{2}{\beta}\arctan
\left[\tanh\left(\frac{\beta r}{2}\right)\right]+\const~~.
\ee
In this case, the section $\cQ$ for the gradient flow of $\Lambda$
is the Euclidean circle $C_1$ of radius $\rho=1$ (which separates the two
funnel regions of $\mA(R)$). Using \eqref{gradFlEq}, more precisely:
\be
\dd q'=-\frac{1}{\beta \cosh(\beta r')}\dd r'~~,
\ee
gives:
\beqa
&& \int_0^q \Lambda(\gamma(q'))\dd q'= \int_0^q  \sinh(\beta r(q'))\dd q'
=-\frac{1}{\beta} \int_0^r  \tanh(\beta r')\dd r'
=-\frac{1}{\beta^2}\log\left[\cosh(\beta r)\right]~~.
\eeqa
Hence, relation \eqref{Vgammaq} implies:
\ben
\label{VmA}
V(r,\theta)=\omega(\theta) e^{2\log \left[\cosh (\beta r)\right] }
=\omega(\theta)\cosh^2(\beta r)
=\frac{\omega(\theta)}{ \cos^2\left(\frac{\pi}{\mu}|\log\rho|\right)}~~,
\een
where we defined $\omega(\theta)\eqdef
V(r,\theta)|_{r=0}=V(\rho,\theta)|_{\rho=1}$. Notice that $\omega$ can
be viewed as a smooth function defined on the unit circle, which
identifies with the restriction of $V$ to the Euclidean circle $C_1$.

The space of Killing vector fields of $\mA(R)$ is generated by
$\pd_\theta$. The latter is a visible symmetry iff $\pd_\theta V=0$,
which amounts to the condition that $\omega$ be independent of
$\theta$. Summarizing everything, we have:

\begin{thm}
The space of Hesse functions of the hyperbolic annulus $\mA(R)$ is
one-dimensional, being generated by the function:
\ben
\label{LambdaSolan}
\Lambda=\sinh(\beta r)=\tan\left(\frac{\pi}{\mu}\log\rho\right)~~,
\een
where $(r,\theta)$ are polar semi-geodesic coordinates for the complete
metric $\cG$ of Gaussian curvature $-\beta^2=-3/8$. This function
generates a Hessian symmetry of the two-field cosmological model with
scalar manifold $\mA_\beta(R)$ iff the scalar potential has the form:
\ben
\label{VSolan}
V(r,\theta)=\omega(\theta)\cosh^2(\beta r)
=\frac{\omega(\theta)}{\cos^2\left(\frac{\pi}{\mu}|\log\rho|\right)}~~,
\een
where $\omega\in \cC^\infty(\rS^1)$ is an arbitrary smooth function
defined on the unit circle (viewed as a $2\pi$-periodic
smooth function of the polar angle $\theta$). In this case, the
corresponding Hessian symmetry is generated by the vector field:
\beqa
X_\Lambda&=&\frac{\Lambda}{\sqrt{a}}\pd_a-\frac{4}{a^{3/2}}\grad
\Lambda= \frac{\sinh(\beta r)}{\sqrt{a}}
  \pd_a-4\beta\frac{\cosh(\beta r)}{a^{3/2}}\pd_r \nn\\
&=&\frac{\tan\left(\frac{\pi}{\mu}\log\rho\right)}{\sqrt{a}}\pd_a
-\left(\frac{4\beta^2\mu}{\pi}\right)\frac{\rho}{a^{3/2}}\pd_\rho~~.
\eeqa
When $\omega$ is not constant, the space of Noether symmetries of such
a model is one-dimensional and coincides with the space of Hessian
symmetries, being spanned by the vector field $X_\Lambda$. When
$\omega$ is constant on $\rS^1$, the model also admits visible
symmetries, whose generators form a one-dimensional linear space
spanned by $\pd_\theta$. In this case, the space of Noether symmetries
is two-dimensional and admits a basis given by the vector fields
$X_\Lambda$ and $\pd_\theta$.
\end{thm}

\

\noindent The radial profiles of $\Lambda$ and $V$ are plotted in
Figure \ref{fig:pDiskLambdaV} for the hyperbolic annulus of modulus
$\mu=2\log 2$ (i.e. $R=2$).

\begin{figure}[H]
\centering
\begin{minipage}{.44\textwidth}
\centering \includegraphics[width=\linewidth]{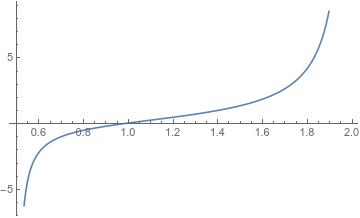}
\vskip 0.45em \subcaption{Plot of $\Lambda(\rho)/C$}
\label{fig:AnnulusLambda}
\end{minipage}\hfill
\begin{minipage}{.44\textwidth}
\centering \includegraphics[width=1\linewidth]{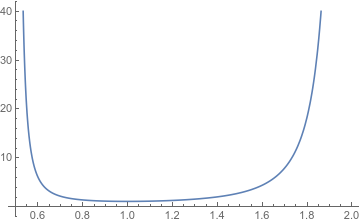}
\vskip 0.45em \subcaption{Plot of $V(\rho,\theta)/\omega(\theta)$}
\label{fig:AnnulusV}
\end{minipage}
\caption{Plot of the radial profiles of $\Lambda/C$ and $V$ on the
  hyperbolic annulus of modulus $\mu=2\log 2$ (i.e. for $R=2$). In
  this case $\rho$ runs from $1/2$ to $2$.}
\label{fig:AnnulusLambdaV}
\end{figure}

\section{Conclusions and further directions}
\label{sec:conclusions}

\noindent We studied time-independent Noether symmetries in two-field
cosmological models, showing that any such symmetry decomposes as a
direct sum of a visible and a Hessian symmetry. While visible
symmetries correspond to those isometries of the scalar manifold which
preserve the scalar potential (and in this sense are ``obvious''
symmetries), Hessian symmetries are ``hidden'' in the sense that they
are not apparent upon direct inspection.  We showed that any Hessian
symmetry is determined by a generating function $\Lambda$. The latter
is a {\em Hesse function} of the scalar manifold
$(\Sigma,G=\frac{3}{8}\cG)$, i.e.  a real-valued function $\Lambda$
defined on $\Sigma$ and which obeys the {\em Hesse equation} of
$(\Sigma,G)$ (a certain second order linear PDE for $\Lambda$ which
involves the rescaled scalar manifold metric
$G=\frac{3}{8}\cG$). Moreover, the scalar potential $V$ of a model
which admits a Hessian symmetry must obey the $\Lambda$-$V$ equation
(a certain first order PDE for $V$ which involves $\Lambda$ and
$\cG$).

When the scalar manifold metric $\cG$ is rotationally invariant, we
showed that the two-field model admits a Hessian symmetry iff $\Sigma$
is a disk, a punctured disk or an annulus and $\cG$ is a complete
metric of Gaussian curvature $K=-\frac{3}{8}$, i.e. iff the model is
an elementary two-field $\alpha$-attractor in the sense of
reference \cite{elem}, for the particular value $\alpha=8/9$ of the
$\alpha$-parameter. In all cases, we determined the explicit general
form of the scalar potential $V$ which is compatible with a given
Hessian symmetry. We also discussed the special cases when such a model
also admits a visible symmetry. Finally, we discussed the integral of
motion of a Hessian symmetry --- showing that it allows one to
simplify the computation of the number of e-folds along cosmological
trajectories.

The present paper opens up a few avenues for further research, some of
which we plan to address in future work. First, we will show in a
separate paper (using a more general framework) that the
classification of Hessian models given in this paper is in fact valid
without assuming rotational invariance of the scalar manifold
metric. One can also show that the existence of a Hessian symmetry
enables an effective one-field model description (as far as one is
concerned with determining classical trajectories) for each fixed
value of the corresponding integral of motion\footnote{Although, of
  course, the fluctuations of both real scalar fields would be
  important, for example, for addressing the issue of perturbative
  stability of a given trajectory.}, a fact which has interesting
implications for contact with observational data. Furthermore, the
approach of the present paper can be extended to the study of
symmetries in $n$-field cosmological models, for which it leads to a
rich mathematical theory.

Another direction for future studies concerns the possible
embeddings of such models into supergravity or string theories, where
we expect them to arise as points of ``enhanced symmetry'' in the
moduli spaces of various compactifications. It is also worth noting that in
recent years there have been a number of investigations of novel
behavior arising from non-trivial angular motion in two-field models on
the hyperbolic disk (see references
\cite{KL7,LWWYA,DFRSW,ARB,MM,CRS,GSRPR,TB}). Our results provide 
a vast arena for even deeper and more involved studies along those lines.
Indeed, having a Noether symmetry enables one to find exact (as opposed 
to numerical) solutions of the cosmological equations of motion, in particular 
obtaining explicit expressions for the Hubble parameter as a function of time; 
see \cite{ABL} (as well as \cite{CDeF} and references therein, in the 
context of extended theories of gravity). This would facilitate investigating 
with analytical means a variety of new regimes of expansion. 
It would be especially interesting to find new {\it non-slow-roll}
inflationary regimes, which are perturbatively stable and produce
a nearly scale-invariant spectrum of fluctuations (as needed for
consistency with observations). Even for single-field models, such a
regime was established only relatively recently \cite{MSY,ASW}. For
two-field models the problem is more challenging, but may also present
new opportunities. As already mentioned, the presence of a Noether 
symmetry in the class of models, considered in the present paper, may 
prove to be of great help in that regard.

It would also be interesting to explore whether the present work can
be useful for a wider program (which was touched upon briefly in
reference \cite{flows}) aimed at studying multifield cosmological
models with methods from the geometric theory of dynamical systems
(see \cite{Palis} for an introduction). As pointed out in
\cite{flows}, the dynamics of such models is quite rich and in
particular it is amenable to certain methods originating in asymptotic
analysis. It would be interesting to gain a deeper understanding of
the simplifications which the presence of a Hessian symmetry may
afford in that context. We hope to report on these and related
problems in future work.

\acknowledgments{\noindent We are grateful to A. Hebecker, T. Van Riet, 
T. Wrase and E. Colgain for interesting discussions on various aspects of 
inflation, cosmology and the swampland conjectures. L.A. and E.M.B. also 
thank the Center for Geometry and Physics of the Institute for Basic Science 
in Pohang for hospitality during the completion of this work. L.A. has received 
partial support from the Bulgarian NSF grant DN 08/3, while E.M.B. acknowledges  
support from the Romanian Ministry of Research and Innovation, contract number 
PN 18090101/2019. The work of C.I.L. was supported by grant IBS-R003-D1.}

\appendix

\section{Geometric formulation of the method of characteristics}
\label{app:char}

\noindent In this appendix, we recall the geometric formulation of the 
method of characteristics for solving a first order PDE of the form: 
\ben
\label{eq2}
\iota_X\dd f=g~~
\een
for an unknown smooth function $f\in \cC^\infty(\Sigma)$ defined on a
manifold $\Sigma$, where $X$ is a vector field given on $\Sigma$, $g\in
\cC^\infty(\Sigma)$ is a given function and $\iota_X$ denotes
contraction with $X$. The method relies on the observation
that the identity $L_X=\dd \iota_X+\iota_X \dd$ allows us to write
\eqref{eq2} in the equivalent form:
\ben
\label{eqL}
L_Xf=g~~,
\een
where $L_X$ is the Lie derivative with respect to $X$. This shows that
$f$ is determined by the flow of $X$ as follows. If
$\gamma:[t_1,t_2]\rightarrow \Sigma$ is a flow curve of $X$
(i.e. $\frac{\dd\gamma(t)}{\dd t}=X(\gamma(t))$), then \eqref{eqL}
gives\footnote{Recall that $(L_Xf)(\gamma(t))=X(f)(\gamma(t))=(\dd_{\gamma(t)}
  f)(X(\gamma(t)))=(\dd_{\gamma(t)}
  f)\left(\frac{\dd \gamma(t)}{\dd t}\right)=\frac{\dd}{\dd
    t}\left[f(\gamma(t))\right]$.}:
\ben
\label{fsol}
f(\gamma(t_2))-f(\gamma(t_1))=\int_{t_1}^{t_2} g(\gamma(t))\dd t~~,
\een
which allows one to determine $f$ if the flow of the vector field
$X$ is known.

As an example, notice that equation \eqref{LambdaV} can
be written in the form \eqref{eq2} by setting $f=\log V$, $X=\grad_\cG
\log \Lambda=\frac{\grad_\cG \Lambda}{\Lambda}$ and $g=2\beta^2$, in which case one can easily see that
\eqref{fsol} is equivalent with \eqref{Vgammaq} upon taking into
account that the flow parameter $t$ of the vector field $\grad_\cG\log \Lambda$ is related
to the gradient flow parameter $q$ of $\Lambda$ through $\dd t=-\Lambda \dd q$.

\section{Orientation-preserving isometries of the hyperbolic disk}
\label{app:iso}

\noindent The group $\Iso_+(\mD)=\Iso_o(\mD)$ of orientation-preserving
isometries of the Poicar\'e disk identifies naturally with the group
$\PSU(1,1)$ as well as with the connected component $\SO_o(1,2)$ of
the Lorentz group in three dimensions. In this appendix, we recall
these well-known identifications in the conventions used in the
present paper.

\paragraph{The group $\SU(1,1)$.}
Consider the matrix:
\be
J\eqdef \left[\begin{array}{cc}1 & 0 \\ 0 &-1\end{array}\right]~~,
\ee
which satisfies $J^\dagger =J=J^{-1}$. Recall that $\SU(1,1)$ is the closed
subgroup of $\SL(2,\C)$ defined through:
\be
\SU(1,1)\eqdef \{U \in \Mat(2,\C)\, \Big{|} \, U^\dagger =JU^{-1}J~~\&~~\det U=+1\}~~,
\ee
Let $Q\eqdef \left[\begin{array}{cc} 1 & \i \\ \i & 1\end{array} \right]$.
Then $\SU(1,1)$ can be identified with $\SL(2,\R)$ through the {\em Cayley isomorphism}:
\ben
\label{Cayley}
\SU(1,1)\ni U \rightarrow Q U Q^{-1}\in \SL(2,\R)~~.
\een

\paragraph{The complex parameterization of $\SU(1,1)$.}
We have:
\be
\SU(1,1)=\{U(\eta,\sigma) \, \Big{|} \, \eta, \sigma\in \C: |\eta|^2-|\sigma|^2=1\}~~,
\ee
where: 
\be
U(\eta,\sigma)\eqdef \left[\begin{array}{cc}\eta & \sigma \\ 
\bar{\sigma} & \bar{\eta} \end{array}\right]~~. 
\ee
The following relations hold in this parameterization:
\be
U(\eta,\sigma)^{-1}=U(\bar{\eta},-\sigma)~~,~~U(\eta,\sigma)^\dagger=U(\bar{\eta},\sigma)~~,
~~(U(\eta,\sigma)^{-1})^\dagger=U(\eta,-\sigma)~~
\ee

\paragraph{The group of orientation-preserving isometries of $\mD$.}

Consider the non-linear action (i.e. morphism of groups
$\psi:\SU(1,1)\rightarrow \Diff(\rD)$) given by fractional
transformations:
\ben
\label{IsoActionD}
\psi_U(u)=\frac{\eta u +\sigma}{\bar{\sigma} u +\bar{\eta}}~~~(u\in \rD)~~,
\een
where $\psi_U\eqdef \psi(U)$. This action is non-effective with kernel
given by $\{-I_2,I_2\}$.  It descends to an effective action of the
group $\PSU(1,1)\eqdef \SU(1,1)/\{-I_2,I_2\}$, which we denote by
$\bpsi:\PSU(1,1)\rightarrow \Diff(\rD)$.  Then the image of $\bpsi$
coincides with the group $\Iso_+(\mD)=\Iso_o(\mD)$ of
orientation-preserving isometries of the Poincar\'e disk:
\be
\bpsi(\PSU(1,1))=\Iso_+(\mD)
\ee
and the isomorphism of groups obtained by co-restricting $\bpsi$ to
its image intertwines the action of $\PSU(1,1)$ by fractional
transformations and the tautological action of $\Iso_+(\mD)$ on
$\rD$. Thus $\Iso_+(\mD)$ identifies with $\PSU(1,1)$ and its
tautological action identifies with the fractional
action of the latter. 

\begin{remark}
Notice the relation:
\ben
\label{dU}
\pd_u [\psi_U(u)]=\frac{1}{(\bar{\sigma} u+\bar{\eta})^2}~~~~\forall u\in \rD~~.
\een
\end{remark}

\paragraph{The angle-boost parameterization of $\SU(1,1)$.}

One can also parameterize the elements of $\SU(1,1)$ by
unconstrained quantities $t\in \R_{\geq 0}$ and $a,b\in \R/(2\pi\Z)$ defined
through:
\be
\eta=\cosh(t/2) e^{\i a}~~,~~\sigma=\sinh(t/2)e^{\i b}~~,
\ee
so that: 
\be
U(t,a,b)=\left[\begin{array}{cc} \cosh(t/2) e^{\i a} & \sinh(t/2)e^{\i b}  \\  
\sinh(t/2)e^{-\i b}  & \cosh(t/2) e^{-\i a} \end{array}\right]~~.
\ee
Then $t$ is called the {\em boost parameter} of $U$, while
$a$ and $b$ are called its {\em angle parameters}.
Notice the relations:
\be
U(t,a,b)^{-1}=U(t,-a,\pi+b)~~,~~U(t,a,b)^\dagger=U(t,-a,b)~~,
~~(U(t,a,b)^{-1})^\dagger=U(t,a,\pi+b)~~.
\ee
and
\ben
\label{Udec}
U(t,a,b)=R(a+b)T(t)R(a-b)~~.
\een

\paragraph{Canonical subgroups.}

The map $R:\R/(4\pi \Z)\simeq \U(1) \rightarrow \SU(1,1)$ defined through:
\be
R(\theta)=\left[\begin{array}{cc} e^{\frac{\i\theta}{2}} & 0 \\
 0 & e^{-\frac{\i\theta}{2}}\end{array}\right]\in \SU(1,1)~~(\theta\in \R/(4\pi \Z) )
\ee
is an injective morphism of groups whose image $\cR$ (called the {\em
  subgroup of rotations}) is the $\U(1)$ subgroup of $\SU(1,1)$
defined by $t=0$ (with $a=\theta/2$), which acts on $\mD$ by rotations
around the origin:
\be
R(\theta)\bullet u= u e^{\i \theta}~~.
\ee
This coincides with the elliptic subgroup of all transformations which fix 
the origin of $\rD$. 
Since the map $R$ is a deformation retract, we have $\pi_1(\SU(1,1))\simeq 
\pi_1(\U(1))\simeq \Z$. 

On the other hand, the map $T:\R\rightarrow \SU(1,1)$ defined through:
\ben
\label{T}
T(t)\eqdef \left[\begin{array}{cc} \cosh(t/2)~ & \sinh(t/2) \\ \sinh(t/2)~ & 
\cosh(t/2) \end{array}\right]~~(t\in \R)
\een
is an injective morphism of groups from $(\R,+)$ to $\SU(1,1)$ whose
image $\cT$ is called the {\em squeeze subgroup}. The squeeze subgroup
is the hyperbolic subgroup consisting of all transformations which fix
the points $u=+1$ and $u=-1$ on the conformal boundary of $\mD$.

Finally, the map $P:\R\rightarrow \SU(1,1)$ defined through:
\ben
\label{P}
P(\kappa)\eqdef  \left[\begin{array}{cc} 1+\i \kappa ~~&~~ \kappa \\
 -\kappa ~~&~~ 1-\i \kappa \end{array}\right]
\een
is an injective morphism of groups from $(\R,+)$ to $\SU(1,1)$ whose
image $\cP$ is called the {\em shear subgroup}. The sheer subgroup
coincides with the parabolic subgroup of $\SU(1,1)$ consisting of all
transformations which fix the point $u=\i$ on the conformal boundary
of $\mD$.

\begin{remark}
An arbitrary parabolic element $\Pi\in \SU(1,1)$ can be parameterized as:
\be
\Pi=\Pi(\kappa,\psi)\eqdef \left[\begin{array}{cc} 1+\i \kappa ~~&~~ 
-\i \kappa e^{\i \psi} \\ \i \kappa e^{-\i \psi} ~~&~~ 1-\i \kappa \end{array}\right]
\ee
where $\kappa\in \R$ and $\psi\in \R/(2\pi \Z)$. In this parameterization,
 we have $P(\kappa)=\Pi(\kappa,\pi/2)$. 
\end{remark}

\paragraph{Conjugacy classes of $\SU(1,1)$.}
Any elliptic element $\PSU(1,1)$ conjugates in $\PSU(1,1)$ to a
rotation, while any hyperbolic element is conjugate to a
boost. Moreover, any parabolic element conjugates to $P_0$ or $-P_0$, where:
\be
P_0\eqdef \Pi(1/2,0)=\left[\begin{array}{cc} 1+\i/2  ~~&~~ -\i/2 \\ \i/2  
~~&~~ 1-\i/2 \end{array}\right]
\ee
More precisely, we have:
\begin{enumerate}[1.]
\itemsep 0.0em
\item If $E\in \SU(1,1)$ is elliptic, then $E=VR(\theta)V^{-1}$ with
  $V=\frac{1}{\sqrt{1-|\alpha|^2}}\left[\begin{array}{cc} 1 ~~&~~
    \alpha \\ \bar{\alpha}~~&~~ 1 \end{array}\right]$, where
  $\alpha\in \C$ and $\theta\in \R$. In this case, we have
  $\tr(E)=2\cos(\frac{\theta}{2})$.
\item If $H\!\in\! \SU(1,1)$ is hyperbolic, then $H\!\!=\!VT(t)V^{-1}$ with
  $V\!\!=\!\frac{1}{\sqrt{1\!-\!|\alpha|^2}}\!\left[\!\begin{array}{cc}
    e^{-\i\frac{\theta}{2}}&~ \alpha e^{-\i\frac{\theta}{2}} \\ \bar{\alpha}
    e^{\i\frac{\theta}{2}}&~ e^{\i\frac{\theta}{2}} \end{array}\!\right]$, where
  $\alpha\in \C$, $t>0$ and $\theta\in \R$. In this case, we have
  $\tr(H)=2\cosh(\frac{t}{2})$.
\item If $\Pi=\Pi(\kappa,\psi)\in \SU(1,1)$ is parabolic, then $\Pi=\pm V P_0
  V^{-1}$ with $V=U(x,\psi/2,\psi/2)=\left[\begin{array}{cc}
      e^{\i\psi/2}\cosh(x) &~ e^{\i\psi/2}\sinh(x)
      \\ e^{-\i\psi/2}\sinh(x) &~ e^{-\i\psi/2}
      \cosh(x)\end{array}\right]$, where $x=\log(2\kappa)$. In this case,
  we have $\tr(\Pi)=\pm 2$.
\end{enumerate}

\paragraph{The Lie algebra and adjoint representation of $\SU(1,1)$.}

The 3-dimensional real Lie algebra of $\SU(1,1)$ is given by: 
\be
\su(1,1)=\{A\in \Mat(2,\C)\,|\, A^\dagger =-J A J~\,\&\,~\tr(A)=0\}~~,
\ee
with the Killing form\footnote{The isomorphism of groups
  \eqref{Cayley} induces an isometry between the Lie algebras
  $\su(1,1)$ and $\sl(2,\R)$, where each Lie algebra is viewed as a
  quadratic space when endowed with its Killing form.}:
\ben
\label{K}
(A, A')_\cK=4\tr(A A')~~,
\een
which is non-degenerate and of signature $(2,1)$. When endowed with
this pairing, the Lie algebra of $\SU(1,1)$ becomes a
three-dimensional Minkowski space and the adjoint representation
$\Ad:\SU(1,1)\rightarrow \Aut_\R(\su(1,1))$:
\ben
\label{Ad}
\Ad(U)(A)\eqdef U A U^{-1}~~,~~\forall U\in \SU(1,1)~,~~\forall A\in \su(1,1)~~.
\een
identifies with the group of proper and orthochronous
Lorentz transformations of this space. It is convenient to perform 
these identifications in two steps.

First, consider the linear isomorphism $F_J:\Mat(2,\C)\stackrel{\sim}{\rightarrow} 
\Mat(2,\C)$ defined through:
\be
F_J(A)\eqdef \frac{\i}{\sqrt{8}} AJ~~,
\ee
which identifies $\su(1,1)$ with the following linear subspace of $\Mat(2,\C)$:
\be
\su_J(1,1)\eqdef F_J(\su(1,1))=\{Z\in \Mat(2,\C)\,|\, Z^\dagger=Z~\,\&\,~\tr(JZ)=0\}
\ee
and transports the Killing form \eqref{K} to the opposite of the
following bilinear form defined on $\su_J(1,1)$:
\ben
\label{KJ}
(Z,Z')_J\eqdef \frac{1}{2}\tr(J Z J Z')~~.
\een
Next, notice that a matrix $Z\in \Mat(2,\C)$ satisfies the two conditions $Z^\dagger =Z$
and $\tr(JZ)=0$ iff it can be written as:
\be
Z=Z(X)\eqdef \left[\begin{array}{cc} X^0 & X^1+\i X^2 \\ X^1-\i X^2 & X^0\end{array}\right]~~,
\ee
for some unique real 3-vector $X\eqdef (X^0,X^1,X^2)\in \R^3$. This
gives a linear isomorphism
$Z:\R^3\stackrel{\sim}{\rightarrow}\su_J(1,1)$. The bilinear form
\eqref{KJ} corresponds through this isomorphism to the canonical
Minkowski pairing of signature $(1,2)$ on $\R^3$:
\ben
\label{MinkowskiD}
(Z(X),Z(Y))_J= (X,Y) \eqdef X^0Y^0-X^1Y^1-X^2Y^2~~.
\een
Hence $Z$ allows us to identify $(\su_J(1,1),(~,~)_J)$ with the
three-dimensional Minkowski space $\R^{1,2}=(\R^3,(~,~))$.  Combining
the above shows that the composite map $Z\circ
F_J:(\su(1,1),(~,~)_\cK)\rightarrow \R^{1,2}$ is an isomorphism of
quadratic spaces.

The linear isomorphism $F_J:\su(1,1)\rightarrow \su_J(1,1)$ transports
\eqref{Ad} into the equivalent representation $\Ad_J\eqdef F_J \circ
\Ad(U)\circ F_J^{-1}:\SU(1,1)\rightarrow \Aut_\R(\su_J(1,1))$, which
acts through:
\be
\Ad_J(U)(Z)= U Z J U^{-1}J=U Z U^\dagger~~.
\ee
Since $F_J$ is an isometry and the adjoint representation preserves the Killing form
\eqref{K}, it follows that $\Ad_J$ preserves the bilinear 
form \eqref{KJ}:
\be
(\Ad_J(U)(Z),\Ad_J(U)(Z'))_J=(Z,Z')_J~~,~~\forall Z,Z'\in \su_J(1,1)~,
~~\forall U\in \SU(1,1)~~.
\ee
Now the linear isomorphism
$Z:\R^3\stackrel{\sim}{\rightarrow}\su_J(1,1)$ transports $\Ad_J$ to
an equivalent representation $\Ad_0\eqdef Z^{-1}\circ \Ad_J(U)\circ
Z:\SU(1,1)\rightarrow \Aut_\R(\R^3)$, which acts though:
\ben
\label{Ad0Z}
Z(\Ad_0(U)(X))=\Ad_J(Z(X))=UZ(X)U^\dagger~,~~\forall X\in \R^3~,~~\forall U\in \SU(1,1)~~.
\een
Since $Z$ is also an isometry, it follows that $\Ad_0$ preserves the Minkowski pairing
\eqref{MinkowskiD}:
\be
(\Ad_0(U)(X), \Ad_0(U)(Y))=(X,Y)~~,~~\forall X, Y\in \R^3~,~~\forall U\in \SU(1,1)~~.
\ee
Hence the operators $\Ad_0(U)$ are three-dimensional Lorentz
transformations. The fact that $(~,~)$ is $\Ad_0$-invariant amounts to
the $(~,~)$-orthogonality condition:
\ben
\label{Ad0Ort}
\Ad_0(U)^\vee=\Ad_0(U)^{-1}~~\mathrm{i.e.}~~\Ad_0(U)^\vee=\Ad_0(U^{-1})~~,
~~\forall U\in \SU(1,1)~~,
\een
where $A^\vee\in \End_\R(\R^3)$ denotes the adjoint of a linear
operator $A\in \End_\R(\R^3)$ with respect to the Minkowski pairing
\eqref{MinkowskiD}. Since $\PSU(1,1)$ is connected, the image of
$\Ad_0$ coincides with the connected component of the identity
$\SO_o(1,2)$, which is the group of proper and orthochronous Lorentz
transformations in three space-time dimensions.
Since $\Ad_0(-I_2)=\id_{\R^3}$, we have an induced morphism
of groups:
\ben
\label{barAd0}
\bAd_0:\PSU(1,1)\rightarrow \SO_o(1,2)~~,
\een
It is a classical fact that \eqref{barAd0} is an isomorphism of groups.

\begin{remark}
Notice the relation:
\be
\det Z(X)=(Z(X),Z(X))_J=(X,X)=(X^0)^2-(X^1)^2-(X^2)^2~~.
\ee
\end{remark}

\paragraph{Explicit expressions for $\Ad_0(U)$.}

The explicit form of the morphism of groups
$\Ad_0:\SU(1,1)\rightarrow \SO_0(1,1)$ can be determined using
relation \eqref{Ad0Z}.  In the complex parameterization of $\SU(1,1)$,
this gives\footnote{Notice that $\Im(\bar{\eta}\sigma)=\Im(\sigma )
  \Re(\eta )-\Im(\eta ) \Re(\sigma)$.}:
\ben
\Ad_0(U)=\left[\begin{array}{ccc} 
|\eta|^2+|\sigma|^2 &  2\Re(\bar{\eta} \sigma) & 2\Im(\bar{\eta}\sigma)\\
2\Re(\eta\sigma) & \Re(\eta^2+\sigma^2)& \Im (\sigma^2-\eta^2)\\
2\Im(\eta\sigma) & \Im(\eta^2+\sigma^2) & \Re(\eta^2-\sigma^2)
\end{array}\right]
\een
while in the angle-boost parameterization one has: 
{\scriptsize
\ben
\label{LorentzMatrix}
\Ad_0(U\!)=\!\left[\!
\begin{array}{ccc}
 \cosh (t) & ~\cos (a-b) \sinh (t)~ & -\sin (a-b) \sinh (t) \\

 \cos (a\!+\!b)\! \sinh (t) & ~\cos (2 a) \cosh ^2\!\left(\frac{t}{2}\right)+\cos (2 b) \sinh

   ^2\!\left(\frac{t}{2}\right)~ & \sin (2b) \sinh^2\!\left(\frac{t}{2}\right)-\sin (2a) \cosh^2\!\left(\frac{t}{2}\right) \\

 \sin (a\!+\!b)\! \sinh (t) & ~\sin (2a) \cosh^2 \!\left(\frac{t}{2}\right)+\sin (2b) \sinh^2\!\left(\frac{t}{2}\right)~

   & \cos (2 a) \cosh^2\!\left(\frac{t}{2}\right)-\cos (2 b) \sinh^2\!\left(\frac{t}{2}\right) \\
\end{array}
\right]\,.
\een
}
\noindent In particular, we have:
{\scriptsize \be
\Ad_0(R(\theta))=\left[
\begin{array}{ccc}
 1 & 0 & 0 \\
 0 & \cos (\theta ) & -\sin (\theta ) \\
 0 & \sin (\theta ) & \cos (\theta ) \\
\end{array}
\right]~~,~~
\Ad_0(T(t))=\left[
\begin{array}{ccc}
 \cosh (t) & \sinh (t) & 0 \\
 \sinh (t) & \cosh (t) & 0 \\
 0 & 0 & 1 \\
\end{array}
\right]~~,~~\Ad_0(P(\kappa))= \left[
\begin{array}{ccc}
 2 \kappa ^2+1 & 2 \kappa  & -2 \kappa ^2 \\
 2 \kappa  & 1 & -2 \kappa  \\
 2 \kappa ^2 & 2 \kappa  & 1-2 \kappa ^2 \\
\end{array}
\right]~~.
\ee}
Thus $R(\theta)$ acts as a counterclockwise rotation by $\theta$ in
the spacelike $(X^1,X^2)$ plane which fixes the time axis,
while $T(t)$ acts as a boost transformation\footnote{The Lorentz
  factor of this boost is $\Upsilon=\cosh(t)\geq 1$ while its speed in
  units where the speed of light equals one is
  $v=\sqrt{1-\frac{1}{\Upsilon^2}}=\tanh(t)$. We have
  $\Upsilon=\frac{1}{\sqrt{1-v^2}}$.} in the direction $X^1$ which fixes
the spacelike coordinate $X^2$. Notice that $R(\theta)$, $T(t)$ and
$P(\kappa)$ fix respectively the three-vectors $(1,0,0)$, $(0,0,1)$
and $(1,0,1)$.

\paragraph{The hyperboloid model and the Weierstrass map.}

\noindent Consider the future sheet of the unit hyperboloid $(X,X)=1$: 
\be
S^+\eqdef \{X\in \R^3\,|\, (X,X)=1\,~\&~ X^0>0\}=\{X\in \R^3\,|\, X^0=\sqrt{1+(X^1)^2+(X^2)^2} \}~~.
\ee
Let $Z\eqdef X^1+\i X^2$, so the condition $(X,X)=1$ amounts to $(X^0)^2=1+|Z|^2$. 
For any $X\in S^+$, define:
\ben
\label{udef}
u\eqdef \frac{Z}{X^0+1}=\frac{X^1+\i X^2}{X^0+1}\in \rD~~.
\een
Then the condition $(X^0)^2=1+|Z|^2$ amounts to: 
\be
|u|^2=\frac{X^0-1}{X^0+1}~~\Longleftrightarrow~ X^0=\frac{1+|u|^2}{1-|u|^2}~~.
\ee
This implies $1-|u|^2=\frac{2}{X^0+1}$, whereby \eqref{udef} gives: 
\be
Z=\frac{2 u}{1-|u|^2}~~.
\ee
Thus $S^+$ is diffeomorphic with $\rD$ through the {\em Weierstrass map}
$\Xi:\rD\rightarrow S^+$, which is given by:
\ben
\label{Xi}
\Xi(u)\eqdef \left(\frac{1+|u|^2}{1-|u|^2}, \frac{2 \Re u}{1-|u|^2}, 
\frac{2 \Im u}{1-|u|^2}\right)~~.
\een
The components $\Xi^0(u),\Xi^1(u),\Xi^2(u)$ (which are not independent
but satisfy the relation $[\Xi^0(u)]^2=1+[\Xi^1(u)]^2+[\Xi_2(u)]^2$)
are the classical {\em Weierstrass coordinates} of $u\in \rD$. The
Weierstrass map can be viewed as a projection of the disk $\rD$ from the
point $x=(-1,0,0)$ onto $S^+$. Direct computation shows that the 
Weierstrass map has the equivariance property:
\ben
\label{HypAction}
\Xi(\psi_U(u))=\Ad_0(U)(\Xi(u))~,~~\forall u\in \rD~,~~\forall U\in \SU(1,1)~~,
\een
where $\psi$ is the fractional action \eqref{IsoActionD} of
$\SU(1,1)$ on $\rD$. 

\section{Solution of the Hesse equation for rotationally-invariant surfaces}
\label{app:SolHesse}

\noindent For a rotationally-invariant surface $(\Sigma,\cG)$ with
$\Sigma\in \{\C,\dot{\rD}\}$, the Hesse equation \eqref{HessCond} is
equivalent with the first three equations of the system
\eqref{LambdaSysElem}:
\beqan
\label{HessElem}
&& \pd_r^2 \Lambda=\frac{3}{8}\Lambda\nn\\
&& \pd_r\pd_\theta\Lambda-\frac{f'}{2f}\pd_\theta \Lambda=0\\
&& \pd_\theta^2 \Lambda+\frac{f'}{2}\pd_r\Lambda=\frac{3}{8}f \Lambda~~.\nn
\eeqan
It is convenient to define $\beta\eqdef \sqrt{\frac{3}{8}}$ (see
equation \eqref{betadef}). Recall that a rotationally-invariant metric
$\cG$ on $\dot{\rD}$ has the form \eqref{ssg}:
\be
\dd s_{\cG}^2=\dd r^2+f(r)\dd\theta^2~~,
\ee
where $f:\R_{>0}\rightarrow \R_{>0}$ is a smooth positive function. 

\subsection{Solving the Hesse equation}

\begin{prop}
Assume that $\cG$ is a rotationally-invariant metric on $\Sigma\in
\{\C,\dot{\rD}\}$.  Then the Hesse equation \eqref{HessCond} for
$(\Sigma,\cG)$ (which is equivalent with the system \eqref{HessElem})
admits non-trivial solutions iff the Gaussian curvature $K$ of $\cG$
satisfies:
\be
K=-\beta^2=-\frac{3}{8}~~.
\ee
In this case, the radial function $f$ has the form:
\ben
\label{f}
f(r)=\left[A_1 \cosh(\beta r)+A_2 \sinh(\beta r)\right]^2~~,
\een
where $A_1$ and $A_2$ are real constants which are not both zero. 
When $f$ is given by \eqref{f}, the general solution of 
\eqref{HessElem} is as follows:
\begin{enumerate}[1.]
\itemsep 0.0em
\item If $|A_1|<|A_2|$, the general solution is:
\beqan
\label{LambdaElem}
\Lambda(r,\theta)&=&\hat B_1\!\cosh(\beta r)+\hat B_2\sinh(\beta r)+\nn\\
&&+
\zeta \cos\!\Big(\!\beta \sqrt{\!A_2^2\!-\!A_1^2}
(\theta-\theta_0)\!\Big)\left[\!\,A_1\cosh(\beta r)+\!A_2\sinh(\beta r)\right]~
\eeqan
where $\theta\in \R/(2\pi \Z)$, $\zeta\geq 0$ and $A_1,A_2,
\hat B_1,\hat B_2$ are constants subject to the condition:
\ben
\label{ABcond}
A_1\hat B_1=A_2\hat B_2~~,
\een
which implies: 
\ben
\label{Bcond1}
|\hat B_1|>|\hat B_2|
\een
In this case, the constants $\hat B_1$ and $\hat B_2$ can be written as: 
\ben
\label{Bprime}
\hat B_1\eqdef A_1\frac{A_1 B_1-A_2 B_2}{A_2^2-A_1^2}+B_1~~,
~~\hat B_2\eqdef A_2\frac{A_1 B_1-A_2 B_2}{A_2^2-A_1^2}+B_2~~,
\een
where $B_1$ and $B_2$ are arbitrary constants. 
\item If $|A_1|>|A_2|$, the general solution is:
\ben
\label{LambdaSpec1}
\Lambda(r,\theta)=\hat B_1\cosh(\beta r)+\hat B_2\sinh(\beta r)~~,
\een
where $\hat B_1$ and $\hat B_2$ are given by \eqref{Bprime} and hence satisfy
\eqref{ABcond}, which implies:
\ben
\label{Bcond2}
|\hat B_1|<|\hat B_2|~~
\een
\item If $A_1=\epsilon A_2$ with $\epsilon\in \{-1,1\}$, then the
  general solution is:
\ben
\label{LambdaSpec2}
\Lambda(r,\theta)=\hat B \left[\cosh(\beta r)+\epsilon  \sinh(\beta r)\right]
=\twopartdef{\hat B e^{\beta r}}{\epsilon=+1}{\hat B e^{-\beta r}}{\epsilon=-1}~~,
\een
where $B'$ is an arbitrary constant. 
\end{enumerate}
\end{prop}

\begin{proof}
The first equation in \eqref{LambdaSysElem} gives:
\ben
\label{Lambda}
\Lambda(r,\theta)=\Theta_1(\theta)\cosh(\beta r)+\Theta_2(\theta)\sinh(\beta r)
\een
where $\Theta_1,\Theta_2\in \cC^\infty(\rS^1,\R)$ are arbitrary
real-valued smooth functions defined on the circle,
i.e. $2\pi$-periodic smooth functions of $\theta$.

The second equation in \eqref{LambdaSysElem} can be written as:
\be
\pd_r \log \frac{\pd_\theta \Lambda}{\sqrt{f}}=0~~
\ee
and hence gives:
\be
(\pd_\theta\Lambda)(r,\theta)=C(\theta) \sqrt{f(r)}
\ee
for some smooth function $C\in \cC^\infty(\rS^1,\R)$. Using
\eqref{Lambda}, this relation becomes:
\ben
\label{cond1}
\sqrt{f(r)}=A_1(\theta)\cosh(\beta r)+A_2(\theta)\sinh(\beta r)~~,
\een
where $A_i(\theta)\eqdef \frac{\Theta_i'(\theta)}{C(\theta)}$ for
$i=1,2$.  Since $\cosh(\beta r)$ and $\sinh(\beta r)$ are functionally
independent on the interval $(0,+\infty)$ (they have no$n$-vanishing
Wronskian determinant), condition \eqref{cond1} requires that
$A_1,A_2$ are independent of $\theta$ and hence that they are
constant. Indeed, differentiating \eqref{cond1} with respect to $r$
gives:
\ben
\label{cond2}
\frac{f'(r)}{2\beta \sqrt{f(r)}}=A_1(\theta)\sinh(\beta r)+A_2(\theta)\cosh(\beta r)~~.
\een
Relations \eqref{cond1}, \eqref{cond2} can be viewed as a system of
two linear equations for $A_1(\theta)$ and $A_2(\theta)$, whose
discriminant equals:
\be
W=\det \left [ \begin{array}{cc} \cosh(\beta r) & \sinh(\beta r) \\
 \sinh(\beta r) & \cosh(\beta r)\end{array} \right]=\cosh^2(\beta r)-\sinh^2(\beta r)=1~~.
\ee
Hence the unique solution of this system is:
\beqa
&& A_1(\theta)=\sqrt{f(r)}\cosh(\beta r)-\frac{f'(r)}{2\beta \sqrt{f(r)}}\sinh(\beta r)\nn\\
&& A_2(\theta)=\frac{f'(r)}{2\beta \sqrt{f(r)}}\cosh(\beta r)-\sqrt{f(r)} \sinh(\beta r)~~.
\eeqa
Since the right hand side of these equations depends only on $r$ while
the left hand side depends only on $\theta$, we conclude that $A_1$
and $A_2$ must be constant. Thus:
\ben
\label{Theta}
\Theta_1(\theta)=A_1 D(\theta)+B_1~~,~~\Theta_2(\theta)=A_2 D(\theta)+B_2~~,
\een
where $B_1,B_2$ are constants and $D\in \cC^\infty(\rS^1,\R)$ is any
fixed primitive of $C(\theta)$. Moreover, condition \eqref{cond1}
becomes \eqref{f}. Notice that $A_1$ and $A_2$ cannot both vanish since $f(r)$ is
strictly positive for all $r>0$. Thus:
\be
A_1^2+A_2^2>0~~.
\ee
Using \eqref{f} in equation \eqref{K_G} shows that the Gaussian
curvature is fixed to the value:
\ben
\label{K0}
K_\cG=-\beta^2=-\frac{3}{8}~~.
\een
Combining \eqref{Lambda} and \eqref{Theta} gives:
\ben
\label{LambdaD}
\Lambda(r,\theta)=[A_1 D(\theta)+B_1]\cosh(\beta r)+[A_2 D(\theta)+B_2]\sinh(\beta r)~~.
\een

\noindent It remains to analyze the third equation in
\eqref{LambdaSysElem}. Substituting \eqref{f} and \eqref{LambdaD} in
that equation gives:
\be
\left[A_1 \cosh(\beta r)+A_2 \sinh(\beta r)\right] \left[
  D''(\theta )+\beta^2 (A_2^2-A_1^2) D(\theta) -\beta^2 (A_1 B_1-  A_2 B_2)\right]=0~~.
\ee
Since $A_1^2+A_2^2\neq 0$, this is equivalent with the following second
order linear inhomogeneous ODE with constant coefficients:
\ben
\label{Deq}
D''(\theta)+\beta^2 (A_2^2-A_1^2) D(\theta)=\beta^2 (A_1B_1-A_2 B_2)~~.
\een
Differentiating this with respect to $\theta$ shows that $D'$
satisfies the corresponding homogeneous second order ODE:
\ben
\label{Deq0}
D'''(\theta)+\beta^2 (A_2^2-A_1^2) D'(\theta)=0~~.
\een
Since $D(\theta)$ is periodic, the same is true of $D'(\theta)$. Hence
$D'$ must be constant or the characteristic equation of \eqref{Deq0}
(viewed as a second order ODE for $D'$) must have no$n$-vanishing
imaginary roots. We therefore distinguish the disjoint cases:
\begin{enumerate}[(a)]
\itemsep 0.0em
\item The characteristic equation of $\eqref{Deq0}$ has no$n$-zero
imaginary roots, i.e. we have:
\ben
\label{Acond1}
|A_1|<|A_2|~~.
\een
This implies $A_2\neq 0$ since $A_1$ and $A_2$ are not both zero. In
this case, setting $D(\theta)=E(\theta)+\frac{A_1 B_1-A_2
B_2}{A_2^2-A_1^2}$ shows that \eqref{Deq0} is equivalent with the
homogeneous equation:
\be
E''(\theta)+\beta^2 (A_2^2-A_1^2) E(\theta)=0~~,
\ee
whose general solution is:
\be
E(\theta)=\zeta_1\cos\left(\beta \sqrt{A_2^2-A_1^2}\, \theta\right)\!+
\zeta_2\sin\left(\beta \sqrt{A_2^2-A_1^2}\, \theta\right)~~(\zeta_1,\zeta_2=\const)~.
\ee
Hence the general solution of \eqref{Deq} takes the form:
\ben
\label{Dsol}
D(\theta)\!=\!\frac{A_1 B_1\!-\!A_2 B_2}{A_2^2-A_1^2}
\!+\!\zeta_1\cos\left(\!\beta \sqrt{A_2^2\!-\!A_1^2}\, \theta\!\right)
\!+\!\zeta_2\sin\left(\!\beta \sqrt{A_2^2\!-\!A_1^2}\, \theta\!\right) .
\een
Let $\zeta\eqdef \sqrt{\zeta_1^2+\zeta_2^2}\geq 0$ and $\theta_0\in
[0,2\pi)$ be defined through:
\be
\zeta_1=\zeta\cos\left(\beta \sqrt{A_2^2-A_1^2}\,\theta_0\right)~~,
~~\zeta_2=-\zeta\sin\left(\beta \sqrt{A_2^2-A_1^2}\,\theta_0\right)~~,
\ee
where we take $\theta_0\eqdef 0$ when $\zeta=0$. Then \eqref{Dsol} takes the form:
\be
D(\theta)=\frac{A_1 B_1-A_2 B_2}{A_2^2-A_1^2}+\zeta\cos\left[\beta \sqrt{A_2^2-A_1^2}
 (\theta-\theta_0)\right]~~.
\ee
Performing a shift $\theta\rightarrow \theta-\theta_0$ of the angular
coordinate $\theta$, we can assume without loss of generality that
$\theta_0=0$, which we shall do from now on. Then \eqref{Dsol}
becomes:
\ben
\label{Dsol0}
D(\theta)=\frac{A_1 B_1-A_2 B_2}{A_2^2-A_1^2}+\zeta\cos\left[\beta \sqrt{A_2^2-A_1^2}\, 
\theta\right]~~,
\een
while \eqref{LambdaD} becomes \eqref{LambdaElem},
where $\hat B_1$ and $\hat B_2$ are given by \eqref{Bprime} and hence are
subject to condition \eqref{ABcond}. This condition implies \eqref{Bcond1}
upon using the relations $|A_1|<|A_2|\neq 0$.
\item $D'$ is constant and the characteristic equation of
  $\eqref{Deq0}$ has real roots. In this case, we have:
\ben
\label{Acond2}
|A_1|\geq |A_2|~~.
\een
Since $A_1$ and $A_2$ are not both zero, relation \eqref{Acond2}
implies $A_1\neq 0$.  \eqref{Deq} and \eqref{Deq0} become:
\ben
\label{D0}
(A_2^2-A_1^2) D = A_1B_1-A_2 B_2~~,~~(A_2^2-A_1^2) D'=0~~.
\een
Distinguish the sub-cases:
\begin{enumerate}[label=(b.\arabic*)]
\itemsep 0.0em
\item $|A_1|>|A_2|$. Then the second equation in \eqref{D0} gives
$D'=0$ (i.e. $D$ is constant), while the first equation gives:
\be
D=\frac{A_1B_1-A_2B_2}{A_2^2-A_1^2}~~.
\ee
In this case, relation \eqref{LambdaD} gives \eqref{LambdaSpec1}, 
where $\hat B_1$ and $\hat B_2$ are defined through \eqref{Bprime} and hence
satisfy \eqref{ABcond}, which implies:
\ben
\label{Bcond2prime}
|\hat B_1|<|\hat B_2|~~
\een
upon using the relations $|A_2|<|A_1|\neq 0$. Notice that
\eqref{LambdaSpec1} can be obtained by formally setting $A_1=A_2=0$ in
\eqref{LambdaElem} (in which case \eqref{ABcond} is automatically
satisfied) and replacing condition \eqref{Bcond1} with \eqref{Bcond2}.
\item $|A_1|=|A_2|$. Since $A_1^2+A_2^2\neq 0$, we must then have:
\be
A_2=\epsilon A_1\neq 0~~,
\ee
where $\epsilon\in\{-1,1\}$. On the other hand, the first equation in
\eqref{D0} gives:
\ben
\label{Brel}
B_2=\epsilon B_1~~.
\een
In this case, $D$ is an arbitrary constant and \eqref{LambdaD} gives
\eqref{LambdaSpec2}, where:
\be
\hat B\eqdef A_1 D +B_1~~,
\ee
so $\hat B$ is an arbitrary constant. Notice that \eqref{LambdaSpec2} can
be obtained by formally setting $A_1=A_2=0$ in \eqref{LambdaElem} (in
which case \eqref{ABcond} is automatically satisfied) and replacing
condition \eqref{Bcond1} with:
\ben
\label{Bcond3}
|\hat B_1|=|\hat B_2|~~.
\een
\end{enumerate}
Both sub-cases (b.1) and (b.2) of case (b) can be recovered by
formally setting $A_1=A_2=0$ in \eqref{LambdaElem} (in which case
\eqref{ABcond} is automatically satisfied) and replacing condition
\eqref{Bcond1} with:
\ben
\label{Bcond4}
|\hat B_1|\leq |\hat B_2|~~.
\een
\end{enumerate}
\end{proof}

\begin{remark}
Notice that \eqref{Deq} can also be obtained by setting $r=0$ in the
second equation of \eqref{LambdaSysElem}. Using the relations:
\beqa
&& \Lambda(0,\theta)=A_1 D(\theta)+B_1~~,~~(\pd_r \Lambda)(0,\theta)
=\beta [A_2 D(\theta)+B_2]\nn\\
&& f(0)=A_1^2~~,~~f'(0)=2\beta A_1 A_2~~,
\eeqa
this gives:
\be
A_1[D''(\theta)+\beta^2 (A_2^2-A_1^2) D(\theta)-\beta^2 (A_1B_1-A_2 B_2)]=0~~.
\ee
\end{remark}

\subsection{Reduction to standard cases}

\noindent Condition \eqref{K0} shows that the rescaled metric $G$ defined
through:
\be
\cG=G/\beta^2~~\mathrm{i.e.}~~G=\beta^2\cG
\ee
satisfies $K_G=-1$ and hence it is a hyperbolic metric
on $\Sigma\in \{\C,\dot{\rD}\}$. We have:
\be
\dd s_G^2=\beta^2\dd s_\cG^2=\dd \br^2+f_\beta(\br)\,\dd\theta^2~~,
\ee
where:
\ben
\label{brr}
\br\eqdef \beta r
\een
is the hyperbolic normal radial coordinate (i.e. the normal radial coordinate 
on $(\Sigma,G)$) and:
\be
f_\beta(\br)\eqdef \beta^2 f\left(\frac{\br}{\beta}\right)~~\mathrm{i.e.}~~f(r)
=\frac{1}{\beta^2}f_\beta(\beta r)~~.
\ee
Since we require $\cG$ (hence also $G$) to be complete, well-known
results from the theory of hyperbolic surfaces (see Appendix
\ref{app:elem}) imply that $(\Sigma,G)$ is isometric with either of
the following :
\begin{itemize}
\itemsep 0.0em
\item The hyperbolic disk $\mD$
\item The hyperbolic punctured disk $\mD^\ast$
\item A hyperbolic annulus $\mA(R)$ of modulus $\mu=2\log R$, where $R>0$.
\end{itemize}
Hence we can always reparameterize the radial coordinate $r$ so as to
bring $\dd s^2_G$ to one of the following four forms:
\begin{enumerate}
\itemsep 0.0em
\item The {\bf hyperbolic disk}:
\ben
\label{metricP}
\dd s^2_G=\dd \br^2+\sinh^2(\br) \,\dd \theta^2
\een
In this case, we have: $f_\beta(\br)=\sinh^2(\br)$ and:
\ben
\label{fP}
f(r)=\frac{1}{\beta^2} \sinh^2(\beta r)~~\mathrm{i.e.}~A_1=0~\&~A_2=1/\beta~~.
\een
Since $|A_1|<|A_2|$, the general solution has the form \eqref{LambdaElem}:
\ben
\label{LambdaP}
\Lambda(r,\theta)=\hat B_1\cosh(\beta r)+\frac{\zeta}{\beta}\sinh(\beta r)
\cos (\theta-\theta_0)~~,
\een
where we noticed that condition \eqref{ABcond} requires $\hat B_2=0$.
Here $\hat B_1$ is an arbitrary constant while $\zeta\geq 0$. One can
also write $\Lambda$ in Euclidean polar coordinates $(\rho,\theta)$ on
$\rD$, which are related to the normal polar coordinates
$(r,\theta)$ of $\cG$ through (cf. eqs. \eqref{brD} and \eqref{brr}):
\ben
\label{rrhoD}
\rho=\tanh(\beta r/2)\in [0,1)\Longleftrightarrow r=\frac{2}{\beta}\arctanh(\rho)=
\frac{1}{\beta}\log \frac{1+\rho}{1-\rho}\in [0,+\infty)~~.
\een
Substituting this in \eqref{LambdaP} gives:
\ben
\label{LambdaDprime}
\Lambda(\rho,\theta)=\frac{B_0(1+\rho^2)+(2\zeta/\beta) \rho
\cos (\theta-\theta_0)}{1-\rho^2}~~.
\een
Defining:
\ben
\label{B12def}
B_1=\frac{\zeta}{\beta} \cos\theta_0~~,~~B_2=\frac{\zeta}{\beta} \sin\theta_0~~,
\een
this can also be written as follows in Euclidean Cartesian coordinates
$x=\rho\cos\theta$ and $y=\rho\sin\theta$ on $\rD$:
\ben
\label{LambdaDCart}
\Lambda(x,y)=\frac{B_0(1+\rho^2)+2B_1 x +2 B_2 y}{1-\rho^2}~~, 
\een
where $\rho=\sqrt{x^2+y^2}$.

\item The {\bf hyperbolic punctured disk}:
\ben
\label{metricpD}
\dd s^2_G=\dd \br^2+\frac{1}{(2\pi)^2} e^{-2\br} \dd \theta^2
\een
In this case, we have $f_\beta(\br)=\frac{1}{(2\pi)^2}e^{-2\br}$ and:
\ben
\label{fpD}
f(r)=\frac{1}{(2\pi\beta)^2}e^{-2\beta r}~~\mathrm{i.e.}~A_1=-A_2=\frac{1}{2\pi\beta}~~.
\een
Since $A_2=-A_1$, the general solution \eqref{LambdaElem} has the form 
\eqref{LambdaSpec2} with $\epsilon=-1$:
\ben
\label{LambdapD}
\Lambda(r,\theta)={\hat B} e^{-\beta r}~~,
\een
where $\hat B$ is an arbitrary constant. One can also write $\Lambda$
in Euclidean polar coordinates $(\rho,\theta)$, which are related to
the normal polar coordinates $(r,\theta)$ of $\cG$ through
(cf. eqs. \eqref{brDast} and \eqref{brr}):
\ben
\label{rrhoDast}
\rho=e^{-2\pi e^{\beta r}}\in (0,1)\Longleftrightarrow r=\frac{1}{\beta} 
\log\left(\frac{|\log \rho|}{2\pi}\right)\in (-\infty,\infty)~~.
\een
Substituting this in \eqref{LambdapD} gives:
\ben
\label{LambdapDCart}
\Lambda(\rho,\theta)=\frac{2\pi {\hat B}}{|\log \rho|}~~.
\een

\item A {\bf hyperbolic annulus}:
\ben
\label{metricA}  
\dd s^2_G=\dd \br^2+\frac{\ell^2}{(2\pi)^2} \cosh^2(\br) \,\dd \theta^2~~,
\een
where $\ell>0$ is given by \eqref{ell}. In this case, we have
$f_\beta(\br)=\frac{\ell^2}{(2\pi)^2}\cosh^2(\br)$ and:
\ben
\label{fA}
f(r)=\frac{\ell^2}{(2\pi\beta)^2}\cosh^2(\beta r)~~\mathrm{i.e.}~A_1=\frac{\ell}{2\pi\beta}
~\&~A_2=0~~.
\een
Since $|A_1|>|A_2|$, the general solution has the form \eqref{LambdaSpec1}:
\ben
\label{LambdaA}
\Lambda(r,\theta)={\hat B_2}\sinh(\beta r)~~,
\een
where we noticed that relation \eqref{ABcond} gives $\hat B_1=0$.
Here $\hat B_2$ is an arbitrary constant. Relation \eqref{Bprime}
gives $\hat B_2= B_2$. One can also write $\Lambda$ in Euclidean polar
coordinates $(\rho,\theta)$, which are related to the normal polar
coordinates $(r,\theta)$ of $\cG$ through (cf. eqs. \eqref{brA} and
\eqref{brr}):
\ben
\label{rrhoA}
\rho= e^{-\frac{\mu}{\pi}\arccos\left(\frac{1}{\cosh(\beta r)}\right)}
 \Longleftrightarrow r=\frac{1}{\beta}\arccosh\left[\frac{1}{\cos\left(\frac{\pi}{\mu}|
\log\rho|\right)}\right]\in
(-\infty,+\infty)
\een
Substituting this in \eqref{LambdaA} gives:
\ben
\label{LabdaACart}
\Lambda(r,\theta)={\hat B_2} \tan\left(\frac{\pi}{\mu}\log\rho\right)~~.
\een
\end{enumerate}

\section{Elementary hyperbolic surfaces} 
\label{app:elem}
\setcounter{equation}{0}

\noindent A complete and connected hyperbolic surface $(\Sigma,G)$ is called
{\em elementary} if it coincides with either of the hyperbolic disk $\mD$, the
hyperbolic punctured disk $\mD^\ast$ or one of the hyperbolic annuli
$\mA(R)$. Here $R>1$ and $\mA(R)$ denotes the hyperbolic annulus of
modulus $\mu=2\log R$. These three surfaces are defined as the
following open subsets of the complex plane with complex coordinate
$u$:
\begin{itemize}
\item $\rD\eqdef \{u\in \C \,\, |\,\, 0\leq |u|<
  1\}$
\item $\dot{\rD}\eqdef \rD\setminus \{0\}$
\item $A(R)\eqdef\{u\in \C\,\, |\,\, \frac{1}{R}<|u|<
  R\}$, where $R>1$~~,
\end{itemize}
endowed with the following hyperbolic metrics:
\beqan
&&\dd s_\mD^2=\frac{4}{(1-\rho^2)^2}\left(\dd \rho^2+\rho^2\dd \theta^2\right)\nn\\
&&\dd s_{\mD^\ast}^2=\frac{1}{(\rho\log\rho)^2}(\dd \rho^2+\rho^2 \dd \theta^2)\\
&&\dd s_{\mA}^2=\left(\frac{\pi}{2\log R}\right)^2\frac{\dd\rho^2+\rho^2
\dd \theta^2}{\left[\rho \cos\left(\frac{\pi\log\rho}{2\log R}\right)\right]^2}~~,\nn
\eeqan
where $\rho\eqdef |u|$ and $\theta=\arg(u)$ are polar coordinates in
the complex plane.

Elementary hyperbolic surfaces admit special semi-geodesic coordinates
$(\br,\theta)$ in which the hyperbolic metric takes the form:
\ben
\label{sg}
\dd s_G^2=\dd \br^2+\rf(\br)\dd\theta^2~~,
\een
where $\rf$ is a smooth real-valued function which is
strictly positive everywhere and satisfies the condition:
\ben
\label{rfcond}
\frac{\pd^2 \sqrt{\rf}}{\pd r^2}=\sqrt{\rf}~~.
\een
More precisely, one has (see \cite{elem}):
\begin{itemize}
\item hyperbolic disk $\mD$: 
\ben
\label{pmetric}
\dd s_\mD^2=\dd \br^2+\sinh^2(\br)\dd\theta^2~~,
\een
where:
\ben
\label{brD}
\br=2\arctanh(\rho)=\log \frac{1+\rho}{1-\rho}\in (0,+\infty)~~.
\een
\item hyperbolic punctured disk $\mD^\ast$: 
\ben
\label{cmetric2}
\dd s_{\mD^\ast}^2=\dd \br^2+\frac{e^{-2\br}}{(2\pi)^2}\dd\theta^2~~,
\een
where:
\ben
\label{brDast}
\br=\log\left(\frac{|\log\rho|}{2\pi}\right)\in (-\infty,+\infty)~~.
\een
\item hyperbolic annulus $\mA(R)$:
\ben
\label{fmetric1}
\dd s_{\mA}^2=\dd \br^2+\frac{\ell^2}{(2\pi)^2}\cosh^2(\br)\dd\theta^2~~,
\een
where:
\ben
\label{brA}
\br=\sign(\br)\arccosh\left[\frac{1}{\cos\left(\frac{\pi}{\mu}|\log\rho|\right)}\right]\in
(-\infty,+\infty)
\een
and the positive quantity $\ell$ is given by:
\ben
\label{ell}
\ell=\frac{\pi^2}{\log R}~~.
\een
\end{itemize}

\end{document}